\documentclass[noshowframe]{imsart}

\RequirePackage{amsthm,amsmath,amsfonts,amssymb}
\RequirePackage[numbers]{natbib}
\usepackage{enumitem}
\usepackage{algorithm}
\usepackage{algpseudocode}
\RequirePackage[colorlinks,citecolor=blue,urlcolor=blue]{hyperref}
\RequirePackage{graphicx}

\startlocaldefs
\theoremstyle{plain}

\newtheorem{theorem}{Theorem}[section]

\newtheorem{corollary}{Corollary}

\theoremstyle{definition}

\newtheorem{assumption}{Assumption}
\theoremstyle{remark}

\newcommand\independent{\protect\mathpalette{\protect\independenT}{\perp}}
\def\independenT#1#2{\mathrel{\rlap{$#1#2$}\mkern2mu{#1#2}}}
\DeclareMathOperator*{\argmax}{argmax}

\def\bX{\mathbf{X}}
 \newcommand{\main}{\frac{A(Z+A)Yw_{\alpha_Z}(A,Z, X, Y)I\{\pi(X) = Z\}}{2\gamma(A,Z,X)f(Z|X)f(A| Z, X)}}
 \newcommand{\deri}{\frac{\partial }{\partial t}}
  \newcommand{\hz}{A(Z+A)}

\newcommand{\q}{Q(A,Z,X )}
 \newcommand{\qp}{Q(1,1,X)}
 \newcommand{\qn}{Q(-1,-1,X)}
 \newcommand{\qhat}{\hat{Q}(A,Z,X)}
  \newcommand{\qhatp}{\hat{Q}(1,1,X)}
  \newcommand{\qhatn}{\hat{Q}(-1,-1,X)}
 \newcommand{\g}{\gamma(A,Z,X)}
 \newcommand{\gp}{\gamma(1,1,X)}
 \newcommand{\gn}{\gamma(-1,-1,X)}
 \newcommand{\ghat}{\hat{\gamma}(A,Z,X)}
  \newcommand{\ghatp}{\hat{\gamma}(1,1,X)}
  \newcommand{\ghatn}{\hat{\gamma}(-1,-1,X)}
 \newcommand{\del}{\delta(A,Z,X)}
 \newcommand{\delhat}{\hat{\delta}(A,Z,X)}
 \newcommand{\I}{I\{\pi(X) = Z\}}
 \newcommand{\w}{w_{\alpha_Z}(A,Z, X, Y)}
 \newcommand{\f}{f(Z|X)}
 \newcommand{\faz}{f(A,Z|X)}
 \newcommand{\p}{f(A|Z,X)}
 \newcommand{\fhat}{\hat{f}(A,Z|X)}
 
 \newcommand{\hzz}{ZA(Z+A)} 
 \newcommand{\hf}{\hat{f}(Z|X)}
 \newcommand{\hp}{\hat{f}(A|Z,X)}
  \newcommand{\hg}{\hat{\gamma}(A,Z,X)}
  
  \newcommand{\hkappa}{\hat{\kappa}(Z,X)}
 \newcommand{\hkappap}{\hat{\kappa}'(X)}
 \newcommand{\hdelta}{\delta(A,Z,X)}
 \newcommand{\Q}{Q(A,Z,X)}

\endlocaldefs

\begin{document}
\begin{frontmatter}
\title{Sensitivity analysis for constructing optimal  regimes in the presence of treatment non-compliance}
%\title{A sample article title with some additional note\thanksref{t1}}
\runtitle{Optimal regimes in the presence of non-compliance}
%\thankstext{T1}{A sample additional note to the title.}

\begin{aug}
%%%%%%%%%%%%%%%%%%%%%%%%%%%%%%%%%%%%%%%%%%%%%%%
%% Only one address is permitted per author. %%
%% Only division, organization and e-mail is %%
%% included in the address.                  %%
%% Additional information can be included in %%
\author[A]{\fnms{Cuong T.}~\snm{Pham}\ead[label=e1]{Cuong\_Pham@URMC.Rochester.edu}},
\author[B]{\fnms{Kevin G.}~\snm{Lynch} },
\author[C]{\fnms{James R.}~\snm{McKay} },
\author[A]{\fnms{Ashkan}~\snm{Ertefaie} }

%%%%%%%%%%%%%%%%%%%%%%%%%%%%%%%%%%%%%%%%%%%%%%
%% Addresses                                %%
%%%%%%%%%%%%%%%%%%%%%%%%%%%%%%%%%%%%%%%%%%%%%%
\address[A]{Department of Biostatistics and Computational Biology,
  University of Rochester Medical Center\printead[presep={, }]{e1}}
\address[B]{Center for Clinical Epidemiology and Biostatistics and Department of Psychiatry, University of Pennsylvania}
\address[C]{Department of Psychiatry, Perelman School of Medicine, University of Pennsylvania}

\runauthor{Pham et al.}
\end{aug}

\begin{abstract}
 The current body of research on developing optimal treatment strategies often places emphasis on intention-to-treat analyses, which fail to take into account the compliance behavior of individuals.  Methods based on instrumental variables have been developed to determine optimal treatment strategies in the presence of endogeneity. However, these existing methods are not applicable when there are two active treatment options and the average causal effects of the treatments cannot be identified using a binary instrument. In order to address this limitation, we present a procedure that can identify an optimal treatment strategy and the corresponding value function as a function of a vector of sensitivity parameters. Additionally, we derive the canonical gradient of the target parameter and propose a multiply robust classification-based estimator for the optimal treatment strategy. Through simulations, we demonstrate the practical need for and usefulness of our proposed method. We apply our method to a randomized trial on Adaptive Treatment for Alcohol and Cocaine Dependence. 
\end{abstract}

\begin{keyword}
\kwd{Canonical gradients}
\kwd{endogeneity} 
\kwd{instrumental variables}
\kwd{value function}
\kwd{weighted classification}
\end{keyword}

\end{frontmatter}
%%%%%%%%%%%%%%%%%%%%%%%%%%%%%%%%%%%%%%%%%%%%%%
%% Please use \tableofcontents for articles %%
%% with 50 pages and more                   %%
%%%%%%%%%%%%%%%%%%%%%%%%%%%%%%%%%%%%%%%%%%%%%%
%\tableofcontents

\section{Introduction}

\label{sec:intro}

An important goal in recent clinical research is to use the available data to understand ``what works best for whom?'' It can be achieved by constructing individualized treatment strategies that map a person's characteristics to a treatment option. An optimal individualized treatment strategy is one that optimizes a specified health outcome. In certain research areas, such as mental health and substance use disorder, the high rate of non-compliance to treatments imposes a substantial challenge in constructing optimal strategies that are generalizable to the general population. 

An optimal strategy can be estimated using policy (strategy) learning or conditional outcome model-based (e.g., Q-learning and A-learning) methods \citep{schulte2014q, chakraborty2013statistical,     ertefaie2021robust}. The main drawback of the outcome model-based approaches is that the quality of the constructed optimal strategy relies on the specified regression models \citep{zhao2009reinforcement, zhao2011reinforcement}. However, policy learning approaches bypass the need for conditional outcome models
by defining the problem of estimating the optimal strategy as a weighted classification problem, thereby directly optimizing the expected outcome among a class of rules \citep{zhao2012estimating, zhao2015new,   zhang2013robust}. The preceding methods rely on the no unmeasured confounder assumption, which is not verifiable using the observed data.
%This assumption is needed even in randomized trials with noncompliance 
 The latter assumption imposes an even stronger limitation in randomized trials with non-compliance. It is because individuals may not comply with prescribed treatment due to many reasons that are not recorded in the data. Hence, when randomized trial data are available, intention-to-treat analyses are often performed to avoid violating the no-unmeasured confounder assumption. However, such analyses have two significant limitations. First, they estimate the effect of randomization on a treatment option and not the actual treatment effect that is of main interest. In fact, the treatment effect estimates are often biased toward the null effect \citep{marasinghe2007noncompliance, lin2008longitudinal}. Second, the concluding results may not be reproducible due to the potential differential compliance behavior in real-world settings \citep{ robins1991correcting, hewitt2006there,lin2008longitudinal}. 
 
Instrumental variable (IV) methods are well-studied for providing unbiased treatment effect estimates
in the presence of unmeasured confounders. An instrument is a random or haphazard encouragement to adopt some change in behavior,
where encouragement can affect outcomes only indirectly through its manipulation of the treatment, and also is independent of unmeasured confounders. The simplest and clearest case is a randomized encouragement \citep{angrist1996identification}.  
Recently, IV-based approaches have been proposed to construct an optimal treatment regimes \cite{cui2021semiparametric,qiu2021optimal}.
. However, these approaches are not applicable to settings with two (or more) active treatment options, which is the case in many comparative effectiveness studies and randomized trials \citep{swanson2015selecting, ertefaie2016sensitivity, ertefaie2016selection}. This is because when there are multiple active treatments, the compliance values will have at least one level more than the (binary) instrument, leading to an identifiability issue \citep{cheng2006bounds}.

 In this paper, we provide a method for estimating an optimal strategy and the corresponding value function among compliers-- individuals who take the assigned treatment--  in the presence of non-compliance and two active treatments. We propose sensitivity analyses that generalize the weighted classification methods by allowing compliance to be endogenous and have more levels than the instrument. To this end, we (1) show that the optimal strategy is identifiable as a function of sensitivity parameters; (2) propose inverse probability weighted and multiply robust sensitivity analysis approaches that identify the optimal regime as a function of sensitivity parameters; and (3) derive a multiply robust estimator for the mean outcome under a treatment strategy. Simulation studies are used to examine the performance of our methods and demonstrate the importance of proper adjustment for non-compliance when estimating optimal strategies. Data from the randomized trial `Adaptive treatment for alcohol and cocaine dependence' are used to illustrate our proposed methods.

  Characterizing an optimal treatment strategy and the value function among compliers is important for the following reasons. First, it helps decision makers to understand the actual causal effect of the strategy, which can motivate them to improve the implementation of the treatment \citep{sommer1991estimating}. Second, it is critical from the patients' perspective as it allows them to weigh the benefits and burdens of the treatment option and motivate them to comply with the treatment \citep{hernan2012beyond, sheiner1995intention}. Moreover, although it is impossible to identify which subjects in the data set are ``compliers,'' one can characterize the compliers in terms of their distribution of observed covariates \citep{brookhart2007preference, angrist2009mostly, baiocchi2014instrumental}. For example, we might see that individuals with a low level of baseline treatment readiness-- a key covariate in substance use disorders-- are underrepresented among compliers. 

\section{Problem Setting} \label{sec:setting}
\subsection{Notation}

Our data consists of $n$ independent, identically distributed trajectories of $\mathcal{O} = (X,Z,A,Y) \sim P_0$.
 The vector $X \in \mathcal{X}$ includes all available baseline covariates
measured before the instrumental variable $ Z \in \mathcal{Z}= \{-1, 1\}$.  Let $A$ be the observed compliance value where $A \in \{-1,0,1\}$.  %so that a decision maker following $\pi$ would recommend $\pi(x)$ to a patient with $X = x$. 
Let $A(z)$ be the potential compliance value given $Z = z$, where $A(z) \in \{-1, 0, 1\}$, e.g.,  $A(1)$ is the level of compliance if the instrument was $1$. In this formulation, $A(z) = 0$ indicates
that no treatment was taken under $Z = z$. The potential outcomes are $Y (z, a)$ for $Z = z$ and $A = a$. %The observed treatment and outcome are $A = A(Z)$ and $Y = Y (Z, A(Z))$, respectively.
Let $\mathcal{D}$ be a class of decision rules. We define a treatment regime as a function $\pi \in \mathcal{D}$ where $\pi$ : $\mathcal{X} \mapsto \mathcal{Z}$. The potential outcome under a regime $\pi$ is defined as
$Y(\pi) = \sum_z\sum_a Y (z, a) I\{\pi(X) = z\}I\{ A = a\}.$ 
For any $\pi$, we define a value function $\mathcal{V}(\pi) = E\{Y(\pi)\}$. An optimal regime $\pi^{opt}$ satisfies $\mathcal{V}(\pi^{opt}) \ge \mathcal{V}(\pi)$ for all $\pi \in \mathcal{D}$. Our goal is to construct an optimal regime among subjects who would comply with their assigned treatment (i.e., compliers) defined as 
$\pi_c^{opt} = \argmax_{\pi \in \mathcal{D}} E\{Y(\pi) \mid A(1)=1,A(-1)=-1\}.$ We also define an $L_2$-norm $\|f\|_{2,\mu}^2 = \int f^2(x) d\mu(x)$  where $\mu$ is an appropriate measure (i.e., the Lebesgue or the counting measure).

\bigskip
\begin{table}[t]
    \centering
    \caption{The list of principal strata}
    \label{tab:PSs}
\begin{tabular}{clcc}
\hline
     & \multicolumn{1}{p{2cm}}{ Principal strata (PS)} 
     & \multicolumn{1}{p{3cm}}{\centering Compliance value $A(-1)$}
     & \multicolumn{1}{p{3cm}}{\centering Compliance value $A(+1)$} \\
     \hline
     & S1 (always -1 taker) & -1 & -1 \\
     & S2 (always +1 taker) & +1 & +1 \\
     & S3 (never taker) & 0  & 0 \\
     & S4 (complier) & -1 & +1 \\
     & S5 & -1 & 0 \\
     & S6 & 0 & +1 \\
     & S7 (defier) & +1 & -1 \\
     & S8 & +1 & 0 \\
     & S9 & 0 & -1 \\
     \hline
\end{tabular}
\end{table}

\subsection{Identification assumptions}

Our identification results rely on the following assumptions.

\begin{assumption} (Fundamental assumptions in IV analyses) \label{assump:coreiv}
 \begin{enumerate}[labelindent=1.3em, labelsep=0.2cm,leftmargin=*, label=(\ref{assump:coreiv}\Alph*)]
\item  If $Z = z$ then $A = A(z)$ and if $Z = z \text{ and } A = a$ then $Y = Y(z,a)$. In other words, the treatment only affects the subject taking the treatment, and there is only one version of IV and treatment. 
\item  $p(A = 1 \mid Z = 1, X = x ) > p(A = 1 \mid Z = -1, X = x).$

\item The IV has no direct effect on the outcome: $Y(Z = -1,A = a) = Y(Z = 1, A = a)$.
\item The instrument $Z$ is independent of the potential outcomes of $Y$ and $A$ given $X$: $Z \independent{ \{A(-1), A(1), Y(-1,-1), Y(-1,0), Y(-1,1), Y(1,-1)}$, \\ $Y(1,0), Y(1,1)\} \mid X.$
\end{enumerate}
\end{assumption}

\begin{assumption} \label{assump:positivity}
$Z | X$ has a
  positive density with respect to a
  dominating measure on $\mathcal{Z}$. We assume
  $\mathrm{Var}(Z|X)$ exists and is between $1/C$ and $C$ for
  some constant $C > 1$ and all $X \in \mathcal{X}$.
\end{assumption}

Assumptions \ref{assump:coreiv}A-\ref{assump:coreiv}D are  standard IV assumptions \citep{angrist1996identification, baiocchi2014instrumental}.
 Assumption \ref{assump:positivity} is the positivity assumption that ensures the possibility of
statistical inference for the treatment effect.

\begin{assumption} (Monotonicity) \label{assump:mono}
The potential outcomes of $A$ satisfy the following properties: (4A) $p\{A(1) = -1, A(-1) = 1\} = 0$ (i.e. there are no defiers); (4B) $p\{A(-1) = 1, A(1) = 0 \} = 0$; (4C) $p\{A(-1) = 0, A(1) = -1\} = 0$. 
\end{assumption}

In our setting, we can classify subjects into the  nine principal strata listed in Table \ref{tab:PSs} \citep{frangakis2002principal}. 
The monotonicity assumptions eliminate the principal strata S7, S8, and S9. They reduce the number of free parameters, thereby enabling the estimation of an optimal treatment strategy among compliers. Importantly, these assumptions are not restrictive in our setting because it is irrational for  subjects to refuse the given treatment but seek it when they are given a different treatment.

\subsection{Preliminaries}

One key assumption in constructing optimal strategies is that of no unmeasured confounders. \citeauthor{cui2021semiparametric} relaxed this assumption by generalizing the weighted classification approach of \citeauthor{zhao2012estimating}, which breaks the endogeneity of the compliance variable using an instrument. In the presence of a binary treatment and a binary compliance value, \citeauthor{cui2021semiparametric} showed that under their assumptions, the  compliers' value function $\mathcal{V}_c(\pi)$ can be identified by
\[
E\left[ \frac{ZAYI\{A = \pi(X)\}}{f(Z| X)\{p(A=1\mid Z=1)-p(A=1\mid Z=-1)\}} \right], 
\]
where $I(.)$ is an indicator function. The  function $\mathcal{V}_c(\pi)$ is defined as \\ $E[ h\{\pi,X,A(1),A(-1)\} \mid A(1)>A(-1)  ]$ where \\
$h\{\pi,X,A(1),A(-1)\} = I\{\pi(X)=1\} E\{Y_1 \mid A(1)>A(-1), X\}+
    I\{\pi(X)=-1\} E\{Y(-1) \mid A(1)>A(-1),X\}$. 
    
The key building block of the latter result is that under a monotonicity assumption with binary instruments and binary compliances, the group of subjects with $A(1)>A(-1)$ consists of the compliers. Importantly, $p\{A(1)>A(-1)\}$ can be identified using the observed data by $p(A=1\mid Z=1)-p(A=1\mid Z=-1)$. However, with a binary instrument and a three-level compliance variable, the event $A(1)>A(-1)$ does not correspond to the complier group (e.g., the inequality is also satisfied among individuals in S6). Hence, the probability of being a complier (i.e., $\eta^{S4}=p\{A(1)=1,A(-1)=-1\}$) is no longer identifiable using the observed data. To see where the issue arises, we let $\eta^{S}$ denote the proportion of strata $S$ and consider the following system of linear equations: $p(A = 1|Z = 1) = \eta^{S2} + \eta^{S4} + \eta^{S6}; 
    p(A = 1|Z = -1) = \eta^{S2};  
    p(A = -1|Z = 1) = \eta^{S1}; 
    p(A=-1|Z = -1) = \eta^{S1} + \eta^{S4} + \eta^{S5}; 
    p(A = 0| Z = -1) = \eta^{S3} + \eta^{S6};$ and  
 $p(A = 0|Z =1) = \eta^{S3} + \eta^{S5}.$ 
The system of equations does not have a unique solution for $\eta^{S4}$ (i.e., $p\{A(1)=1,A(-1)=-1)\}$.

\section{Methodology}
 We will show that the compliers' value function $E\{Y(\pi) \mid PS =S4\}$ can be identified using a vector of sensitivity parameters. Let $f(y\mid PS = S4,X,A(Z)=Z)$ denote the conditional density of the outcome given $S4$, baseline covariates, and $A(Z)=Z$. Then by the Bayes rule,  

{\footnotesize
 \begin{align} \label{eq:fycom}
f(y\mid PS = S4,X,A(Z)=Z) = \frac{p(PS = S4 \mid Y=y,X,A(Z)=Z)f(y\mid X,A(Z)=Z)}{p(PS = S4 \mid X,A(Z)=Z)}.
\end{align}
}

 Let $w_{\alpha_Z}(A,Z,X,Y) = p(PS = S4| Y,X,A,Z)$. Assuming logit models for the probabilities $p(PS = S4 \mid Y=y,X,A(1)=1)$ and $p(PS = S4 \mid Y=y,X,A(-1)=-1)$, we have $$w_{\alpha_Z}(A,Z,X,Y) =I(A=Z)\frac{\exp\{\mathcal{G}(X,Y,\alpha_{Z})\}}{1 + \exp\{\mathcal{G}(X,Y,\alpha_{Z})\}}, $$ %and perform the sensitivity analysis based on the following models
where  $\mathcal{G}$ is a known user-specified function of $(X,Y)$ parametrized using the vector $\alpha_Z$ (i.e., the sensitivity parameters). Throughout the paper, we denote the denominator of (\ref{eq:fycom}) as 
$\gamma(A,Z,X)= \int w_{\alpha_Z}(A,Z,X,y) f(y\mid X,A,Z) dy. $

 Generally, a sensitivity analysis does not rely on
a correctly specified sensitivity model as the model is unknown and, typically, unidentifiable using the
observed data. The goal is to construct a model that captures the sensitivity of the results to the departure of certain assumptions. In our case, the goal is to assess the departure from the independence assumption $S4 \perp Y \mid X, A(Z)=Z$. Under the latter assumption, $f(y\mid PS = S4,X,A(Z)=Z) = f(y\mid X,A(Z)=Z)$, which is identifiable using the observed data. A reasonable choice for the sensitivity function $\mathcal{G}$ is a linear function $\alpha_{Z}^0+X\alpha_{Z}^X + Y\alpha_{Z}^Y$ where $\alpha_Z=(\alpha_{Z}^0,\alpha_{Z}^X,\alpha_{Z}^Y)^{\top}$. The main drawback of the latter function is that even for low dimensional  $X$, the resulting sensitivity analyses might not be interpretable due to the presence of several sensitivity parameters.  A common choice that mitigates the dimensionality issue while providing meaningful sensitivity analyses is $\mathcal{G}(X,Y,\alpha_{Z})=\alpha_{Z}^0 + Y \alpha_{Z}^Y$  \citep{franks2019flexible, robins2000sensitivity, scharfstein2021semiparametric}.  
Among individuals with $Z=1$, $\alpha_{+1}^Y>0$ ($\alpha_{+1}^Y<0$) implies that 
the odds of being a complier among patients who have complied with $Z=1$ (i.e., $p(PS = S4| A=1, Z=1, Y=y)$)
 is higher (lower) for a larger outcome $y$. 
The parameter $\alpha_{-1}^Y$ has a similar interpretation but among those with $Z=-1$. In Section \ref{sec:addsim}, we discuss other choices of $\mathcal{G}$ function.

\begin{theorem} \label{main-theo}
Under Assumptions \ref{assump:coreiv} - \ref{assump:mono}, for any $\pi \in \mathcal{D}$, the compliers' value function $\mathcal{V}^c(\pi) = E\{Y(\pi)|PS = S4\}$ is nonparametrically identified for a given set of sensitivity parameters $(\alpha_{-1},\alpha_{+1})$, 
\begin{align} \label{eq:nonrobust}
    \mathcal{V}^c(\pi) = E\left[ \frac{I\{\pi(X) = Z\}A(A+Z)Yw_{\alpha_Z}(A,Z, X, Y)}{2\gamma(A,Z,X)f(A,Z|X)} \right].
\end{align}
Accordingly, the optimal regime among compliers is identified as 
\begin{align} \label{eq:optpolicy}
    \pi_c^{opt} = \underset{\pi\in \mathcal{D}}{\arg \! \max}~ E\left[ \frac{I\{\pi(X) = Z\}A(A+Z)Yw_{\alpha_Z}(A,Z, X, Y)}{2\gamma(A,Z,X)f(A,Z|X)} \right].
\end{align}
\end{theorem}

\subsection{Weighted Learning with sensitivity parameters} \label{sec:ipw}

Our result in Theorem~\ref{main-theo} indicates that under certain assumptions, the optimal regime can be identified as a function of sensitivity parameters. 
 The optimization task of equation (\ref{eq:optpolicy}) is equivalent to $\underset{\pi \in \mathcal{D}}{\arg \! \min}~ E\Big[WI\{Z \ne \pi(X)\} \Big]$, { where $W \equiv W(\gamma,f) = \frac{A(A+Z)Yw_{\alpha_Z}(A,Z,X,Y)}{2\gamma(A,Z,X)f(A,Z|X )}$.} To avoid the weight $W$ to be negative, \cite{liu2016robust} proposed  to minimize $E\Big[|W|I\{\text{sign}(W)Z \ne \pi(X)\} \Big]$ where $\text{sign}(W) = W/|W|$ for $|W| > 0$ and $\text{sign}(W) = 0$ otherwise. 

The minimization of the objective function is difficult due to discontinuity and non-convexity of 0-1 loss. To overcome this, we proceed with convex relaxation via the use of the hinge loss function with the added penalty term to avoid over-fitting \citep{zhao2012estimating}. Thus, we estimate the optimal strategy by minimizing the following objective function
\begin{equation} \label{four}
    \frac{1}{n} \sum_{i = 1}^n |W_i| \phi\{\text{sign}(W_i)Z_ig(X_i)\} + \frac{\lambda}{2}\|g\|^2,
\end{equation}
 where $\phi$ is the hinge loss function,  $\lambda$ is penalty term, $g$ is the decision function,   and $\|g\|$ is the Euclidean norm of $g$. We refer to \citep{zhao2012estimating} for selecting the appropriate norm. Let $\hat W \equiv W(\hat \gamma,\hat f)$  denote an estimator of $W$.   Accordingly, we define our inverse probability weighted  optimal strategy estimator as 
 \begin{align} \label{eq:ipwpiest}
     \hat \pi_{ipw} = \text{sign} ( \hat g),
 \end{align}
 where $\hat g =\underset{g }{\arg \! \min}~ 
  \frac{1}{n} \sum_{i = 1}^n |\hat W_i| \phi\{\text{sign}(\hat W_i)Z_ig(X_i)\} + \frac{\lambda}{2}\|g\|^2.$
For any $\pi \in \mathcal{D}$, the asymptotic linearity of the corresponding inverse probability weighted value estimator $\hat{\mathcal{V}}_{ipw}^c(\pi)$ relies on $\|\hat \gamma - \gamma\|_{2,\mu} + \|\hat f - f \|_{2,\mu}=o_p(n^{-1/2})$ where
\begin{align} \label{eq:ipwvest}
 \hat{\mathcal{V}}^c_{ipw}(\pi) = \frac{1}{n} \sum_{i=1}^n \hat W_i I\{\pi(X_i)=Z_i\}. %Y_i. 
\end{align}

 \subsection{Multiply robust sensitivity analyses}
 
  We propose a multiply-robust estimator of $W$ and $\mathcal{V}_{c}(\pi)$, which leads to  our proposed multiply-robust sensitivity analysis approach.  

 The optimal strategy considered %in Section \ref{sec:ipw}, 
 can be equivalently defined as the minimizer of
 \begin{align} \label{eq:blip}
     E\left[Z\Delta(X)I\{\pi(X) \ne Z\}\right]
 \end{align}
 with respect to regime $\pi \in \mathcal{D}$ where $\Delta(X) = E\big[Y(1) - Y(-1)\mid X, PS = S4\big]$ is a blip function. By definition, an optimal regime is the one that maps $X$ to $\text{sign}\{\Delta(X)\}$. That is when the assigned treatment is suboptimal $Z \Delta(X)<0$. Hence, the objective function (\ref{eq:blip}) identifies a regime that if not followed %(i.e., the agent takes exactly opposite actions to those suggested by the regime) 
 would result in the largest loss (i.e., the most negative value of the objective function). Using  (\ref{eq:fycom}),  the blip function can be written as a function of our sensitivity parameters as
 $\Delta(X) = \frac{Q(1,1, X)}{\gamma(1, 1,X)} - \frac{Q(-1, -1,X)}{\gamma(-1, -1, X)},$ 
 where 
 $Q(A,Z,X) = E\{Yw_{\alpha_Z}(A,Z, X, Y)|A,Z,X\}$.  The  estimator $\hat \Delta(X)$ can be obtained by replacing the functions $Q$ and $\gamma$ with their corresponding estimators $\hat Q$ and $\hat \gamma$. Hence, the consistency of $\hat \Delta(X)$ relies on the consistency of both $\hat Q$ and $\hat \gamma$. To improve robustness to model misspecification,  in Theorem \ref{th:phi} and Corollary \ref{cor:phi}  (\ref{sec:supp-Deltamr} of the supplementary material), we derive the following multiply robust statistic 
 
 \begin{align} \label{deltaXmr}
     \Delta_{mr}(X) &=  \frac{A(A+Z)}{2\gamma(A,Z,X)f(A,Z|X)}\Big[Yw_{\alpha_Z}(A,Z, X, Y) - Q(A,Z,X)  \\
    &- \delta(A,Z,X)\big\{w_{\alpha_Z}(A,Z, X, Y) - {\gamma}(A,Z,X)\big\} \Big] + {\Delta}(X). \nonumber
 \end{align}
where   $\delta(A,Z,X) = \frac{Q(A,Z,X)}{\gamma(A,Z,X)}$. Accordingly, we define a multiply robust weight function $Z\Delta_{mr}(X)$ as

{\small
\begin{align*}
    W_{mr} \equiv W_{mr}(Q,\gamma,\delta,f) &=  \frac{A(A+Z)}{2\gamma(A,Z,X)f(A,Z|X)}\Big[Yw_{\alpha_Z}(A,Z, X, Y) - Q(A,Z,X) \\
    &- \delta(A,Z,X)\big\{w_{\alpha_Z}(A,Z, X, Y) - {\gamma}(A,Z,X)\big\} \Big] + Z{\Delta}(X). 
\end{align*}
}
 Let $W^*_{mr} \equiv W^*_{mr}(Q^*,\gamma^*,\delta^*,f^*)$.  Then, the  function 
$E[W^*_{mr} I\{\mathcal{D}(L) \ne Z\}]$ is multiply robust in the sense that maximizes the value function in the union of the following models:
\begin{align*}
    &\mathcal{M}_1: \text{ Models for } %\delta^z(\bX),
    Q(A,Z,X), \gamma(A,Z,X) \text{ are correctly specified} \\ 
    &\mathcal{M}_2: \text{ Models for } f(A,Z|\bX), \gamma(A,Z,X) \text{ are correctly specified.}
\end{align*}

Let $\hat W_{mr} \equiv W_{mr}(\hat Q,\hat\gamma,\hat\delta,\hat f)$ be the estimator of $W_{mr}$. We can plug it in the objective function (\ref{four}) and define our multiply robust  optimal strategy estimator as 
\begin{align} \label{eq:mrpiest}
    \hat \pi_{mr} = \text{sign} ( \hat g_{mr}),
\end{align}
where $\hat g_{mr} =\underset{g}{\arg \! \min}~ 
  \frac{1}{n} \sum_{i = 1}^n |\hat W_{mr,i}| \phi\{\text{sign}(\hat W_{mr,i})Z_ig(X_i)\} + \frac{\lambda}{2}\|g\|^2.$

Treatment strategies are often quantified by the value function. Theorem \ref{th:val} proposes a multiply robust estimator of a compliers' value function. The results show that the corresponding estimator is asymptotically linear  in the union of the nuisance models  $\mathcal{M}^\dagger_1$, $\mathcal{M}^\dagger_2$, and $\mathcal{M}^\dagger_3$. The nuisance models are defined as
\begin{align*}
    &\mathcal{M}^\dagger_1: f(A|Z,X), \gamma(A,Z,X), f(Z|X) \text{ are correctly specified.} \\
    &\mathcal{M}^\dagger_2: f(A|Z,X), \gamma(A,Z,X), \kappa(Z,X) \text{ are correctly specified.} \\
    &\mathcal{M}^\dagger_3: \gamma(A,Z,X), f(Z|X), Q(A,Z,X) \text{ are correctly specified.}
\end{align*}

where $\kappa(Z,X) = E\left\{\frac{A(Z + A)Yw_{\alpha_Z}(A,Z, X, Y)}{2\gamma(A,Z,X)f_A(A|Z,X)}|Z,X\right\}.$ 

%In practice, we cannot always be confident that we specify correctly all the models for nuisance parameters. Therefore, it is necessary to develop a multiply robust estimator for the compliers' value function. It is robust in the sense that even if we misspecify one or more nuisance parameters, we can have valid inference about $V^c(\pi)$. 
\begin{assumption} (Accuracy of the nuisance models for value function) \label{assump:vrate}
Let $f_A \equiv f(A|Z,X)$ and $f_Z \equiv f(Z|X)$. For any $z$,
$(\|f_Z- \hat f_Z \|_{2} + \|f_A - \hat f_A \|_{2} +\|\gamma- \hat \gamma \|_{2})\|\gamma- \hat \gamma \|_{2} + (\|f_A - \hat f_A \|_{2} +\|\gamma- \hat \gamma \|_{2})\|Q- \hat Q \|_{2}
+  \|\gamma - \hat \gamma\|_{2}\|\delta - \hat \delta \|_{2} + \|\kappa - \hat \kappa \|_{2}\|f_Z- \hat f_Z \|_{2} = o_p(n^{-1/2})$.
%where $p = f(A|Z,X)$ and $f = f(Z|X)$
\end{assumption}

  Assumption \ref{assump:vrate} is satisfied, for example,  when the nuisance parameters in each $\mathcal{M}_{i}^\dagger$, $i \in \{1,2,3\}$  converge to their true values at a rate of $o_p(n^{-1/4})$. This allows us to use nonparametric methods to estimate nuisance parameters while having a root-$n$ estimator for the value function. The highly adaptive lasso is one of the several nonparametric methods that satisfies this rate condition \citep{van2017generally}.    

\begin{theorem} \label{th:val}
 Under Assumptions \ref{assump:coreiv} - \ref{assump:mono} and \ref{assump:vrate}, the estimator \\$\hat{\mathcal{V}}_{mr}^c(\pi)=P_n \xi(O,\hat Q,\hat\gamma,\hat\kappa,\hat f_A,\hat f_Z)$ is an asymptotic linear estimator of ${\mathcal{V}}_{mr}^c(\pi)$ such that  
$
\hat{\mathcal{V}}_{mr}^c(\pi) - \mathcal{V}_{mr}^c(\pi)=P_n \xi_V(O,Q,\gamma,\kappa,f_A,f_Z)+o_p(n^{-1/2}),
$ %\end{align*}
where the efficient influence function $\xi_V(O,Q,\gamma,\kappa,f_A,f_Z)$ is defined in the Section \ref{supp:proofth2} of the supplementary materials.  
The estimator $\hat{\mathcal{V}}_{mr}^c(\pi)$ is  consistent  under \\ $ \mathcal{M}^\dagger_1 \cup \mathcal{M}^\dagger_2 \cup \mathcal{M}^\dagger_3$. Furthermore, $\hat{\mathcal{V}}_{mr}^c(\pi)$ is semiparametric locally efficient in $\mathcal{M}_{union}$ at the intersection sub-model $\mathcal{M}_{int} = \mathcal{M}_1^\dagger \cap \mathcal{M}_2^\dagger \cap \mathcal{M}_3^\dagger.$  
\end{theorem}

 Despite the ease of use, the inverse probability weighted estimator $\hat{\mathcal{V}}_{ipw}^c(\pi)$ requires the correct specification of all the nuisance parameters (i.e., weight functions),
is known to be inefficient, and suffers from the curse of dimensionality. In fact, $\hat{\mathcal{V}}_{ipw}^c(\pi)$ fails to be asymptotically linear when weight functions are estimated using data-adaptive techniques (i.e., non-parametric methods) \citep{ertefaie2022nonparametric}. The multiply robust estimator $\hat{\mathcal{V}}_{mr}^c(\pi)$ overcomes  these shortcomings and is asymptotically linear  as long as the rate Assumption \ref{assump:vrate} holds.  Similarly, the multiply robust optimal strategy estimator $\hat \pi_{mr}$ improves the performance of $\hat \pi_{ipw}$ by including a multiply robust weight function $W_{mr}$.

\section{Simulation Studies}
\subsection{Scenarios and competing methods}
 We conduct simulation studies to demonstrate the performance of our  methods. First, we compare our proposed inverse probability weighted  (IPW) and multiply robust estimators (MR) to the outcome weighted learning estimator (OWL) of \cite{zhao2012estimating} and the IV-based estimator (IVT) of \cite{cui2021semiparametric}. We violate the no unmeasured confounders assumption and allow the levels of the endogenous variable (i.e., compliance) to be greater than the IV, thereby the latter two methods are expected to fail.  Second, we will assess the robustness of our proposed estimators presented in Theorems \ref{main-theo}  and \ref{th:val} to the misspecification of nuisance parameters. In this section, we refer to the value function estimators in equation (\ref{eq:ipwvest}) and Theorem \ref{th:val} as the IPW and MR estimators, respectively.
 
Our simulation studies consist of two main sections. In the first section, we assume that the sensitivity parameters $\alpha_Z$ are known to assess the robustness of our estimators to certain misspecifications. In the second section, the true values of those sensitivity parameters are unknown. %In that case, an analysis of sensitivity will be conducted.
Throughout, we consider the following scenarios to demonstrate the robustness properties of the classification weight function $W_{mr}$: (1) All nuisance parameters are correctly specified ($\mathcal{M}_1 \cup \mathcal{M}_2$); (2) $f(A,Z|X)$ is incorrectly specified ($\mathcal{M}_1$); (3) $Q(A,Z,X)$ is incorrectly specified ($\mathcal{M}_2$). 
Moreover,  to examine the robustness of the multiply robust value function estimator, we consider the following scenarios: (1) All nuisance parameters are correctly specified ($\mathcal{M}^\dagger_1 \cup \mathcal{M}^\dagger_2 \cup \mathcal{M}^\dagger_3$); (2) $f(Z|X)$ is incorrectly specified ($\mathcal{M}^\dagger_2$); (3) $f(A|Z,X)$ is incorrectly specified ($\mathcal{M}^\dagger_3$); (4) $Q(A,Z,X)$ is incorrectly specified ($\mathcal{M}^\dagger_1$).

\subsection{Generative models}

%We set the true values of the sensitivity parameters $( \alpha_{-1}, \alpha_{+1})$ to be (-0.5,0.5) and (0.5,-0.5) for the first and second section of the simulation, respectively. 
We generate the covariates $X_1$ and $X_2$ from the uniform distributions $(-1,1)^2$ and Z is generated from a Bernoulli distribution with $p(Z = 1\mid X_1,X_2) = \frac{\exp(0.3 - 2X_1 + 2X_2)}{1+\exp(0.3 -2X_1 + 2X_2)}$. The unmeasured confounder, $U$ is drawn from the bridge distribution with density $g_b(u) = \frac{1}{2\pi}\frac{\text{sin}(0.5 \pi)}{\text{cosh}(0.5 u) + \text{cos}(0.5 \pi)}, ( -\infty < u < \infty)$.  The bridge distribution ensures that the marginalized distribution $f(A\mid Z,X,U)$ over $U$ (i.e.,  $f(A\mid Z,X)$) remains a multinomial logistic regression \citep{wang2003matching}. We consider the strata proportions as following $\log\Big\{\frac{p(PS = S1 \mid X,Z,U)}{p(PS = S3\mid X,Z,U)}\Big\} =  0.5X_1 + 0.5Z + U,$ $\log\Big\{\frac{p(PS = S2\mid X,Z,U)}{p(PS = S3\mid X,Z,U)}\Big\} =  -0.5X_1 + 0.5Z + U,$ \\$\log\Big\{\frac{p(PS = S4\mid X,Z,U)}{p(PS = S3\mid X,Z,U)}\Big\} =  -0.5X_1 + 0.5Z - U$, $\log\Big\{\frac{p(PS = S5\mid X,Z,U)}{p(PS = S3\mid X,Z,U)}\Big\} =  0.5X_1 + 0.5Z - U$, $\log\Big\{\frac{p(PS = S6\mid X,Z,U)}{p(PS = S3\mid X,Z,U)}\Big\} =  0.5X_1 + 0.5Z - U.$  We set the sensitivity function $\mathcal{G}(X,Y,\alpha_Z) = Y\alpha_Z^Y $. Additional simulations where $\mathcal{G}(X,Y, \alpha_Z) = \alpha_{Z}^0+X\alpha_{Z}^X + Y\alpha_{Z}^Y$ can be found in the \ref{appendixd} in the Supplementary Material.    

 The compliance level $A$ is then determined based on the subject's treatment and type of principal strata. The outcome of individuals with $A = Z = -1$ and  $A = Z = 1$ are generated from the normal distribution  $f_{-1}(y|X_1, X_2) = N(1 + 2X_1 + 2X_2, 0.5^2)$ and $f_{+1}(y|X_1, X_2) = N(1, 0.5^2)$.  From the distribution $f_{-1}(y)$ and $f_{+1}(y)$ and the pre-specified value of $\alpha$'s, we generate the compliers outcome using rejection sampling (\ref{rejectionsampling}). The outcome of individuals with $(Z= -1, A = 1)$, $(Z = 1, A = -1)$, $(Z = 1, A = 0)$, $(Z = -1, A = 0)$ are generated from $N(3 + X_1 + X_2, 0.5^2)$, $N(-1 + X_1 + X_2, 0.5^2),$ $N(5, 0.1^2)$, $N(-5,0.1^2)$, respectively.  In our data generative model, the variable $U$ is an unmeasured confounder as it is associated with both $A$ and $Y$ through compliance classes.   We generate 500 datasets with a sample size of 500. Our generative model results in roughly 240 compliers in each data set.  We report the average correct classification rate of the estimated optimal policy and the resulting estimated value function  over the 500 simulated datasets. In Tables \ref{tab:clasrate}-\ref{tab:valrob}, we report the standard errors in parentheses.

\subsection{Nuisance parameters approximation and  estimation}

 Our sensitivity analyses rely on correctly specified $\gamma(a,z,x)$. However, it cannot be directly inferred from the generative model. To overcome this issue, 
 we approximate $\gamma(a,z,x) \approx \frac{1}{5000} \sum_{j=1}^{5000} y_{jz} w_{\alpha_z}(a,z,y_{jz})$ where $y_{jz}$ is a random sample from $f_{-1}(y|x_1, x_2)$ when $z = -1$, and $f_{+1}(y|x_1, x_2)$ otherwise. 
  In our simulations, to estimate $f_{z}(y|x_1, x_2)$,  we assume a normal distribution for the outcome $Y$ and estimate the corresponding means and variances using linear regressions.  Alternatively, to gain robustness, one can estimate  the conditional cumulative function $Y$ given $(A=z,Z=z,X=x)$, i.e. $F_{z}(y|x_1, x_2)$ using a single-index model \citep{chiang2012new}. We estimate the propensity score $f(Z|X)$ using logistic regression models and $f(A|Z,X)$ using multinomial logistic regression. %Following \cite{cui2021semiparametric, zhao2012estimating}, we use cross-validation to determine the tuning parameter in the algorithm. 
 To estimate the value function, we will use the cross-fitted Highly Adaptive Lasso (HAL)   \citep{benkeser2016highly} to estimate  $\kappa(Z,X)$. The misspecified nuisance models are fitted similarly where both $X_1$ and $X_2$ are removed from the model.

\begin{table}
\centering
\caption{Simulation Result: Mean (SD) of correct classification rate. The sensitivity parameter are ($ \alpha_{-1}, \alpha_{+1}$) = (0.5,0.5) }
\begin{tabular}{rlllll}
  \hline
 & Case & OWL & IVT & IPW & MR\\ 
  \hline
  1 & All correctly specified & 0.51 (0.03) & 0.73 (0.07) & 0.88 (0.04) & 0.95 (0.02) \\ 
  2 & $f(A, Z|X)$ misspecified & 0.51 (0.03) & 0.73 (0.06) & 0.78 (0.04) & 0.97 (0.02) \\ 
 % 3 & f(A$|$Z,X) misspecified & 0.55 (0.03) & 0.76 (0.05) & 0.91 (0.03) & 0.96 (0.02) \\ 
  3 & $Q(A,Z,X)$ misspecified & 0.51 (0.03) & 0.72 (0.07) & 0.88 (0.05) & 0.88 (0.05) \\ 
   \hline
\end{tabular} \label{tab:clasrate}
\end{table}

\subsection{Known \texorpdfstring{$\alpha$'s}{alpha's}}
\label{simknown}
In this simulation, the true set of parameters $\alpha_Z^Y = ( \alpha_{-1}^Y, \alpha_{+1}^Y)$ is (0.5,0.5). Tables \ref{tab:clasrate} and \ref{tab:value} report the correct classification rate of the estimated optimal rules and their corresponding Monte Carlo approximated value function, respectively. As expected, because of the presence of unmeasured confounders, the OWL method fails. The average correct classification rate for OWL is 51\text{\%} resulting in the lowest value function among the methods considered in all three scenarios. The IVT method results in the average correct classification rates of roughly 73\text{\%} and the average value function of roughly 1.4. In contrast, the proposed IPW estimator produces an average correct classification rate of 88\text{\%}, and the average value function estimators are around 1.64 when all the nuisance parameters are correctly specified. Notably, even when the nuisance parameter, $f(A,Z|X)$, is misspecified, the IPW estimator still outperforms the existing methods. Importantly, the proposed multiply robust estimator shows the best performance among the methods considered. Specifically, the corresponding rule matches the true rule 95\% of the time when all the nuisance parameters are correctly specified. The accuracy of the rules estimated using multiply robust method drops slightly to 88\% when 
$Q(A,Z,X)$ is misspecified.

Table \ref{tab:valrob} demonstrates the performance of our proposed value function estimators in equation (\ref{eq:ipwvest}) and Theorem \ref{th:val}. In case 1, when all nuisance parameters are correctly specified, the estimated value functions using IPW and MR estimators are relatively close to the true value. When a nuisance parameter is misspecified (i.e., cases 2-4), the estimated value of the MR estimator closely matches the corresponding true value. However, the IPW estimator  overestimates the true value function when either $f(Z|X)$ or $f(A|Z,X)$ are misspecified. The misspecification of $Q(A,Z,X)$ does not impact the IPW estimator as the estimator does not depend on it. As expected, the standard error of the estimated value function using MR is uniformly smaller than those obtained by IPW.

\begin{table}[t] 
\centering
 \caption{Simulation Result: Mean (SD) of value functions. The true optimal value function is 1.68. The sensitivity parameter are ($\alpha_{-1}^Y, \alpha_{+1}^Y$) = (0.5,0.5) }
\begin{tabular}{rlllll}
  \hline
 & Case & OWL & IVT & IPW & MR\\ 
  \hline
  1 & All correctly specified & 1.05 (0.00) & 1.49 (0.12) & 1.64 (0.06) & 1.67 (0.06) \\ 
  2 & f(A, Z$|$X) misspecified & 1.05 (0.01) & 1.49 (0.10) & 1.55 (0.07) & 1.68 (0.06) \\ 
  3 & $Q(A,Z,X)$ misspecified & 1.05 (0.00) & 1.49 (0.12) & 1.65 (0.07) & 1.64 (0.07) \\ 
   \hline
\end{tabular} \label{tab:value}
\end{table}

\begin{table}
\centering
\caption{Simulation Result: Comparing the performance of the proposed value function estimators. The sensitivity parameter are ($ \alpha_{-1}^Y, \alpha_{+1}^Y$) = (0.5,0.5)}
\begin{tabular}{rllll}
  \hline
 & Case & True & IPW & MR \\ 
  \hline
  1 & All correctly specified & 1.65  & 1.65 (0.14) & 1.62 (0.09) \\ 
  2 & f(Z$|$X) misspecified & 1.52  & 1.92 (0.12) & 1.48 (0.07) \\ 
  3 & f(A$|$Z,X) misspecified & 1.64  & 2.55 (0.31) & 1.60 (0.14) \\ 
  4 & $Q(A,Z,X)$ misspecified & 1.64  & 1.66 (0.13) & 1.62 (0.10) \\ 
   \hline
\end{tabular}\label{tab:valrob}
\end{table}

\subsection{Unknown \texorpdfstring{$\alpha$'s}{alpha's}}
\label{simunknown}

\begin{figure}
    \centering
    \includegraphics[scale=.5]{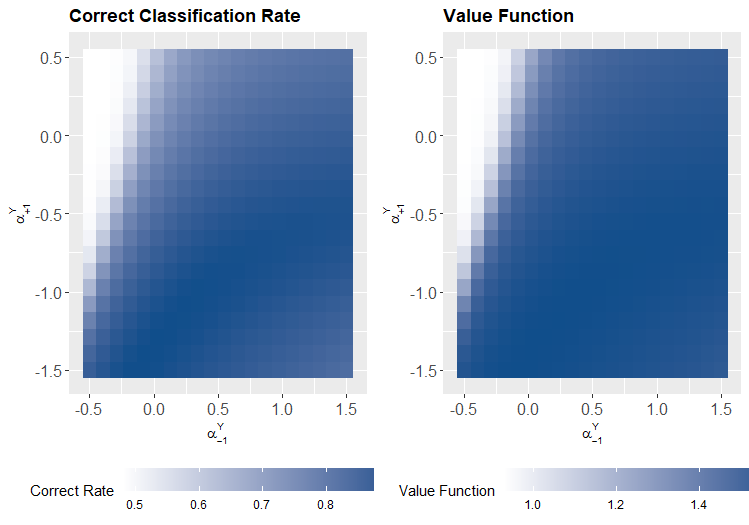}
    \caption{Sensitivity analysis of the proposed method (IPW): Sensitivity parameters are ($ \alpha_{-1}^Y, \alpha_{+1}^Y$). The darker area indicates a higher correct classification rate/value function. The true vector of sensitivity parameters are ($ \alpha_{-1}^Y, \alpha_{+1}^Y$) = (0.5,-0.5). }
    \label{fig:prop-sensi}
\end{figure}

\begin{figure}
    \centering
   \includegraphics[scale=.5]{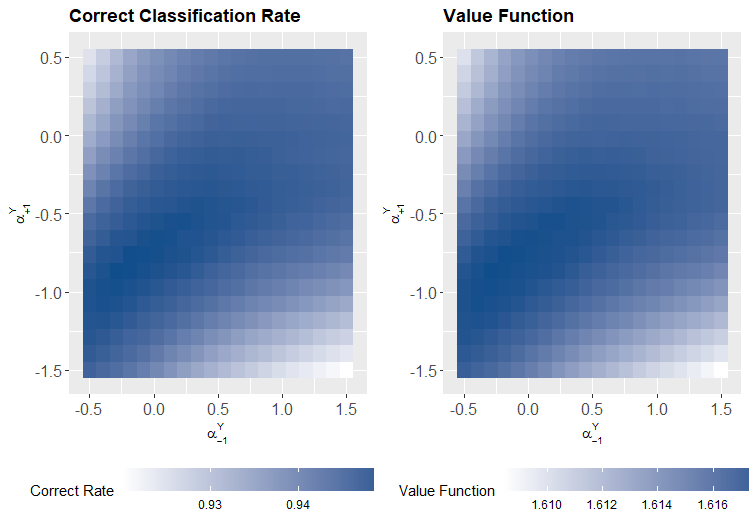}
    \caption{Sensitivity analysis of the multiply robust method when all nuisance parameters are correctly specified: Sensitivity parameters are ($ \alpha_{-1}^Y, \alpha_{+1}^Y$). The darker area indicates a higher correct classification rate/value function. The true vector of sensitivity parameters are ($ \alpha_{-1}^Y, \alpha_{+1}^Y$) = (0.5,-0.5). }
    \label{fig:mr-sensi}
\end{figure}

We now consider the case that the true sensitivity parameters $\alpha_Z^Y = (\alpha_{-1}^Y, \alpha_{+1}^Y)$ are unknown. Hence, we conduct a sensitivity analysis to examine our methods' performance under the misspecification of  $\alpha$'s. In this simulation, we set the true set of parameters $( \alpha_{-1}^Y, \alpha_{+1}^Y)$ to be (0.5,-0.5). With this set of parameters, the average correct classification rate obtained using the IVT method is 68.96\text{\%} with the average corresponding value function of 1.37.
Figure \ref{fig:prop-sensi} displays the performance of the IPW estimators. The heat map in the left column presents average correct classification rates and average value functions  when the $\alpha$'s are varied from their true value. For most of the specified grids, the IPW method outperforms  the IVT. However, when $\alpha_{-1}^Y<0$, and $\alpha_{+1}^Y>-0.5$, the average classification rate drops to nearly 0.5 with the average corresponding value function of nearly 1.0.  In contrast, the MR estimator with correctly specified nuisance parameters is less sensitive to the deviation from the true $\alpha$'s and remains superior to IVT for a substantially larger area of the grid points (Figure \ref{fig:mr-sensi}). When some of the nuisance parameters are misspecified, the correct classification rate drops slightly but still outperforms both the IVT and the IPW approach (Figures \ref{fig:mrwsensi2} - \ref{fig:mrwsensi4} in the Web \ref{appendixd}).

Our simulation results highlight several important points. First,  the MR method is  more robust to the misspecification of $\alpha$'s than the IPW. Second, any knowledge about the sign of $\alpha$'s can improve the interpretability of the results, particularly for the IPW-based approach. Third, the existing IV-based methods (e.g., IVT) may result in severely suboptimal rules. 

\subsection{Additional simulation studies} \label{sec:addsim}

In \ref{missenmod} of the supplementary material, we assess the sensitivity of our approach to misspecification of the sensitivity function $\mathcal{G}(X,Y, \alpha_Z)$. Specifically, we allow  $\mathcal{G}$ to depend on both $Y$ and $X$ in the true generative model, but in our analyses, we only include either $Y$ or the first component of the principal component analysis (PCA) of $(X,Y)$ in $\mathcal{G}$. Our MR methods show robustness to the misspecification of the $\mathcal{G}$ function in this particular setting. Specifically, the MR method performed well for all values of the sensitivity parameters $\alpha_{Z}$. The average correct classification rate is consistently above 90\text{\%}. This finding is consistent with the results obtained in Section \ref{simknown}, where the sensitivity function was correctly specified. There are two possible explanations for  this robustness property. First, because the measured confounders $X$ have already been accounted for,  the inclusion of them in the sensitivity function $\mathcal{G}$ does not change the result in a meaningful way. Second, the outcome $Y$ is a summary function of measured and unmeasured confounders. So either $Y$ or the first component of PCA of $(X,Y)$ might be sufficient to model the probability of being a complier. However, it is important to note that this robustness may not hold in all settings, and thus, further investigation is needed.    

\section{Application}
\label{sec:application}

We implement our approach to the Adaptive Treatment for Alcohol and Cocaine Dependence (ENGAGE) study \citep{mckay2015effect}.
 This study recruited 500 individuals to enter the intensive outpatient program (IOP) consisting of attending three sessions per
week for two weeks. Those who failed to attend at least two sessions in Week 2 were eligible to get randomized to one of two telephone motivational interviewing
(MI) based interventions ($n=189$). One intervention was to encourage patients to engage in IOP (MI-IOP), and the other included a choice of IOP or three other treatment options (MI-PC). The binary outcome of interest is whether an individual has any cocaine use at the end of the eight-week program (i.e., $Y=-1$ indicates cocaine use and $Y=1$ indicates no cocaine use).

Our analyses include the following baseline variables:  race, sex, education, smoking status, treatment readiness, and general health. We include all these covariates in the weighted SVM models for treatment classification and models for estimating nuisance parameters. The instrument variable $Z$ is the assigned intervention, either MI-IOP ($Z=-1$) or MI-PC ($Z=1$). The observed compliance level $A$ is decided based on the number of MI  sessions attended by an individual. If a person attends more than the median of the attended sessions of other patients, we set $A=Z$. Otherwise, we set $A=0$.

 We note that since treatments in the ENGAGE study are randomized, we know the true $f(Z|X)$. This allows us to simplify our multiply robust estimator as follows
 \begin{align*}
     \mathcal{\hat {\tilde{V}}}^c_{mr}(\pi) &= P_n\frac{A(Z+A)I\{\pi(X) = Z\}Yw_{\alpha_Z}(A,Z, X, Y)}{2\hat\gamma(A,Z,X) f(Z|X)\hat f(A|Z,X)} \\
     &- \Big[\frac{A(A+Z)I\{\pi(X) = Z\}  \hat \delta(A,Z,X)}{2\hat \gamma(A,Z,X) f(Z|X)\hat f(A|Z,X)}\{w_{\alpha_Z}(A,Z, X, Y) - \hat \gamma(A,Z,X)\} \Big]  \\ 
      &- \Big[ \frac{A(Z+A)I\{\pi(X) = Z\}\hat Q(A,Z,X)}{2\hat\gamma(A,Z,X)  f(Z|X)\hat f(A|Z,X)} \\
      &\hspace{0.5cm}-  \sum_a  \frac{a(a+Z)I\{\pi(X) = Z\} \hat Q(a,Z,X)}{2\hat \gamma(a,Z,X)f(Z|X)} \Big]. 
 \end{align*}
We will use this modified multiply robust estimator to estimate the value function.

The key difference between $\mathcal{\hat {\tilde{V}}}^c_{mr}(\pi)$ and $\mathcal{\hat {{V}}}^c_{mr}(\pi)$ is that, in the former, the corresponding influence function is no longer orthogonal to the nuisance tangent space of $f(Z|X)$. This choice is made for two reasons: (1) the treatment is randomized, so from the consistency point of view, there is no need to be robust to possible misspecification of $f(Z|X)$; (2) due to a relatively small sample size, the estimation of the complex functional  $\kappa(Z,X)$ may lead to a biased estimate of the value function.    We acknowledge that estimating known nuisance parameters leads to efficiency gain \citep{van2003unified}, but we decided to sacrifice the slight efficiency gain in favor of reducing finite sample bias \citep{liang2020semiparametric}.

The nuisance parameters $f(Z|X)$ and $f(A|Z,X)$ are fitted using the logistic regression and multinomial logistic regression, respectively, and $Q(A,Z,X)$ is obtained using a random forest model. 

 We approximate $\gamma(a,z,x) = \frac{1}{5000} \sum_{i=1}^{5000} y^i_z w_{\alpha_z}(a,z,x,y)$ where $y^i_z$ is a sample from the Bernoulli distribution  with the success probability of $p_z(Y = 1|X)$ is estimated using a random forest model. Note that estimating $f(Z|X)$ may lead to a slightly conservative estimate of the variance of $\mathcal{\hat {\tilde{V}}}^c_{mr}(\pi)$. %  estimation of does not invalidate the inference function corresponding to $\mathcal{\hat {\tilde{V}}}^c_{mr}(\pi)$. 

  We set the sensitivity function $\mathcal{G}(X,Y,\alpha_Z) = \alpha^0_Z + Y\alpha_Z^Y $. The parameter $\alpha_{-1}^Y$ is the log odds coefficient of cocaine use in predicting the probability
that a patient would take MI-PC (i.e., $A=1$) if she were assigned to MI-PC (i.e., $Z=1$) given that the patient would take MI-IOP (i.e., $A=-1$) if she were assigned to MI-IOP (i.e., $Z=-1$). The parameter  $\alpha_{+1}^Y$ is defined similarly.
 There is a clinical sense that individuals with more severe substance use (e.g., $Y=-1$) are less likely to  be compliers due to higher craving \citep{evans2012differential,pergolizzi2020opioid,cavicchioli2020craving}.   That is $p(PS = S4 \mid A = z, Z = z, Y=1 ) > p(PS = S4 \mid A = z, Z = z, Y=0 )$ for both $z=1$ and $z=-1$. This implies that $\alpha_{-1}^Y > 0$ and $\alpha_{+1}^Y > 0$. 
 In the ENGAGE study, 39\% of individuals adhered with their assigned intervention IOP (i.e., $p(A = -1|Z = -1)=0.39$), while 44\% adhered with PC (i.e., $p(A = +1|Z = +1)=0.44$) \citep{mckay2015effect}. Since the probability of complier cannot exceed either of these probabilities, its upper bound is 39\% (i.e., $p(PS = S4) \leq 0.39$). In \ref{calalpha0}, we show that one can obtain a range of plausible $\alpha^0_Z$ as a function of $\alpha^Y_Z$ and $p(PS = S4)$. In this Section, we consider $p(PS = S4) = 0.30$ for our sensitivity analyses. The readers can find the result for different choices of $p(PS = S4)$ in \ref{sensiengage}.

Let $d(\pi,\pi')=\mathcal{\hat {\tilde{V}}}^c_{mr}(\pi)-\mathcal{\hat {\tilde{V}}}^c_{mr}(\pi')$ denote the difference between the estimated multiply robust value function estimators under policies $\pi$ and $\pi'$.  Specifically,  we randomly split the data set into a train and a test sets in a 60:40 ratio. We first estimate the nuisance parameters and the optimal strategy using the train set (i.e., 60\% of the data) and then estimate the corresponding optimal value function using the test set (i.e., 40\% of the data).  We repeat this process one hundred times and implement IPW, MR, IVT, and OWL methods on each resampled data set. Let $\hat \pi_{ipw}$, $\hat \pi_{mr}$, $\hat \pi_{ivt}$ and $\hat \pi_{owl}$ denote the estimated optimal policy of the corresponding methods.  
Figure \ref{fig:compareEric} and \ref{fig:compareOWL} compare the performance of the proposed methods to IVT and OWL methods, respectively.  The reported estimated optimal value functions are averaged over the 100 test sets.   In Figure \ref{fig:compareEric}, the left heat maps show the difference between the estimated value of the estimated optimal policy using the proposed (i.e., IPW and MR) and IVT estimators (i.e., $d(\hat\pi_{mr},\hat\pi_{ivt})$ and  $d(\hat\pi_{ipw},\hat\pi_{ivt})$). The darker blue indicates the area of $\alpha^Y$'s where the magnitude of the difference is larger (i.e., the darker the better performance of our proposed methods). The right heat maps show whether the value functions of the constructed optimal regimes using our methods outperform the IVT method where the red color indicates $d(\hat\pi_{mr},\hat\pi_{ivt})<0$ or $d(\hat\pi_{ipw},\hat\pi_{ivt})<0$.
Figure \ref{fig:compareOWL} shows the same information but for the OWL method. Table \ref{tab:valueapp} shows the value functions and the standard errors of the proposed methods, IVT, and OWL for some selected $\alpha^Y = (\alpha^Y_{-1},\alpha^Y_{+1})$. The MR method yields a higher value compared with both OWL and IVT methods in the region where  $\alpha_{+1}^Y>0$ and $\alpha_{-1}^Y>0$ (i.e., the plausible range). 
In specific regions of $\alpha^Y$, the estimated value functions obtained by the IPW method are lower compared to those produced by the OWL and IVT methods. We hypothesize that the inferior performance of the IPW method relative to the MR method may stem from a potential misspecification of either $f(A|Z,X)$, $f(Z|X)$, or both.

\begin{table}[ht]
\centering
\caption{Value Function (SE) of the IPW, MR, IVT, and OWL methods for some selected $\alpha^Y$.}
\begin{tabular}{rrllll}
  \hline
  & $\alpha^Y=(\alpha^Y_{-1}, \alpha^Y_{+1}$) & IPW & MR & IVT & OWL \\ 
  \hline
1 & (0.00, 0.00) & 0.42 (0.12) & 0.47 (0.12) & 0.45 (0.16) & 0.42 (0.18) \\ 
  2 & (0.00, 0.53) & 0.43 (0.13) & 0.49 (0.13) & 0.47 (0.16) & 0.44 (0.18) \\ 
  3 & (0.00, 1.05) & 0.46 (0.14) & 0.49 (0.14) & 0.49 (0.16) & 0.47 (0.18) \\ 
  4 & (0.00, 1.47) & 0.50 (0.14) & 0.51 (0.14) & 0.51 (0.16) & 0.49 (0.18) \\ 
  5 & (0.00, 2.00) & 0.52 (0.12) & 0.55 (0.12) & 0.53 (0.16) & 0.52 (0.18) \\ 
  6 & (0.53, 0.00) & 0.51 (0.12) & 0.58 (0.12) & 0.52 (0.15) & 0.47 (0.17) \\ 
  7 & (0.53, 0.53) & 0.51 (0.13) & 0.58 (0.13) & 0.54 (0.15) & 0.50 (0.17) \\ 
  8 & (0.53, 1.050 & 0.53 (0.12) & 0.58 (0.12) & 0.56 (0.15) & 0.52 (0.17) \\ 
  9 & (0.53, 1.47) & 0.54 (0.13) & 0.58 (0.13) & 0.57 (0.15) & 0.54 (0.17) \\ 
  10 &(0.53, 2.00) & 0.57 (0.12) & 0.60 (0.12) & 0.59 (0.15) & 0.57 (0.17) \\ 
  11 & (1.05, 0.00) & 0.61 (0.12) & 0.69 (0.12) & 0.58 (0.14) & 0.53 (0.16) \\ 
  12 & (1.05, 0.53) & 0.60 (0.12) & 0.69 (0.12) & 0.60 (0.15) & 0.55 (0.16) \\ 
  13 & (1.05, 1.05) & 0.60 (0.12) & 0.69 (0.12) & 0.62 (0.15) & 0.58 (0.16) \\ 
  14 & (1.05, 1.47) & 0.60 (0.12) & 0.68 (0.12) & 0.64 (0.14) & 0.60 (0.16) \\ 
  15 & (1.05, 2.00) & 0.61 (0.12) & 0.66 (0.12) & 0.66 (0.14) & 0.62 (0.15) \\ 
  16 & (1.47, 0.00) & 0.68 (0.13) & 0.75 (0.13) & 0.63 (0.14) & 0.56 (0.15) \\ 
  17 & (1.47, 0.53) & 0.67 (0.12) & 0.75 (0.12) & 0.65 (0.14) & 0.59 (0.15) \\ 
  18 & (1.47, 1.05) & 0.66 (0.11) & 0.77 (0.11) & 0.67 (0.14) & 0.61 (0.15) \\ 
  19 & (1.47, 1.47) & 0.65 (0.11) & 0.77 (0.11) & 0.68 (0.14) & 0.63 (0.15) \\ 
  20 & (1.47, 2.00) & 0.65 (0.12) & 0.75 (0.12) & 0.70 (0.14) & 0.66 (0.15) \\ 
  21 & (2.00, 0.00) & 0.74 (0.11) & 0.83 (0.11) & 0.68 (0.14) & 0.60 (0.15) \\ 
  22 & (2.00, 0.53) & 0.73 (0.12) & 0.82 (0.12) & 0.69 (0.14) & 0.62 (0.15) \\ 
  23 & (2.00, 1.05) & 0.71 (0.12) & 0.82 (0.12) & 0.71 (0.14) & 0.65 (0.14) \\ 
  24 & (2.00, 1.47) & 0.71 (0.12) & 0.85 (0.12) & 0.73 (0.14) & 0.67 (0.14) \\ 
  25 & (2.00, 2.00) & 0.70 (0.11) & 0.84 (0.11) & 0.75 (0.13) & 0.70 (0.14) \\ 
   \hline
\end{tabular}

\label{tab:valueapp}
\end{table}

\begin{figure}
    \centering
    \includegraphics[scale=.5]{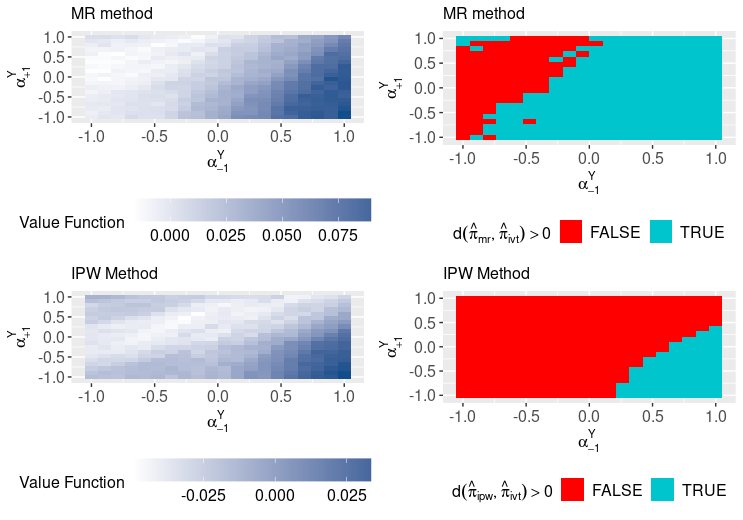}
    \caption{ENGAGE study: Sensitivity analysis of the difference between the proposed methods (MR and IPW) and the IVT method. The sensitivity parameters are $(\alpha_{+1}^Y, \alpha_{-1}^Y)$. The probability of compliance is fixed at 0.30. The left panel shows the magnitude of the difference between the proposed and IVT methods  (i.e., $d(\hat\pi_{mr},\hat\pi_{ivt})$ and $d(\hat\pi_{ipw},\hat\pi_{ivt})$) and the darker color the better performance of our methods. The right panel shows the area on the grid in which our methods outperform the IVT method where the red color indicates $d(\hat\pi_{mr},\hat\pi_{ivt})<0$ or $d(\hat\pi_{ipw},\hat\pi_{ivt})<0$.}% The dark blue or teal color indicates the area where the proposed method outperformed the IVT method.}  
    \label{fig:compareEric}
\end{figure}

\begin{figure}
    \centering
    \includegraphics[scale=.5]{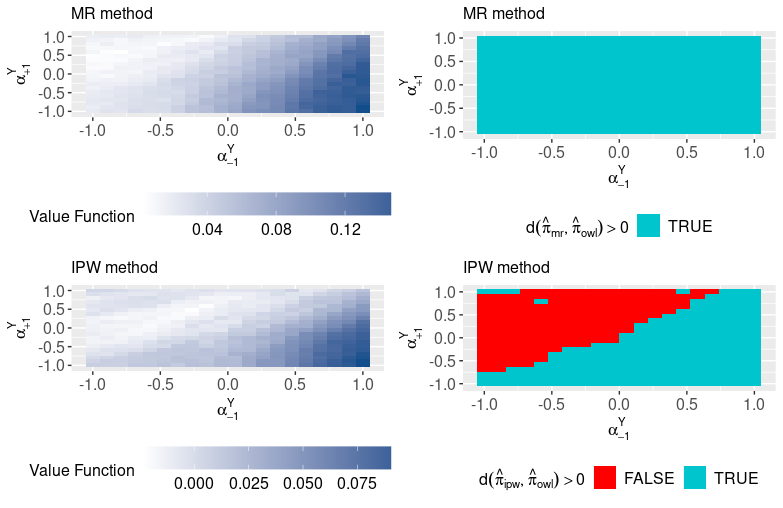}
    \caption{ENGAGE study: Sensitivity analysis of the difference between the proposed methods (MR and IPW) and the OWL method. The sensitivity parameters are $(\alpha_{+1}^Y, \alpha_{-1}^Y)$. The probability of compliance is fixed at 0.30. The left panel shows the magnitude of the difference between the proposed  and OWL methods ($d(\hat\pi_{mr},\hat\pi_{owl})$ and $d(\hat\pi_{ipw},\hat\pi_{owl})$) and the darker color the better performance of our methods). The right panel shows the area on the grid in which our methods outperform the OWL method where the red color indicates $d(\hat\pi_{mr},\hat\pi_{owl})<0$ or $d(\hat\pi_{ipw},\hat\pi_{owl})<0$.}
    \label{fig:compareOWL}
\end{figure}

\section{Discussion}
 
The proposed sensitivity analysis approach requires a model for the probability of being a complier given the outcome and observed covariates. As the number of covariates increases, the sensitivity analysis results become infeasible due to the increased number of sensitivity parameters.   To address this issue, we propose the following potential solutions. One approach is to set the parameters corresponding to covariates or a specific subset of the covariates to zero. Another approach is to use principal component regression (PCR) to reduce the dimension of the covariates. A third option is to model the probability of compliers as a function of the outcome $Y$ and the  residual of $Y$, i.e., $Y - E[Y|X]$. The residual is included to capture the effect of unmeasured confounders.

Our  methods can be  generalized in several directions. First, we can extend our method to dynamic treatment regimes, in which we identify a sequence of decision rules \citep{robins2004optimal, zhao2015new}. Second, although the treatment assigned is binary in this project, it would be interesting to develop a method that can identify the optimal dose instead of the optimal treatment \citep{zhou2021parsimonious, chen2016personalized}. Finally, in practice, it is possible that we have weak instruments \citep{hansen2006many}.  Therefore, understanding how weak instruments impact the estimate of the optimal treatment regimes or how to construct a robust approach to weak instruments  is an exciting problem \citep{ertefaie2018quantitative}.

%%%%%%%%%%%%%%%%%%%%%%%%%%%%%%%%%%%%%%%%%%%%%%
%% Example with multiple Appendixes:        %%
%%%%%%%%%%%%%%%%%%%%%%%%%%%%%%%%%%%%%%%%%%%%%%

%%%%%%%%%%%%%%%%%%%%%%%%%%%%%%%%%%%%%%%%%%%%%%
%% Acknowledgements                         %%
%% should be provided in the                %%
%% Acknowledgements section.                %%
%%%%%%%%%%%%%%%%%%%%%%%%%%%%%%%%%%%%%%%%%%%%%%

%%%%%%%%%%%%%%%%%%%%%%%%%%%%%%%%%%%%%%%%%%%%%%
%% Funding information, if any,             %%
%% should be provided in the                %%
%% funding section.                         %%
%%%%%%%%%%%%%%%%%%%%%%%%%%%%%%%%%%%%%%%%%%%%%%
\begin{funding}

This work was supported in part by R01DA048764, R61NS120240, and \\ R33NS120240 from the National Institutes of Health.

\end{funding}

%% if your bibliography is in bibtex format, uncomment commands:
\bibliographystyle{imsart-nameyear} % Style BST file (imsart-number.bst or imsart-nameyear.bst)
\bibliography{Binary-bib}
%%%%%%%%%%%%%%%%%%%%%%%%%%%%%%%%%%%%%%%%%%%%%%
%% Supplementary Material, including data   %%
%% sets and code, should be provided in     %%
%% {supplement} environment with title      %%
%% and short description. It cannot be      %%
%% available exclusively as external link.  %%
%% All Supplementary Material must be       %%
%% available to the reader on Project       %%
%% Euclid with the published article.       %%
%%%%%%%%%%%%%%%%%%%%%%%%%%%%%%%%%%%%%%%%%%%%%%
\begin{center}
{\large\bf SUPPLEMENTARY MATERIAL}
\end{center}

\section{Proof of Theorems \ref{main-theo} and \ref{th:val}}
\label{appendixa}
\subsection{Proof of Theorem \ref{main-theo}}
 By definition, we have the value function 
 {\footnotesize
 \begin{align*}
     \mathcal{V}^c(\pi) =& E\left[ I\{\pi(X) = 1\}Y(1) + I\{\pi(X = - 1)\}Y(-1)|PS = S4 \right]\\
     =&E\left[ E\big[I\{\pi(X) = 1\}Y(1)|PS = S4, X \big] + E\big[I\{\pi(X) = -1\}Y(-1) |PS =S4, X \big]\right] \\
     =& E\Big[I\{\pi(X) = 1\} E\big[Y|Z = 1, A = 1, PS = S4, X\big] \\
     &+ I\{\pi(X) = -1\} E\big[Y|Z = -1, A = -1, PS =  S4, X\big] \Big].
 \end{align*}
}
We can verify that

\begin{equation*}
    \label{eq:great}
    f(y| Z = 1, A =1, X) = E\Big[f(y|Z, A, X)\frac{A(A+Z)(A+1)}{4f(Z,A|X)} \Big|X \Big],
\end{equation*}

and 

\begin{equation*}
    f(y| Z = -1, A = -1, X) = E\Big[f(y|Z, A, X)\frac{A(A+Z)(1-A)}{4f(Z,A|X)} \Big|X \Big].
\end{equation*}

We have

{\small
\begin{align*}
    E\big[&Y|Z=1, A = 1, PS = S4,X\big] = \int \frac{yw_{\alpha_{+1}}(1,1,X,y)}{\gamma(1,1,X)}f(y|Z = 1, A = 1, X)dy \\
    =& \frac{1}{\gamma(1,1,X)} \int yw_{\alpha_{+1}}(1,1,X,y) E[f(y|Z,A,X)\frac{A(A+Z)(A+1)}{4f(Z,A|X)} \mid X]dy\\
    =& \frac{1}{\gamma(1,1,X)} \int \int yw_{\alpha_{+1}}(1,1,X,y)\frac{a(a+z)(a+1)}{4f(z,a|X)}f(y|z,a,X)f(z,a|X)d(za)dy\\
    =& \frac{1}{\gamma(1,1,X)} \int \frac{a(a+z)(a+1)yw_{\alpha_{+1}}(1,1,X,y)}{4f(z,a|X)}f(y,z,a|X)d(yza) \\
    =& \frac{1}{\gamma(1,1,X)}E\Big[\frac{A(A+Z)(A+1)Yw_{\alpha_{+1}}(1,1,X,Y)}{4f(A,Z|X)}|X \Big] \\
    =& E\Big[\frac{A(A+Z)(A+1)Yw_{\alpha_{+1}}(1,1,X,Y)}{4f(A,Z|X)\gamma(1,1,X)}|X \Big]. 
\end{align*}
}

Similarly, we have 

{\small
\begin{equation*}
    E[Y|Z=-1, A = -1, PS = S4,X] = E\Big[\frac{A(A+Z)(1-A)Yw_{\alpha_{-1}}(-1,-1,X,Y)}{4f(A,Z|X)\gamma(-1,-1,X)} \mid X \Big]
\end{equation*}
}
Hence, 

\begin{align*}
    \mathcal{V}^c(\pi)=& E\Big[ E\Big[ \frac{I\{\pi(X) = 1\}A(A+Z)(1+A)Yw_{\alpha_{+1}}(1,1,X,Y)}{4f(A,Z|X)\gamma(1,1,X)} \\
    &+ \frac{I\{\pi(X) = -1\}A(A+Z)(1-A)Yw_{\alpha_{-1}}(-1,-1,X,Y) }{4f(A,Z|X)\gamma(-1,-1,X)} \mid X \Big]\Big]\\
    =& E\Big[ E\Big[ \frac{I\{\pi(X) = Z\}A(A+Z)Yw_{\alpha_Z}(A,Z, X, Y)}{\gamma(A,Z,X)f(A,Z|X)} \mid X \Big] \Big]\\
    =& E\Big[\frac{I\{\pi(X) = Z\}A(A+Z)Yw_{\alpha_Z}(A,Z, X, Y)}{\gamma(A,Z,X)f(A,Z|X)}\Big]. 
\end{align*}

\subsection{Proof of Theorem \ref{th:val}} \label{supp:proofth2}

 Defining $\frac{0}{0} = 0$, the value function for complier is

$$\mathcal{V}^c(\mathcal{D}) =  E\Big[\main \Big].$$

 To find the efficient influence function for $\mathcal{V}(\mathcal{D})$, we need to find the canonical gradient $G$ for $\mathcal{V}^c(\mathcal{D})$ in the nonparametric model $\mathcal{M}_{np}$. In other words, we want to find $G$ such that $E[G] = 0$ and for any one-dimensional parametric submodel of $\mathcal{M}_{np}$, 
 
  \begin{equation*}
     \frac{\partial}{\partial t}\mathcal{V}^c(\mathcal{D})_t \Bigg|_{t=0} = E\left[GS(\mathcal{O};t) \right]\Bigg|_{t=0},
 \end{equation*}
 where $S(\mathcal{O},t) = \partial log f(\mathcal{O},t)/\partial t$. We have
{\small
\begin{align*}
     \frac{\partial}{\partial t} &\mathcal{V}^c_t(\mathcal{D}) \mid_{t = 0} = \frac{\partial}{\partial t} E_t\Big[\frac{A(Z + A)Yw_{\alpha_Z}(A,Z, X, Y)I\{\pi(X) = Z\}}{2\gamma(A,Z,X)f(Z|X)f(A|Z,X)}\Big] \Big |_{t = 0} \\ 
     =& \int \frac{A(Z + A)Yw_{\alpha_Z}(Z,A, Y)I\{\pi(X) = Z\}}{2\gamma(A,Z,X)f(Z|X)f(A|Z,X)} \deri f_t(o)\mid_{t=0}do\\
     & + \int \deri \Big(\main\Big) \Big|_{t=0} f(o)do \\
     =& E\Big[\main S(\mathcal{O})\Big]\\
     & + E\Big[\deri \Big(\main \Big) \Big]\Big|_{t=0} \\
     =&  E\left[\main S(\mathcal{O})\right] \\
     -& E\left[\frac{A(Z + A)Yw_{\alpha_Z}(A,Z, X, Y)I\{\pi(X) = Z\}}{2\gamma(A,Z,X)^2f(Z|X)^2f(A|Z,X)^2}\left\{\gamma(A,Z,X)f(A \mid Z,X)\deri f_t(Z|X) \right. \right. \\
     &  \left. \left. + f_t(Z|X)f(A|Z,X) \deri \gamma_t(A,Z,X)+\gamma(A,Z,X)f(Z|X) \deri f_t(A\mid Z,X)\right\}\right] \Big|_{t=0} \\
     =&  E\left[\main S(\mathcal{O})\right]\\ 
     &  - E\left[\frac{ZAYw_{\alpha_Z}(A,Z, X, Y)I\{\pi(X) = Z\}}{\gamma(A,Z,X)^2f(Z|X)^2f(A|Z,X)^2}\left\{\gamma(A,Z,X)f(A |Z,X)\deri f_t(Z|X) \right. \right. \\
     &  \left. \left. + f_t(Z|X)f(A|Z,X) \deri \gamma_t(A,Z,X)+\gamma(Z,A, X)f(Z|X) \deri f_t(A\mid Z,X)\right\}\right] \Big |_{t=0} \\
     =& (I) - (II) - (III) - (IV).
\end{align*}
}
We have will look further into the term (II), (III), (IV)

\begin{align*}
    (II) =& E\Big[\frac{ A(Z + A)Yw_{\alpha_Z}(A,Z, X, Y)I(\pi(X) = Z)}{2\gamma(A,Z,X)f(Z|X)^2f(A|Z,X)} \deri f_t(Z|X) \Big] \Big|_{t=0} \\
    =& E\Big[ \main \frac{\deri f_t(Z|X)}{f(Z|X)} \Big] \Big|_{t=0}\\
    =& E\Big[ \main S(Z|X) \Big] \\
    =& E\Big[E[\main|Z,X] S(Z|X) \Big] \\ 
    =& E\Big[ \Big(E[\main|Z,X]\\
    &- E[\sum_z \frac{A(z + A)Yw_{\alpha_Z}(A,z,X,Y) I\{\pi(X) = z\}}{2\gamma(A,z,X)f(A|z,X)}| Z = z, X]\Big) S(Z,X)\Big] \\
    =& E\Big[ \Big(E[\main|Z,X]\\
    &- E\big[\sum_z \frac{A(z + A)Yw_{\alpha_z}(A,z,X,Y) I\{\pi(X) = z\}}{2\gamma(A,z,X)f(A|z,X)}| Z = z, X\big]\Big) S(\mathcal{O}) \Big] \\
    =& E\Big[\Big(\frac{I\{\pi(X) = Z\}}{f(Z|X)}E\big[\frac{A(Z + A)Yw_{\alpha_Z}(Z,A, Y)}{2\gamma(A,Z,X)f(A|Z,X)} \mid Z,X \big]\\
    &- \sum_z I\{\pi(X) = z\} E\big[\frac{A(A+z)Yw_{\alpha_Z}(z,A, Y)}{2\gamma(A,z,X)f(A|z,X)}|Z = z, X\big] \Big) S(\mathcal{O}) \Big].
\end{align*}

{\small
\begin{align*}
    (III) &= E\Big[ \frac{A(Z+A)Yw_{\alpha_Z}(A,Z, X, Y)I\{\pi(X) = Z\}}{2\gamma(A,Z,X)^2f(Z|X)f(A|Z,X)}\deri \gamma(A,Z,X)\Big]\Big|_{t=0} \\
    =& E\Big[\frac{A(Z + A)(A+1)Yw_{\alpha_{+1}}(1,1,X,Y)I(\pi(X) = Z)}{4\gamma(1,1,X)^2f(Z|X)f(A|Z,X)} \deri \gamma(1,1,X)\Big]\Big|_{t=0}\\
    &+ E\Big[\frac{A(Z + A)(1 - A)Yw_{\alpha_{-1} }(-1,-1,X,Y)I\{\pi(X) = Z\}}{4\gamma(-1,-1,X)^2f(Z|X)f(A|Z,X)} \deri \gamma(-1,-1,X)\Big]\Big|_{t=0}. 
\end{align*}
}

{\small
\begin{align*}
    &E\Big[\frac{A(Z + A)(A+1)Yw_{\alpha_{+1}}(1,1,X,Y)I(\pi(X) = Z)}{4\gamma(1,1,X)^2f(Z|X)f(A|Z,X)} \deri \gamma(1,1,X)\Big]\Big|_{t=0} \\
    =&E\Big[E\big[\sum_a \sum_z\frac{a(z + a)(a+1)Yw_{\alpha_{+1}}(1,1,X,Y)I\{\pi(X) = z\}}{4\gamma(1,1,X)^2}|A = a, Z = z, X\big]\\
    &*\deri \gamma(1,1,X)\Big]\Big|_{t=0} \\
    =& E\Big[\frac{I(\pi(X) = 1)\delta(1,1,X)}{\gamma(1,1,X)}\deri \gamma(1,1,X)\Big]\Big|_{t=0} \\
    =& E\Big[\frac{I(\pi(X) = 1)\delta(1,1,X)}{\gamma(1,1,X)}E\Big[\frac{A(A+Z)(A+1)}{4f(Z,A|X)} \\
    &\hspace{4cm} *\big[w_{\alpha_{+1}}(1,1,X,Y) - \gamma(1,1,X)\big]S(\mathcal{O})\Big|X\Big] \Big] \\
    =& E\Big[\frac{I(\pi(X) = 1)\delta(1,1,X)}{\gamma(1,1,X)}\frac{A(A+Z)(A+1)}{4f(Z,A|X)}\big[w_{\alpha_{+1}}(1,1,X,Y) - \gamma(1,1,X)\big]S(\mathcal{O})  \Big] 
\end{align*}
}

Similarly, we have 
   \begin{align*}
   &E\Big[\frac{A(Z + A)(1-A)Yw_{\alpha_{-1} }(-1,-1,X,Y)I(\pi(X) = Z)}{4\gamma(-1,-1,X)^2f(Z|X)f(A|Z,X)} \deri \gamma(-1,-1,X)\Big]\Big|_{t=0} \\
   =& E\Big[\frac{I(\pi(X) = -1)\delta(-1,-1,X)}{\gamma(-1,-1,X)}\frac{A(A+Z)(1-A)}{4f(Z,A|X)}\\
   &\hspace{0.5cm}*[w_{\alpha_{-1} }(-1,-1,X,Y) - \gamma(-1,-1,X)]S(\mathcal{O}) \Big].
\end{align*}
Putting the two terms together, we have 
{\small
\begin{align*}
    (III) = E\Bigg[\frac{A(A+Z)I\{\pi(X) = Z\}\delta(A,Z,X)}{2\gamma(A,Z,X)f(Z|X)f(A|Z,X)}\{w_{\alpha_Z}(Z,A,X,Y) - \gamma(A,Z,Y)\}S(\mathcal{O}) \Bigg].
\end{align*}
}
{\small
\begin{align*}
     (IV) =& E\Big[\frac{A(Z + A)Yw_{\alpha_Z}(A,Z, X, Y)I(\pi(X) = Z)}{\gamma(A,Z,X)f(Z|X)f(A|Z,X)} \deri f_t(A|Z,X) \Big]\Big|_{t=0} \\
    =& E\Big[ \main \frac{\deri f_t(A|Z, X)}{f(A|Z, X)} \Big]\Big|_{t=0} \\
    =& E\Big[ \main S(A|Z, X) \Big] \\
    =& E\Big[E[\main|Z,A,X] S(A|Z,X) \Big] \\ 
    =& E\Big[ \Big(E[\main|Z,A,X]\\
    &- E[\sum_a  \frac{a(Z+a)Yw_{\alpha_Z}(a,Z,X) I(\pi(X) = Z)}{2\gamma(A,Z,X)f(Z|X)}|A =a, Z, X]\Big) S(A,Z,X)\Big] \\
    =& E\Big[ \Big(E[\main|A, Z,X]\\
    &- \sum_a \frac{ I(\pi(X) = Z)}{2\gamma(a,Z,X)f(Z|X)} E\{ a(Z+a)Yw_{\alpha_Z}(a,Z,X,Y) \mid A=a,Z,X\} \Big) S(\mathcal{O})  \Big] \\
    =& E\Big[\frac{A(Z+A)I(\pi(X) = Z)E[Yw_{\alpha_Z}(A,Z, X, Y)|A,Z,X]}{2\gamma(A,Z,X)f(Z|X)f(A|Z,X)}\\
    &- \sum_a  \frac{a(a+Z)I\{\pi(X) = Z\} E[Yw_{\alpha_Z}(a,Z,X,Y)|A=a, Z, X]}{2\gamma(a,Z,X)f(Z|X)}S(\mathcal{O}) \Big]. 
\end{align*}
}

Thus, the $\deri \mathcal{V}^c(\mathcal{O}) |_{t=0}$ is 

\begin{align*}
     &E\Big[\main S(\mathcal{O})\Big] \\ 
     &- E\Big[\Big(\frac{I\{\pi(X) = Z)\}}{f(Z|X)}E\big[\frac{A(Z + A)Yw_{\alpha_Z}(A,Z,X, Y)]}{2\gamma(A,Z,X)f(A|Z,X)} \mid Z,X \big]\\
     &- \sum_z I\{\pi(X) = z\} E\big[\frac{A(A+z)Yw_{\alpha_Z}(A,z,X,Y) }{2\gamma(z,A, X)f(A|z,X)}|Z = z, X\big] \Big) S(\mathcal{O}) \Big] \\
     &- E\Bigg[\frac{A(A+Z)I\{\pi(X) = Z\}\delta(A,Z,X)}{2\gamma(A,Z,X)f(Z|X)f(A|Z,X)}\{w_{\alpha_Z}(A,Z, X, Y) - \gamma(Z,A, Y)\}S(\mathcal{O}) \Bigg] \\
     &- E\Big[\frac{A(Z+A)I\{\pi(X)=Z\}E[Yw_{\alpha_Z}(A,Z, X, Y)|A,Z,X]}{2\gamma(A,Z,X)f(Z|X)f(A|Z,X)}\\
     &- \sum_a  \frac{a(a+Z)I\{\pi(X) = Z\} E[Yw_{\alpha_Z}(a,Z,X,Y)|A=a, Z, X]}{2\gamma(a,Z,X)f(Z|X)}S(\mathcal{O}) \Big]. 
\end{align*}

Accordingly, the canonical gradient $\xi_V(O,Q,\gamma,\kappa,f_A,f_Z)$ is
\begin{align*}
    \xi_V(&O,Q,\gamma,\kappa,f_A,f_Z)= \main  \\&- \Big\{\frac{I(\pi(X) = Z)}{f(Z|X)}E\big[\frac{A(Z + A)Yw_{\alpha_Z}(A,Z, X, Y)}{2\gamma(A,Z,X)f(A|Z,X)} \mid Z,X \big] 
     \\&- \sum_z I(\pi(X) = z) E\big[\frac{A(A+z)Yw_{\alpha_z}(A,z,X,Y) }{2\gamma(A,z,X)f(A|z,X)}|Z = z, X\big] \Big\}\\
     &- \Big\{\frac{A(A+Z)I\{\pi(X) = Z\} \delta(A,Z,X)}{2\gamma(A,Z,X)f(Z|X)f(A|Z,X)}[w_{\alpha_Z}(A,Z, X, Y) - \gamma(A,Z,X)] \Big\}  \\ 
      &- \Big\{ \frac{A(Z+A)I\{\pi(X) = Z\}E[Yw_{\alpha_Z}(A,Z, X, Y)|A,Z,X]}{2\gamma(A,Z,X)f(Z|X)f(A|Z,X)}\\
      &- \sum_a  \frac{a(a+Z)I\{\pi(X) = Z\} E[Yw_{\alpha_Z}(a,Z,X,Y)|A=a, Z, X]}{2\gamma(a,Z,X)f(Z|X)} \Big\} - \mathcal{V}^c(\pi).
    % &- \Big(\frac{A(Z+A)I(\pi(X) = Z)\delta^Z}{2f(Z|X)p(A|Z,X)} - \sum_z I(\pi(X) = %z)\delta^z\Big)
\end{align*}

We let 
\begin{align*}
      &\kappa(Z,X) = E[\frac{A(Z + A)Yw_{\alpha_Z}(A,Z, X, Y)}{2\gamma(A,Z,X)f(A|Z,X)}|Z,X].
\end{align*}

Under some regularity condition, the nuisance parameters $\hat f(A|Z,X), \\
\hat f(Z|X), \hat \kappa(Z,X), \hat \gamma(A,Z,X)$, $\hat Q(A,Z,X)$ converge in probability to $f^*(A|Z,X), \\ f^*(Z|X), \kappa^*(Z,X), \gamma^*(A,Z,X)$, $ Q^*(A,Z,X)$. 

We propose the following three models for our multiply robust estimator. Under each or union of those models, $E\left[\xi(O,Q,\gamma^*,\kappa^*,f^*_A,f^*_Z)\right] = \mathcal{V}^c(\pi) $
\begin{align*}
    &\mathcal{M}^\dagger_1: f(A|Z,X), \gamma(A,Z,X), f(Z|X) \text{ are correctly specified.} \\
    &\mathcal{M}^\dagger_2: f(A|Z,X), \gamma(A,Z,X), \kappa(Z,X) \text{ are correctly specified.} \\
    &\mathcal{M}^\dagger_3: \gamma(A,Z,X), f(Z|X), Q(A,Z,X) \text{ are correctly specified.}
\end{align*}

Under the model $\mathcal{M}^\dagger_1$, $f^*(A|Z,X) = f(A|Z,X), f^*(Z|X) = f(Z|X)$ and $\gamma^*(A,Z,X) = \gamma(A,Z,X)$, but $\kappa^*(Z,X) \ne \kappa(Z,X)$ and \\ $Q^*(A,Z,X) \ne Q(A,Z,X)$.

\begin{align*}
     &E\Bigg[\main \Bigg] \\
     &- E\left[ \Big\{\frac{I(\pi(X) = z)}{\gamma(A,Z,X)f(Z|X)}\kappa^{*}(Z,X) -  \sum_z\frac{I(\pi(X) = z)}{\gamma(A,z,X)} \kappa^{*}(z,X) \Big\} \right] \\
     &- E \left[ \frac{\hz I(\pi(X) = Z)\delta^*(A,Z,X)}{2\gamma(A,Z,X)f(A,Z|X)}\left\{w_{\alpha_Z}(A,Z, X, Y) - \gamma(A,Z,X)\right\} \right] \\ 
     &- E\Big[\Big\{ \frac{A(Z+A)I(\pi(X) = Z)Q^*(A,Z,X)}{2\gamma(A,Z,X)f(Z|X)f(A|Z,X)} \\
     &\hspace{1cm}- \sum_a \frac{Z(a + Z)I\{\pi(X) = Z\}Q^*(a, Z,X)}{2\gamma(A,Z,X)f(Z|X)}\Big\}\Big] \\ 
     =& E\Big[\main \Big] = E[Y_{\pi(X)}] = \mathcal{V}^c(\pi)
\end{align*}

Under the model $\mathcal{M}^\dagger_2$, $f^*(A|Z,X) =f(A|Z,X), \kappa^*(Z,X) = \kappa(Z,X)$, \\ $\gamma^*(A,Z,X) = \gamma(A,Z,X)$, but $Q^*(A,Z,X) \ne Q(A,Z,X)$ and $f^*(Z|X) \ne f(Z|X)$

\begin{align*}
     &E\left[ \frac{A(Z + A)Yw_{\alpha_Z}(A,Z, X, Y)I(\pi(X) = Z)}{2\gamma(A,Z,X)f^*(Z|X)f(A|Z,X)} \right] \\
     &- E\left[ \Big\{\frac{I(\pi(X) = Z)}{\gamma(A,Z,X)f^*(Z|X)}\kappa(Z,X) - \sum_z \frac{I(\pi(X) = z)}{\gamma(A,z,X)} \kappa(z,X) \Big\} \right] \\
     &- E\left[ \frac{\hz I(\pi(X) = Z)\delta^*(A,Z,X)}{2\gamma(A,Z,X)f^*(Z|X)f(A|Z,X)}\left\{w_{\alpha_Z}(A,Z, X, Y) - \gamma(A,Z,X)\right\} \right]\\
     &- E\Big[ \Big( \frac{A(Z+A)I(\pi(X) = Z)Q^*(A,Z,X)}{2\gamma(A,Z,X)f^*(Z|X)f(A|Z,X)} \\
     &\hspace{1cm}- \sum_a \frac{Z(a + Z)I\{\pi(X) = Z\}Q^*(a, Z,X)}{2\gamma(A,Z,X)f^*(Z|X)}\Big) \Big] \\ 
     =& E\Bigg[\frac{A(Z + A)Yw_{\alpha_Z}(A,Z, X, Y)I(\pi(X) = Z)}{2\gamma(A,Z,X)f^*(Z|X)f(A|Z,X)}  - \frac{I(\pi(X) = Z)}{\gamma(A,Z,X)f^*(Z|X)}\kappa(Z,X)\Bigg]\\
     &+ E\Bigg[\sum_z \frac{I(\pi(X) = z)}{\gamma(A,z,X)} \kappa(z,X) \Big\}\Bigg] \\
     &-  E\Bigg[\frac{\hz I(\pi(X) = X)\delta^*(A,Z,X)}{2\gamma(A,Z,X)f^*(Z|X)f(A|Z,X)}\{w_{\alpha_Z}(A,Z, X, Y) - \gamma(A,Z,X)\}\Bigg]\\
     &- E\Bigg[\frac{A(Z+A)I(\pi(X) = Z)Q^*(A,Z,X)}{2\gamma(A,Z,X)f^*(Z|X)f(A|Z,X)} \\
     &\hspace{1cm}- \sum_a \frac{a(Z + a)I\{\pi(X) = Z\}Q^*(a, Z,X)}{2\gamma(A,Z,X)f^*(Z|X)}\Bigg]\\
     =&  E\Bigg[\sum_z \frac{I(\pi(X) = z)}{\gamma(A,z,X)} \kappa(z,X) \Big\}\Bigg] = E[Y_{\pi(X)}] = \mathcal{V}^c(\pi).
\end{align*}

Under the model $\mathcal{M}^\dagger_3$, $f^*(Z|X) = f(Z|X), Q^*(A,Z,X) = Q(A,Z,X)$ and $\gamma^*(A,Z,X) = \gamma(A,Z,X)$, but $f^*(A|Z,X) \ne f(A|Z,X)$ and \\$\kappa^*(Z,X) \ne \kappa(Z,X)$

\begin{align*}
     &E\Bigg[\frac{A(Z+A)Yw_{\alpha_Z}(A,Z, X, Y)}{2\gamma(A,Z,X)f(Z|X)f^*(A|Z,X)} -\frac{A(Z+A)I(\pi(X) = Z)Q(A,Z,X)}{2\gamma(A,Z,X)f(Z|X)f^*(A|Z,X)} \Bigg] \\ 
     &- E\Bigg[\frac{I(\pi(X) = Z)}{\gamma(A,Z,X)f(Z|X)}\kappa^*(Z,X) - \sum_z \frac{I(\pi(X) = z)}{\gamma(A,z,X)} \kappa^*(z,X) \Bigg] \\
     &-  E\Bigg[\frac{\hz I(\pi(X) = Z)\delta(A,Z,X)}{2\gamma(A,Z,X)f(Z|X)f^*(A|Z,X)}\{w_{\alpha_Z}(A,Z, X, Y) - \gamma(A,Z,X)\}\Bigg] \\ 
     &+ E\Bigg[\sum_a \frac{Z(a + Z)I\{\pi(X) = Z\}Q(a, Z,X)}{2\gamma(a,Z,X)f(Z|X)}\Bigg] \\ 
     =& E\Bigg[ \sum_a \frac{Z(a + Z)I\{\pi(X) = Z\}Q(a, Z,X)}{2\gamma(a,Z,X)f(Z|X)}\Bigg] = E[Y_{\pi(X)}] = \mathcal{V}^c(\pi).
\end{align*}

\section{Multiply robust \texorpdfstring{$\Delta$}{Delta}}
\label{sec:supp-Deltamr}
\subsection{The estimator}

The parameter $\Delta(X) = E\big[Y(1) - Y(-1)\mid X, PS = S4\big]$ is not pathwise differentiable, which makes the construction of a multiply robust estimator challenging. However, the expected value of $\Delta(X)$ with respect to the distribution of $X$ (i.e., $\psi = E\{\Delta(X) \mid PS = S4\}$) is pathwise differentiable. Hence the canonical gradient of $\psi$ exists in nonparametric models, which is  the corresponding efficient influence function \citep{bickel1993efficient}. This is helpful because  an intuitive candidate for the unknown function $\Delta(X)$ would be a component
of the efficient influence function for  $\psi$. 

%One way to construct a multiply robust estimator for  The parameter $\Delta(X)$ is not pathwise differentiable.  Let $\psi = E\{\Delta(\bX) \mid PS = S4\}$ denote the compliers average treatment effect. An intuitive candidate for the unknown function $\Delta(\bX)$ would be a component of the efficient influence function for  $\psi$. The functional parameter $\psi$ is pathwise differentiable and, in nonparametric models, the corresponding efficient influence function $\phi_{\psi}(O,..)$ will be multiply robust. That is $E\{ \phi_{\psi}(O,..) \}=0$ in the union of certain nuisance models. 
 
\begin{theorem} \label{th:phi}
   Under a non-parametric model, the efficient influence function  for $\psi=E\{\Delta(X) \mid PS = S4\}$ is given by
\begin{align*}
 \phi_{\psi}(\mathcal{O},Q,\gamma,\delta,f) =&\frac{ZA(A+Z)}{2\gamma(A,Z,X)f(A,Z|X)}\Big[Yw_{\alpha_Z}(A,Z, X, Y) - Q(A,Z,X) \\
 &- \delta(A,Z,X)\big\{w_{\alpha_Z}(A,Z, Y) - \gamma(A,Z,X)\big\} \Big]+ \Delta(X)-\psi,
\end{align*}
  where %$\gamma^{S4}_z(X) = E\{w^Z(Y, \alpha_z)|\bX,Z=z \}$ and $Q^z(\bX) = E\{Yw^Z(Y, \alpha_z)|\bX,Z=z\}$. We note that 
  $\delta(A,Z,X) = \frac{Q(A,Z,X)}{\gamma(A,Z,X)}$.
\end{theorem}

 \begin{proof}
 To find the efficient influence function for $\Delta$, we need to find the canonical gradient $G$ for $\Delta$ in the nonparametric model $\mathcal{M}_{np}$. In other words, we want to find $G$ such that $E[G] = 0$ and for any one-dimensional parametric submodel of $\mathcal{M}_{np}$, 
 
 \begin{equation*}
     \frac{\partial}{\partial t}\Delta_t \Bigg|_{t=0} = E\left[GS(\mathcal{O};t) \right]\Bigg|_{t=0},
 \end{equation*}
 
 where $S(\mathcal{O},t) = \partial log f(\mathcal{O},t)/\partial t$. We have

 {\small
 \begin{align*}
      \frac{\partial}{\partial t}\Delta_t \Bigg|_{t=0} 
      =& \frac{\partial}{\partial t} E_t[E_t[Y| Z = 1, PS = S4, X] - E_t[Y| Z = -1, PS =4, X] \Bigg|_{t = 0}\\
      =& \frac{\partial}{\partial t}  \int \left[ \int \frac{yw_{\alpha_{+1}}(1,1,X,y)}{\gamma_t(1,1,X)}f_t(y| Z = 1, A = 1, X)dy \right.\\
      & \left.- \int \frac{yw_{\alpha_{-1}}(-1,-1,X,y)}{\gamma_t(-1,-1,X)}f_t(y| Z = -1, A = -1, X)dy \right]f_t(X) dX \Bigg|_{t = 0}\\
      =& \frac{\partial}{\partial t} E_t \left[\frac{\int yw_{\alpha_{+1}}(1,1,X,y)f_t(y|A = 1, Z = 1, X)dy }{\gamma_t(1,1,X)} \right.\\
      &\left.- \frac{\int yw_{\alpha_{-1}}(-1,-1,X,y)f_t(y|A = -1, Z = -1, X)dy }{\gamma_t(-1,-1,X)} \right] \Bigg|_{t = 0}
 \end{align*}
} 
 We let 

\begin{align*}
    &Q_t(1,1,X) = \int yw_{\alpha_{+1}}(1,1,X,y)f_t(y|A = 1, Z = 1, X)dy\\
    &Q_t(-1,-1,X) = \int yw_{\alpha_{-1}}(-1,-1,X,y)f_t(y|A = -1, Z = -1, X)dy
\end{align*}

 We continue
 
\begin{align*}
    =& \frac{\partial}{\partial t} E_t \Big[\frac{Q_t(1,1,X)}{\gamma_t(1,1,X)} -  \frac{Q_t(-1,-1,X)}{\gamma_t(-1,-1,X)}\Big] \Bigg|_{t=0}\\
    =& \int \Big(\frac{Q(1,1,X)}{\gamma(1,1,X)} -  \frac{Q(-1,-1,X)}{\gamma(-1,-1,X)}\Big)\frac{\partial}{\partial t}f_t(x)|_{t=0}dx \\
    &+ \int f(x) \frac{\partial}{\partial t}\Big(\frac{Q_t(1,1,X)}{\gamma_t(1,1,X)} -  \frac{Q_t(-1,-1,X)}{\gamma_t(-1,-1,X)}\Big)\big|_{t=0}dx\\
    =& \int \Big(\frac{Q(1,1,X)}{\gamma(1,1,X)} -  \frac{Q(-1,-1,X)}{\gamma(-1,-1,X)}\Big)\frac{\partial}{\partial t}\log f_t(x)|_{t=0}f(x)dx\\
    &+ E\Big[\frac{\partial}{\partial t}\Big(\frac{Q_t(1,1,X)}{\gamma_t(1,1,X)} -  \frac{Q_t(-1,-1,X)}{\gamma_t(-1,-1,X)}\Big)|_{t=0}\Big]  \\
    =& E\Big[\Big(\frac{Q(1,1,X)}{\gamma(1,1,X)} -  \frac{Q(-1,-1,X)}{\gamma(-1,-1,X)}\Big)S(X)  \Big] \\
    &+ E\Big[\frac{\partial}{\partial t}\Big(\frac{Q_t(1,1,X)}{\gamma_t(1,1,X)} -  \frac{Q_t(-1,-1,X)}{\gamma_t(-1,-1,X)}\Big)\Big]\Big|_{t=0} \\
    =& E\Big[\Big(\frac{Q(1,1,X)}{\gamma(1,1,X)} -  \frac{Q(-1,-1,X)}{\gamma(-1,-1,X)}\Big)S(\mathcal{O})  \Big] \\
    &+ E\Big[\frac{\partial}{\partial t}\Big(\frac{Q_t(1,1,X)}{\gamma_t(1,1,X)} \Big)\Big]\Big|_{t=0} - E\Big[\frac{\partial}{\partial t} \Big(\frac{Q_t(-1,-1,X)}{\gamma_t(-1,-1,X)}\Big)\Big]\Big|_{t=0} \\
    =& (I) + (II) + (III)
\end{align*}

First, we consider the term

\begin{align*}
    (II) = E\Big[\frac{\gamma(1,1,X)\frac{\partial}{\partial t}Q_t(1,1,X) - Q(1,1,X)\frac{\partial}{\partial t}\gamma_t(1,1,X)}{\big[\gamma(1,1,X)\big]^2}\Big]\Big|_{t=0}.
\end{align*}
    
We can verify that

$$f(Y| Z = 1, A =1, X) = E\Big[f(Y|Z, A, X)\frac{A(A+Z)(A+1)}{4f(Z,A|X)} \Big|X \Big]$$

and 

$$f(Y| Z = -1, A = -1, X) = E\Big[f(Y|Z, A, X)\frac{A(A+Z)(1-A)}{4f(Z,A|X)} \Big|X \Big]$$

Now we have 
{\small
\begin{align*}
    \frac{\partial}{\partial t} &Q_t(1,1,X)\Big|_{t=0} = \int yw_{\alpha_{+1}}(1,1,X,y) \frac{\partial}{\partial t} f_t(y|A =1, Z = 1, X)\big|_{t=0}dy \\
    =& \int yw_{\alpha_{+1}}(1,1,X,y)E\Big[\frac{\partial}{\partial t}f_t(y|A, Z,X)\frac{A(A+Z)(A+1)}{4f(Z,A|X)}\Big|X\Big]\Big|_{t=0}dy\\
    =& \int yw_{\alpha_{+1}}(1,1,X,y)E\Big[S(y|A,Z,X)f(y|Z,A,X) \frac{A(A+Z)(A+1)}{4f(Z,A|X)} \Big|X \Big]dy \\
    =& \int \int yw_{\alpha_{+1}}(1,1,X,y)S(y|a,z,X)\frac{a(a+z)(a+1)}{4f(z|x)}f(y|z,a,X)f(z,a|X)dyd(az) \\
    =&  \int \int yw_{\alpha_{+1}}(1,1,X,y)(S(\mathcal{O}) -E[S(\mathcal{O})|a,z,X])\frac{a(a+z)(a+1)}{4f(z,a|X)}\\
    &\hspace{0.5cm}*f(y|z,a,X)f(z,a|X)dyd(az) \\
    =& E\Big[\frac{Yw_{\alpha_{+1}}(1,1,X,Y)A(A+Z)(A+1)}{4f(Z|X)}S(\mathcal{O}) \Big| X\Big]\\ 
    &- \int Q^+(a,z,X) \int S(\mathcal{O})\frac{a(a+z)(a+1)}{4f(a,z|X)}f(y|z,a,X)f(z,a|X)dyd(az) \\
    =& E\Big[\frac{Yw_{\alpha_{+1}}(1,1,X,Y)A(A+Z)(A+1)}{4f(A,Z|X)}S(\mathcal{O}) \Big| X\Big]\\
    &-  E\Big[\frac{Q^+(A,Z,X)A(A+Z)(A+1)}{4f(A,Z|X)}S(\mathcal{O}) \Big| X\Big] \\
    =& E\Big[\frac{A(A+Z)(A+1)}{4f(Z,A|X)}\big[Yw_{\alpha_{+1}}(1,1,X,Y) - Q^+(A ,Z ,X)\big]S(\mathcal{O})\Big|X\Big]
\end{align*}
}

where, for z = 1 or -1, 

\begin{align*}
    &Q^z(A,Z,X) = \int yw_{\alpha_z}(z,z,y)f(y|A,Z,X)dy\\
    &= E[Yw_{\alpha_z}(a,z,x,y)|A = Z = z, X] = Q(z,z,X)
\end{align*}

This is because $w_{\alpha_Z}(A,Z, X, Y) = 0$ for $A \ne Z$. 

Following the same step, we can calculate

{\small
\begin{align*}
    \frac{\partial}{\partial t} \gamma_t(1,1,X) \big |_{t=0} =& \int w_{\alpha_{+1}}(1,1,X,Y) \frac{\partial}{\partial t} f_t(y|A =1, Z = 1, X)\big |_{t=0}dy \\
    =& E\Big[\frac{A(A+Z)(A+1)}{4f(Z,A|X)}\big[w_{\alpha_{+1}}(1,1,X,Y) - \gamma(1,1,X)\big]S(\mathcal{O})\Big|X\Big]
\end{align*}
}

Therefore, we have 

\begin{align*}
    (II) =& E\Big[\frac{\gamma(1,1,X)\frac{\partial}{\partial t}Q_t(1,1,X) - Q(1,1,X)\frac{\partial}{\partial t}\gamma_t(1,1,X)}{\big[\gamma(1,1,X)\big]^2}\Big]\Big |_{t=0} \\
    =& E\Big[\frac{\frac{\partial}{\partial t}Q_t(1,1,X)}{\gamma(1,1,X)}\Big]\Big |_{t=0} - E\Big[ \frac{Q(1,1,X)\frac{\partial}{\partial t}\gamma_t(1,1,X)}{\big[\gamma(1,1,X)\big]^2} \Big]\Big |_{t=0} \\
    =& E\Big[\frac{A(A+Z)(A+1)}{4\gamma(1,1,X)f(A,Z|X)}[Yw_{\alpha_{+1}}(1,1,X,Y) - Q(1,1,X)]S(\mathcal{O})\Big]\\
    &- E\Big[\frac{Q(1,1,X)A(A + Z)(A + 1)}{4\gamma(1,1,X)^2f(A,Z|X)}[w_{\alpha_{+1}}(1,1,X,Y) - \gamma(1,1,X)S(\mathcal{O}) \Big]\\
    =& E\Big[\frac{A(A+Z)(A+1)}{\gamma(1,1,X)4f(A,Z|X)}[Yw_{\alpha_{+1}}(1,1,X,Y) - Q(1,1,X)]S(\mathcal{O})\Big]\\
    &- E\Big[\frac{Q(1,1,X)A(A + Z)(A + 1)}{4\gamma(1,1,X)^2f(A,Z|X)}\big[w_{\alpha_{+1}}(1,1,X,Y) - \gamma(1,1,X)\big]S(\mathcal{O}) \Big]
\end{align*}

Consider 

\begin{align*}
    (III) = E\Big[\frac{\gamma(-1,-1,X)\frac{\partial}{\partial t}Q_t(-1,-1,X) - Q(-1,-1,X)\frac{\partial}{\partial t}\gamma_t(-1,-1,X)}{[\gamma(-1,-1,X)]^2}\Big]\Big |_{t=0}
\end{align*}

Following the same argument in the previous part, we have

\begin{align*}
     \frac{\partial}{\partial t} Q_t(-1,-1,X)\Big |_{t=0} = E\Big[\frac{A(A+Z)(1-A)}{4f(A,Z|X)}\big[&Yw_{\alpha_{-1} }(-1,-1,X,Y)\\
     &- Q(-1,-1,X)\big]S(\mathcal{O})\Big|X\Big]
\end{align*}

\begin{align*}
    \frac{\partial}{\partial t} \gamma_t(-1,-1,X)\Big |_{t=0} =  E\Big[\frac{A(A+Z)(1-A)}{4f(A,Z|X)}\big[&w_{\alpha_{-1} }(-1,-1,X,Y)\\ &- \gamma(-1,-1,X)\big]S(\mathcal{O})\Big|X\Big]
\end{align*}

Hence, we have

{\small
\begin{align*}
    &(III) = E\Big[\frac{\gamma(-1,-1,X)\frac{\partial}{\partial t}Q_t(-1,-1,X) - Q(-1,-1,X)\frac{\partial}{\partial t}\gamma_t(-1,-1,X)}{[\gamma(-1,-1,X)]^2}\Big]\Big |_{t=0} \\
    =& E\Big[\frac{\frac{\partial}{\partial t}Q_t(-1,-1,X)}{\gamma(-1,-1,X)}\Big]\Big |_{t=0} - E\Big[ \frac{Q(-1,-1,X)\frac{\partial}{\partial t}\gamma_t(-1,-1,X)}{\gamma(-1,-1,X)^2} \Big]\Big |_{t=0} \\
    =& E\Big[\frac{A(A+Z)(1-A)}{4\gamma(-1,-1,X)f(Z,A|X)}[Yw_{\alpha_{-1} }(-1,-1,X,Y) - Q(-1,-1,X)]S(\mathcal{O})\Big]\\
    &- E\Big[\frac{Q(-1,-1,X)A(A + Z)(1 - A)}{4\gamma(-1,-1,X)^2f(A,Z|X)}[w_{\alpha_{-1} }(-1,-1,X,Y) - \gamma(-1,-1,X)]S(\mathcal{O}) \Big]\\
     =& E\Big[\frac{A(A+Z)(1-A)}{4\gamma(-1,-1,X)f(A,Z|X)}[Yw_{\alpha_{-1} }(-1,-1,X,Y) - Q(-1,-1,X)]S(\mathcal{O})\Big]\\
    &- E\Big[\frac{Q(-1,-1,X)A(A + Z)(1 - A)}{4\gamma(-1,-1,X)^2f(A,Z|X)}[w_{\alpha_{-1} }(-1,-1,X,Y) - \gamma(-1,-1,X)]S(\mathcal{O})\Big]
\end{align*}
}

The efficient curve for $\Delta(X)$ is 

\begin{align}
    E&IF_\Delta = \Big(\frac{Q(1,1,X)}{\gamma(1,1,X)} -  \frac{Q(-1,-1,X)}{\gamma(-1,-1,X)}\Big) \nonumber \\
    &+ \frac{A(A+Z)(A+1)}{4f(Z,A|X)}\Big(\frac{1}{\gamma(1,1,X)}\big[Yw_{\alpha_{+1}}(1,1,X,Y) - Q(1,1,X)\big] \nonumber \\
    & - \nonumber \frac{Q(1,1,X)}{\gamma(1,1,X)^2}\big[w_{\alpha_{+1}}(1,1,X,Y) - \gamma(1,1,X)\big]\Big) \nonumber \\ 
    &- \frac{A(A+Z)(1-A)}{4f(A,Z|X)}\Big(\frac{1}{\gamma(-1,-1,X)}\big[Yw_{\alpha_{-1} }(-1,-1,X,Y) - Q(-1,-1,X)\big] \nonumber \\
    &-  \frac{Q(-1,-1,X)}{\gamma(-1,-1,X)^2}\big[w_{\alpha_{-1} }(-1,-1,X,Y) - \gamma(-1,-1,X)\big] \Big).
\end{align}

We let 

\begin{align*}
    \delta(1,1,X) =& \frac{Q(1,1,X)}{\gamma(1,1,X)} \\ 
    \delta(-1,-1,X) =& \frac{Q(-1,-1,X)}{\gamma(-1,-1,X)}. 
\end{align*}

Then (9) becomes

\begin{align*}
   EIF_{\Delta} =& \frac{A(A+Z)(A+1)}{4f(A,Z|X)\gamma(1,1,X)}\Big(Yw_{\alpha_{+1}}(1,1,X,Y) - Q(1,1,X)\\
   &- \delta(1,1,X)\{w_{\alpha_{+1}}(1,1,X,Y) - \gamma(1,1,X)\} \Big) \\ 
   &- \frac{A(A+Z)(1-A)}{4f(A,Z|X)\gamma(-1,-1,X)}\Big(Yw_{\alpha_{-1} }(-1,-1,X,Y) - Q(-1,-1,X) \\
   &- \delta(-1,-1,X)\{w_{\alpha_{-1} }(-1,-1,X,Y) - \gamma(-1,-1,X)\} \Big) \\
   &+ \delta(1,1,X) - \delta(-1,-1,X) \\
   =& \frac{ZA(A+Z)}{2f(A,Z|X)\gamma(A,Z,X)}\Big(Yw_{\alpha_Z}(A,Z, X, Y)\\
   &- Q(A,Z,X) - \delta(A,Z,X)\{w_{\alpha_Z}(A,Z, X, Y) 
   -\gamma(A,Z,X)\}\Big)+ \Delta(X).
\end{align*}

\end{proof}

\subsection{Multiply robustness property of the estimator}
\label{robustdelta}
  We have multiply robust estimators for $\Delta$

\begin{align*}
     \Delta_{mr} =& \Delta(X) +\frac{ZA(A+Z)}{2f(A,Z|X)\gamma(A,Z,X)}\Big[Yw_{\alpha_Z}(A,Z, X, Y) - Q(A,Z,X)\\
     &- \delta(A,Z,X)\{w_{\alpha_Z}(A,Z, X, Y) -  \gamma(A,Z,X)\}\Big]\\
      =& \frac{Q(1,1,X)}{\gamma(1,1,X)} - \frac{Q(-1,-1,X)}{\gamma(-1,-1,X)}\\
      &+\frac{ZA(A+Z)}{2f(A,Z|X)\gamma(A,Z,X)}\Big[Yw_{\alpha_Z}(A,Z, X, Y) - Q(A,Z,X)\\
      &- \frac{Q(A,Z,X)}{\gamma(A,Z,X)}\{w_{\alpha_Z}(A,Z, X, Y) -  \gamma(A,Z,X)\}\Big]\\
\end{align*}

Our proposed models are
\begin{align*}
    &\mathcal{M}_1: Q(A,Z,X), \gamma(A,Z,X) \text{ are correctly specified} \\ 
    &\mathcal{M}_2: f(A,Z|X), \gamma(A,Z,X) \text{ are correctly specified}
\end{align*}

Under suitable regularity condition, the nuisance parameters $\hat Q(A,Z,X), \\
\hat f(A,Z|X), \hat \gamma(A,Z,X)$ converge in probability to $Q^*(A,Z,X), f^*(A,Z|X),\\ \gamma^*(A,Z,X)$. It is sufficient to show that $E[\Delta^*_{mr}] = \Delta$ in the union of $\mathcal{M}_1$ and $\mathcal{M}_2$.

Suppose that only $\mathcal{M}_1$ holds, $Q^*(A,Z,X) = Q(A,Z,X),\\ \gamma^*(A,Z,X) = \gamma(A,Z,X)$ but $f^*(A,Z \mid X) \ne f(A,Z \mid X)$.
{\small
\begin{align*}
     E&[\Delta_{mr}] = E[\delta(1,1,X)]-  E[\delta(-1,-1,X)] \\ 
     &+  E\frac{A(A+Z)(A+1)}{4f^*(A,Z|X)\gamma(1,1,X)}\Big( E[Yw_{\alpha_{+1}}(1,1,X,Y)|A,Z,X] - Q(1,1,X)\\
     &- \delta(1,1,X)\{ E[w_{\alpha_{+1}}(1,1,X,Y)|A,Z,X] - \gamma(1,1,X)\} \Big) \\ 
   &- E\Big[\frac{A(A+Z)(1-A)}{4f^*(A,Z|X)\gamma(-1,-1,X)}\Big(E[Yw_{\alpha_{-1} }(-1,-1,X,Y)|A,Z,X] - Q(-1,-1,X)\\
   &- \delta(-1,-1,X)\{E[w_{\alpha_{-1} }(-1,-1,X,Y)|A,Z,X] - \gamma(-1,-1,X)\} \Big)\Big] \\
   =& E[\delta(1,1,X) - \delta(-1,-1,X)] = E\left\{\Delta(X)\right\} = \Delta.
\end{align*}
}
Suppose that only $\mathcal{M}_2$ holds, $f^*(A,Z,X) = f(A,Z\mid X)$ but $Q^*(A,Z,X) \ne Q(A,Z,X)$, and $\gamma^*(A,Z,X) \ne \gamma(A,Z,X)$ 

{\small
\begin{align*}
     E&[\Delta_{mr}] = E[\delta^{*}(1,1,X)]-  E[\delta^{*}(-1,-1,X)]\\ 
     &+ E\frac{A(A+Z)(A+1)}{4f^(A,Z|X)\gamma(1,1,X)}\Big( E[Yw_{\alpha_{+1}}(1,1,X,Y)|A,Z,X] - Q^{*}(1,1,X)\\
     &- \delta^{*}(1,1,X)\{ E[w_{\alpha_{+1}}(1,1,X,Y)|A,Z,X] - \gamma(1,1,X)\} \Big) \\ 
   &- E\Big[\frac{A(A+Z)(1-A)}{4f(A,Z|X)\gamma(-1,-1,X)}\Big(E[Yw_{\alpha_{-1} }(-1,-1,X,Y)|A,Z,X] - Q^{-*}(X)\\
   &- \delta^{-*}(X)\{E[w_{\alpha_{-1} }(-1,-1,X,Y)|A,Z,X] - \gamma(-1,-1,X)\} \Big)\Big] \\
   =& E[\delta^{*}(1,1,X)] - E\Big[\frac{A(A+Z)(A+1)Q^{*}(1,1,X)}{4\gamma(-1,-1,X)f(A,Z|X)} \Big]\\
   &+ E\Big[\frac{A(A+Z)(A+1)Yw_{\alpha_{+1}}(1,1,X,Y)}{4\gamma(-1,-1,X)f(A,Z|X)} \Big] \\
   &- E[\delta^{*}(-1,-1,X)] + E\Big[\frac{A(A+Z)(1-A)Q^{*}(-1,-1,X)}{4\gamma(-1,-1,X)f(A,Z|X)} \Big]\\
   &- E\Big[\frac{A(A+Z)(1-A)Yw_{\alpha_{-1} }(-1,-1,X,Y)}{4\gamma(-1,-1,X)f(A,Z|X)} \Big]\\
   =& E\Big[\frac{A(A+Z)(A+1)Yw_{\alpha_{+1}}(1,1,X,Y)}{4\gamma(-1,-1,X)f(A,Z|X)} - \frac{A(A+Z)(1-A)Yw_{\alpha_{-1} }(-1,-1,X,Y)}{4\gamma(-1,-1,X)f(A,Z|X)} \Big]\\
   =& E[\delta(1,1,X) - \delta(-1,-1,X)] = E\left\{\Delta(X)\right\} = \Delta.
\end{align*}
}

\subsection{The remainder terms} \label{collary1}

 \begin{corollary} \label{cor:phi}
Let $\hat{\psi}= P_n \phi(O,\hat Q,\hat\gamma,\hat\delta,\hat f)$ and $\psi_0 = P_0 \phi(O,Q,\gamma,\delta,f)$. Under Assumptions \ref{assump:coreiv} - \ref{assump:mono},  we have
\[
\hat{\psi} - \psi_0 = P_n \phi_{\psi}(O,Q,\gamma,\delta,f) - \psi_0 + R(\mathcal{M}_1,\mathcal{M}_2),
\]
where   the remainder term  $R(\mathcal{M}_1,\mathcal{M}_2)$ is given below. Moreover, under Assumption \ref{assump:vrate},  $R(\mathcal{M}_1,\mathcal{M}_2)=o_p(n^{-1/2})$ which implies the asymptotic linearity of the estimator $\hat\psi$. 
 \end{corollary}

\begin{proof}
    
By definition, $\psi = P_0 \phi_{\psi}(\mathcal{O},Q,\gamma,\delta,f) = (P_0 - P_n)\phi_{\psi}(\mathcal{O},Q,\gamma,\delta,f) + P_n \phi_{\psi}(\mathcal{O},Q,\gamma,\delta,f)$ and $\hat \psi = P_n \hat \phi_{\psi}(\mathcal{O}, \hat Q,\hat\gamma,\hat\delta,\hat f) $

\begin{align*}
    \hat{\psi} - \psi =& \hat \psi - (P_0 - P_n)\phi_\psi(\mathcal{O},Q,\gamma,\delta,f) - P_n\phi_\psi(\mathcal{O},\hat Q,\hat \gamma,\hat \delta,\hat f)\\
    =& (P_n - P_0)\phi_\psi(\mathcal{O},Q,\gamma,\delta,f) + \hat \psi - P_n \phi_\psi(\mathcal{O}, Q,\gamma,\delta,f) 
\end{align*}

Hence, the remainder term is defined 

{\small
\begin{align*}
    R &= \hat \psi - P_n \phi_\psi(\mathcal{O},\hat Q,\hat \gamma,\hat \delta,\hat f)\\
    &= P_n \frac{\hzz\Big\{\big[Y - \qhat \big] - \delhat\big[w_{\alpha_z}(A,Z,X,Y) - \ghat\big] \Big\}}{2\ghat \fhat} \\
    &\hspace{0.25cm}+ P_n \hat \Delta(X) \\
    &\hspace{0.25cm}- P_n \frac{\hzz \Big\{Y - \q - \del\big[w_{\alpha_z}(A,Z,X,Y) - \g \big]  \Big\} }{2\g\f}\\
    &\hspace{0.25cm}- P_n \Delta( X)\\
    &\hspace{0.25cm}\pm P_n \frac{\hzz\Big\{Y - \q - \del\big[w_{\alpha_z}(A,Z,X,Y) - \g] \Big\}}{2\ghat \fhat}\\ 
    &= P_n \frac{\hzz}{2\ghat\fhat}\Big\{\big[\q - \qhat\big]\\
    &\hspace{0.25cm}+ \del\big[w_{\alpha_z}(A,Z,X,Y) - \g\big]\\
    &\hspace{0.25cm}- \delhat\big[w_{\alpha_z}(A,Z,X,Y) - \ghat \big] \Big\}
    + P_n \Big[\Delta( X) - \hat \Delta(  X) \Big] \\
    &\hspace{0.25cm}+ P_n \frac{\hzz}{2}\Big\{Y - \q - \del\big[w_{\alpha_z}(A,Z,X,Y) - \g \big] \Big\}\\
    &\hspace{0.25cm}*\Bigg[\frac{1}{\ghat \fhat}   - \frac{1}{\g \faz} \Bigg] \\
    &=  P_n \frac{\hzz}{2\ghat\fhat}\Big\{\big[\q - \qhat\big]\Big\} \\
    &\hspace{0.25cm}+ P_n \frac{\hzz}{2\ghat\fhat}\Big\{\del\big[w_{\alpha_z}(A,Z,X,Y) - \g\big]\\
    &\hspace{0.25cm}- \delhat\big[w_{\alpha_z}(A,Z,X,Y) - \ghat \big]\Big\} \\
    &\hspace{0.25cm} + P_n \Big[\hat \Delta(  X) - \Delta(  X) \Big] + o_p(n^{-1/2})  \\
    &= (I) + (II) + (III) + o_p(n^{-1/2})\\
\end{align*}
}
Consider the term (I)

\begin{align*}
   (I) =& P_n \frac{\hzz}{2\ghat \fhat}\Big[\q - \qhat\Big]\\
   &\pm P_n \frac{\hzz}{2\g \fhat}\Big[\q - \qhat \Big] \\
   =& P_n \frac{\hzz}{2\fhat} \frac{\Big[\q - \qhat\Big]\Big[\g - \ghat\Big]}{\g \ghat}\\
   &+ P_n \frac{\hzz}{2\g \fhat}\Big[\q - \qhat\Big] \\
   &\pm P_n \frac{\hzz}{2\g \faz}\Big[\q - \qhat\Big]\\ 
   =& P_n \frac{\hzz}{2\fhat} \frac{\Big[\q - \qhat\Big]\Big[\g - \ghat\Big]}{\g \ghat}\\
   &+ P_n \frac{\hzz}{2\g} \frac{\Big[\q - \qhat\Big]}{\fhat \faz}\Big[\faz - \fhat\Big] \\
   &+ P_n \frac{\hzz}{2\g \faz}\Big[\q - \qhat \Big] \\
\end{align*}

Consider term (II)
{\small
\begin{align*}
    &(II) = \frac{\hzz}{2\ghat \fhat}\Big\{\del\big[w_{\alpha_z}(A,Z,X,Y) - \g\big]\\
    &- \delhat\big[w_{\alpha_z}(A,Z,X,Y) - \ghat\big] \Big\}\\
    &\pm P_n \frac{\hzz}{2\ghat \fhat}\del\big[w_{\alpha_z}(A,Z,X,Y) - \ghat \big]\\
    =& P_n \frac{\hzz}{2\ghat \fhat}\Big\{ \del\Big[\ghat - \g \Big]\\
    &- \big[\delhat - \del\big]\big[w_{\alpha_z}(A,Z,X,Y) - \ghat \big] \Big\}\\
    =& P_n \frac{\hzz}{2\ghat \fhat}\Big\{ \del\Big[\ghat - \g \Big] + o_p(n^{-1/2})\\
    =& P_n \frac{\hzz}{\ghat \fhat}\Big\{ \del\Big[\ghat - \g \Big]+ o_p(n^{-1/2}) \\
    &\pm P_n \frac{\hzz}{2\ghat \faz}\del\big[\ghat - \g\big]\\ 
    =& P_n \frac{\hzz}{2\ghat}\del\big[\ghat - \g \big]\Big[\frac{\faz - \fhat}{\faz \fhat} \Big]\\
    &+ P_n \frac{\hzz}{2\ghat \faz}\del\big[\ghat - \g\big]+ o_p(n^{-1/2})\\
    &\pm P_n \frac{\hzz}{2\faz \g}\del\Big[\ghat - \g\Big]\\ 
    =& P_n \frac{\hzz}{2\ghat}\del\big[\ghat - \g\big]\Big[\frac{\faz - \fhat}{\faz \fhat} \Big] \\
    &+ P_n \frac{\hzz}{2\faz}\del\big[\ghat - \g\big]\Big[\frac{\g - \ghat}{\g\ghat} \Big]\\
    &+ P_n \frac{\hzz}{2\faz \g}\del\Big[\ghat - \g\Big]+ o_p(n^{-1/2}) \\
    %=& P_n \frac{Z}{\f \g}\del\Big[\ghat - \g\Big] + O_p(n^{-1/2}) \\ 
    =& P_0 \frac{\qp}{\gp}\Big[\frac{\ghatp - \gp}{\gp} \Big]\\
    &- P_0 \frac{\qn}{\gn}\Big[\frac{\ghatn - \gn}{\gn} \Big]\\
    &+ P_n \frac{\hzz}{2\ghat}\del\big[\ghat - \g\big]\Big[\frac{\faz - \fhat}{\faz \fhat} \Big] \\
    &+ P_n \frac{\hzz}{2\faz}\del\big[\ghat - \g\big]\Big[\frac{\g - \ghat}{\g\ghat} \Big] \\
    &+ o_p(n^{-1/2})\\
\end{align*}
}
Putting everything together
{\small
\begin{align*}
    &R = (I) + (II) + (III) \\
    =& P_n \frac{\hzz}{2\fhat} \frac{\Big[\q - \qhat\Big]\Big[\g - \ghat\Big]}{\g \ghat}\\
    &+ P_n \frac{\hzz}{2\g} \frac{\Big[\q - \qhat\Big]}{\fhat \faz}\Big[\faz - \fhat\Big]\\
   &+ P_n \frac{\hzz}{2\g \faz}\Big[\q - \qhat \Big]\\
   &+ P_0 \frac{\qp}{\gp}\Big[\frac{\ghatp - \gp}{\gp} \Big]\\
   &- P_0 \frac{\qn}{\gn}\Big[\frac{\ghatn - \gn}{\gn} \Big]\\
    &+ P_n \frac{\hzz}{2\ghat}\del\big[\ghat - \g\big]\Big[\frac{\faz - \fhat}{\faz \fhat} \Big] \\
    &+ P_n \frac{\hzz}{2\faz}\del\big[\ghat - \g\big]\Big[\frac{\g - \ghat}{\g\ghat} \Big]\\
    &+ P_n \Big[\hat \Delta(  X) - \Delta(  X) \Big]+ o_p(n^{-1/2})\\
    =& P_n \frac{\hzz}{2\fhat} \frac{\Big[\q - \qhat\Big]\Big[\g - \ghat\Big]}{\g \ghat}\\
    &+ P_n \frac{\hzz}{2\g} \frac{\Big[\q - \qhat\Big]}{\fhat \faz}\Big[\faz - \fhat\Big] \\
    &+ P_n \frac{\hzz}{2\ghat}\del\big[\ghat - \g\big]\Big[\frac{\faz - \fhat}{\faz\fhat} \Big] \\
    &+ P_n \frac{\hzz}{2\faz}\del\big[\ghat - \g\big]\Big[\frac{\g - \ghat}{\g\ghat} \Big]\\ 
    &+ \Big\{P_n \Big[\hat \Delta(  X) - \Delta(  X) \Big] +  P_n \frac{\hzz}{2\g \faz}\Big[\q - \qhat \Big]\Big\} \\
     &+ P_0 \frac{\qp}{\gp}\Big[\frac{\ghatp - \gp}{\gp} \Big]\\
     &- P_0 \frac{\qn}{\gn}\Big[\frac{\ghatn - \gn}{\gn} \Big]+ o_p(n^{-1/2}).
\end{align*}
}
Consider the term

{\small
\begin{align*}
     &P_n \frac{\hzz}{2\g \faz}\Big[\q - \qhat \Big] +P_n \Big[\hat \Delta(  X) - \Delta(  X) \Big] \\
    =& -P_0\Bigg[ \frac{\qhatp}{\gp} - \frac{\qhatn}{\gn} \Bigg] + P_0\Bigg[ \frac{\qhatp}{\ghatp} - \frac{\qhatn}{\ghatn} \Bigg] + o_p(n^{-1/2}) \\
    =& -P_0 \qhatp\Big[ \frac{\ghatp - \gp}{\gp \ghatp} \Big]\\
    &+ P_0 \qhatn\Big[ \frac{\ghatn - \gn}{\gn \ghatn} \Big]
    + o_p(n^{-1/2}).
\end{align*}
}
The remainder term becomes

{\small
\begin{align*}
   &R = P_n \frac{\hzz}{2\fhat} \frac{\Big[\q - \qhat\Big]\Big[\g - \ghat\Big]}{\g \ghat}\\
   &+ P_n \frac{\hzz}{2\g} \frac{\Big[\q - \qhat\Big]}{\fhat \faz}\Big[\faz - \fhat\Big] \\
    &+ P_n \frac{\hzz}{2\ghat}\del\big[\ghat - \g\big]\Big[\frac{\faz - \fhat}{\faz \fhat} \Big] \\
    &+ P_n \frac{\hzz}{2\faz}\del\big[\ghat - \g\big]\Big[\frac{\g - \ghat}{\g\ghat} \Big]\\ 
    &+ \Big\{-P_0 \qhatp\Big[ \frac{\ghatp - \gp}{\gp \ghatp} \Big]\\
    &+ P_0 \qhatn\Big[ \frac{\ghatn - \gn}{\gn \ghatn} \Big] \Big\} \\
     &+ P_0 \frac{\qp}{\gp}\Big[\frac{\ghatp - \gp}{\gp} \Big]\\
    & - P_0 \frac{\qn}{\gn}\Big[\frac{\ghatn - \gn}{\gn} \Big] + o_p(n^{-1/2})\\
     =& P_0 \frac{\hzz}{2\fhat} \frac{\Big[\q - \qhat\Big]\Big[\g - \ghat\Big]}{\g \ghat}\\
     &+ P_0 \frac{\hzz}{2\g} \frac{\Big[\q - \qhat\Big]}{\fhat \faz}\Big[\faz - \fhat\Big] \\
    &+ P_0 \frac{\hzz}{2\ghat}\del\big[\ghat - \g\big]\Big[\frac{\faz - \fhat}{\faz \fhat} \Big] \\
    &+ P_0 \frac{\hzz}{2\faz}\del\big[\ghat - \g\big]\Big[\frac{\g - \ghat}{\g\ghat} \Big]\\ 
    &+ P_0 \frac{\ghatp - \gp}{\gp}\Big[\delta(1,1, X) - \hat \delta(1,1, X) \Big]\\
    &+  P_0 \frac{\ghatn - \gn}{\gn}\Big[\delta(-1,-1,  X) - \hat \delta(-1,-1,  X) \Big]  + o_p(n^{-1/2}).
\end{align*}
}
Thus, the remainder R will go to $o_p(n^{-1/2})$ under the assumption

 $(\|f- \hat f \|_{P_0}+\|\gamma- \hat \gamma \|_{P_0})(\|\delta- \hat \delta \|_{P_0}+\|\gamma- \hat \gamma \|_{P_0}+\|Q- \hat Q \|_{P_0}) = o_p(n^{-1/2})$.
\end{proof}

\subsection{Proof of Corollary 2}
We let 
\begin{align*}
    \kappa(Z,X) =& E[\frac{A(Z + A)Yw_{\alpha_Z}(A,Z, X, Y)}{2\gamma(A,Z,X)f(A|Z,X)}|Z,X]\\
    \kappa'(X) =& \sum_z I\{\pi(X) = z\}E\Big[\frac{A(z+A)Yw_{\alpha_z}(A,z,X,Y)}{2\gamma(A,z,X)f(A|z,X)} \mid Z=z, X\Big]\\
    =& \sum_z I\{\pi(X) = z\}\kappa(Z=z,X) \\
    \theta(Z,X) =& \sum_a \frac{a(a + Z)}{2\gamma(a,Z,X)}E[Yw_{\alpha_z}(a,Z,X)|A=a,Z,X]\\
    =& \sum_a \frac{a(a + Z)}{2\gamma(a,Z,X)}Q(a,Z,X)
\end{align*}
We have
\begin{align*}
    \hat{\mathcal{V}}^c(\pi) - \mathcal{V}^c(\pi) =& \hat{\mathcal{V}}^c(\pi) + (P_n - P_o)\xi_V - P_n\xi_V. 
\end{align*}
We have the remainder term

{\small
\begin{align*}
    R =& \hat{\mathcal{V}}^c(\pi) - P_n\xi_V \\ 
    =& P_n \frac{\hz Y\w \I}{2\hg \hf \hp} \\
    &- \Big\{ P_n \frac{\I \hkappa}{\hf} - P_n \hkappap \Big\}\\ 
    &- P_n\frac{\hz\ \hdelta [\w - \hg]}{2\hg \hf \hp}\\
    &- \Big\{ P_n\frac{\hz Y\w \I \hat{Q}(A,Z,X)}{2\hg \hf \hp} - P_n\frac{\hat{\theta}(Z,X)}{\hf} \Big\}\\
    &- P_n \frac{\hz Y \w \I}{2\gamma(A,Z,X) \f \p}\\
    &+ P_n \Big\{ \frac{\I\kappa(Z,X)}{\f} - \kappa'(X)\Big\}\\ 
    &+ P_n\frac{\hz \I \delta(A,Z,X)[\w - \gamma(A,Z,X)]}{2\gamma(A,Z,X) \f \p}\\ 
    &+ \Big\{P_n\frac{\hz Y \w\I Q(A,Z,X)}{2\gamma(A,Z,X) \f \p} - P_n\frac{\theta(Z,X)}{\f} \Big\} \\
    =& P_n\Big\{ \frac{\hz}{2} Y \w \I\Big(\frac{1}{\hg \hf \hp}\\
    &- \frac{1}{\gamma(A,Z,X) \f \p} \Big)\Big\}\\
    &+ P_n \Big(\frac{\I\kappa(Z,X)}{\f} - \frac{\I \hat{\kappa}(Z,X)}{\hf}\Big)\\
    &+ P_n\Big\{\hkappap - \kappa'(X)\Big\} +P_n\Big\{\frac{\hat{\theta}(Z,X)}{\hf} - \frac{\theta(Z,X)}{\f}\Big\} \\ 
    &+ P_n\Big\{\frac{\hz \I \delta(A,Z,X) [\w - \gamma(A,Z,X)]}{2\gamma(A,Z,X) \f \p}\\ 
    &- \frac{\hz \I \hat{\delta}(A,Z,X)[\w - \hg]}{2\hg \hf \hp}\Big\}\\ 
    &+ P_n\Big\{\frac{\hz \I Q(A,Z,X)}{2\gamma(A,Z,X) \f \p} - \frac{\hz \hat{Q}(A,Z,X) \I}{2\hg \hf \hp}\Big\} \\
    =& (1) +(2.1) + (2.2) + (4.2) + (3) + (4.1). 
\end{align*}
}
We will keep term (1), (2.2) and simplify (2.1), (3), (4.1), (4.2)

{\small
\begin{align*}
    &(2.1) = P_n(\frac{\I}{\f}\kappa(Z,X) - \frac{\I}{\hf}\hkappa)\\ 
    =& P_n\Big[\frac{\I}{\f}\kappa(Z,X) - \frac{\I}{\hf}\hkappa \Big]\\
    &\pm P_n\Big[\frac{\I}{\hf}\kappa(Z,X) \Big]\\ 
    =& P_n\kappa(Z,X)\Big[\frac{\I}{\f} - \frac{\I}{\hf}\Big] - P_n\hkappa\frac{\I}{\hf}\\
    &+ P_n\frac{\I}{\hf}\kappa(Z,X)\\
    =& P_n\kappa(Z,X)\Big[\frac{\I}{\f} - \frac{\I}{\hf}\Big] +P_n\frac{\I}{\hf}\Big[\kappa(Z,X) - \hkappa\Big]\\
    &\pm P_n\frac{\I}{\f}\Big[\kappa(Z,X) - \hkappa\Big]\\ 
    =& P_n\kappa(Z,X)\Big[\frac{\I}{\f} - \frac{\I}{\hf}\Big]\\
    &+ P_n\Big[\kappa(Z,X) - \hkappa\Big]\Big[\frac{\I}{\hf} - \frac{\I}{\f}\Big]\\ 
    &+ P_n\frac{\I}{\f}\Big[\kappa(Z,X) - \hkappa\Big]\\
    =& P_n\kappa(Z,X)\Big[\frac{\I}{\f} - \frac{\I}{\hf}\Big] + P_n\frac{\I}{\f}\big[\kappa(Z,X) - \hkappa\big]\\ 
    &+ P_n\big[\kappa(Z,X) - \hkappa\big]\big[\frac{\I}{\hf} - \frac{\I}{\f}\big]\\
    &\pm P_0\big[\kappa(Z,X) - \hkappa\big]\big[\frac{\I}{\hf} - \frac{\I}{\f}\big]\\
    =& P_n\kappa(Z,X) \I\big[\frac{1}{\f} - \frac{1}{\hf}\big]\\
    &+ P_n\frac{\I}{\f}\big[\kappa(Z,X) - \hkappa\big]\\
    &+ (P_n - P_0)\big[\kappa(Z,X) - \hkappa\big]\big[\frac{\I}{\hf} - \frac{\I}{\f}\big]\\
    &+ P_0\big[\kappa(Z,X) - \hkappa\big]\big[\frac{\I}{\hf} - \frac{\I}{\f}\big]\\
    =&  P_n\kappa(Z,X)\I\big[\frac{1}{\f} - \frac{1}{\hf}\big]\\ &+ P_n\frac{\I}{\f}\big[\kappa(Z,X) - \hkappa\big]\\
    &+ P_0\big[\kappa(Z,X) - \hkappa\big]\big[\frac{\I}{\hf} - \frac{\I}{\f}\big] + o_p(n^{-1/2})
\end{align*}
}
\begin{align*}
    (4.2) =& P_n\big[\frac{\hat{\theta}(Z,X)}{\hf} - \frac{\theta(Z,X)}{\f}\big]
    \pm P_n\big[\frac{\hat{\theta}(Z,X)}{\f}\big]\\
    =& P_n\hat{\theta}(Z,X)\big[\frac{1}{\hf} - \frac{1}{\f}\big] + P_n\frac{\hat{\theta}(Z,X) - \theta(Z,X)}{\f} \\ 
    =& P_n\hat{\theta}(Z,X)\big[\frac{1}{\hf} - \frac{1}{\f}\big] \pm P_n\theta(Z,X)\big[\frac{1}{\hf} - \frac{1}{\f}\big]\\
    &+ P_n\frac{\hat{\theta}(Z,X) - \theta(Z,X)}{\f}\\
    =& P_n\big[\hat{\theta}(Z,X) - \theta(Z,X)\big]\big[\frac{1}{\hf} - \frac{1}{\f}\big]\\
    &+ P_n\theta(Z,X)\big[\frac{1}{\hf} - \frac{1}{\f}\big]
    + P_n\frac{\hat{\theta}(Z,X) - \theta(Z,X)}{\f}\\
    =& P_0\big[\hat{\theta}(Z,X) - \theta(Z,X)\big]\big[\frac{1}{\hf} - \frac{1}{\f}\big]\\
    &+ P_n\theta(Z,X)\big[\frac{1}{\hf} - \frac{1}{\f}\big]
    + P_n\frac{\hat{\theta}(Z,X) - \theta(Z,X)}{\f}
\end{align*}

{\small
\begin{align*}
    (3) =& P_n\Big\{\frac{\hz \I \delta(A,Z,X)[\w - \gamma(A,Z,X)]}{2\gamma(A,Z,X) \f \p}\\
    &- \frac{\hz I\hat{\delta}(A,Z,X)[\w - \hg]}{2\hg \hf \hp}\Big\}\\
    &\pm \frac{A(Z+A)\I\delta(A,Z,X) [\w - \gamma(A,Z,X)]}{2\hg \hf \hp}\\
    =& P_n \frac{\hz}{2} \I \delta(A,Z,X)\Big(\frac{\w - \gamma(A,Z,X)}{\gamma(A,Z,X) \f \p}\\
    &- \frac{\w - \gamma(A,Z,X)}{\hg \hf \hp} \Big)\\  
    &+ P_n\frac{A(Z+A)\I\delta(A,Z,X) [\w - \gamma(A,Z,X)]}{2\hg \hf \hp}\\
    &- P_n\frac{A(Z+A)\I\hat{\delta}(A,Z,X)[\w - \hg]}{2\hg \hf \hp}\\
    =&(P_n - P_0)\frac{\hz}{2}\I\delta(A,Z,X)\Big(\frac{\w - \gamma(A,Z,X)}{\gamma(A,Z,X) \f \p}\\
    &- \frac{\w - \gamma(A,Z,X)}{\hg \hf \hp}\Big)\\
    &+ P_0\frac{\hz \I}{2}\delta(A,Z,X)\Big(\frac{\w - \gamma(A,Z,X)}{\gamma(A,Z,X) \f \p}\\
    &- \frac{\w - \gamma(A,Z,X)}{\hg \hf \hp} \Big)\\
    &+ P_n\frac{A(Z+A)\I\delta(A,Z,X)[\w - \gamma(A,Z,X)]}{2\hg \hf \hp}\\ 
    &- P_n\frac{A(Z+A)\I\hat{\delta}(A,Z,X)[\w - \hg]}{2\hg \hf \hp}\\
    =& P_n\frac{A(Z+A)\I\delta(A,Z,X)[\w - \gamma(A,Z,X)]}{2\hg \hf \hp}\\ 
    &- P_n\frac{A(Z+A)\I\hat{\delta}(A,Z,X)[\w - \hg]}{2\hg \hf \hp}\\
    &+ o_p(n^{-1/2}) \textbf{ \{since } P_0\{w_{\alpha_Z}(A,Z, X, Y) - \gamma(A,Z,X)\} = o_p(n^{-1/2}) \textbf{ \} }\\
    =& P_n\frac{A(Z+A)\I\delta(A,Z,X)[\w - \gamma(A,Z,X)]}{2\hg \hf \hp}\\ 
    &- P_n\frac{A(Z+A)\I\hat{\delta}(A,Z,X)[\w - \hg]}{2\hg \hf \hp}\\
\end{align*}
}

{\small
\begin{align*}
    \pm& P_n\frac{A(Z+A)\I\delta(A,Z,X)[\w - \hg]}{2\hg \hf \hp} + o_p(n^{-1/2})\\
    & = P_n\frac{A(Z+A)\I\delta(A,Z,X)[\hg - \gamma(A,Z,X)]}{2\hg \hf \hp}\\
    &+ P_n\frac{A(Z+A)\I[\w - \hg]}{2\hg \hf \hp} \\
    &*[\delta(A,Z,X) - \hat{\delta}(A,Z,X)] + o_p(n^{-1/2}) \\
    =& P_n\frac{A(Z+A)\I \delta(A,Z,X)[\hg - \gamma(A,Z,X)]}{\hg \hf \hp}\\ 
    &\pm P_n\frac{A(Z+A)\I\delta(A,Z,X)[2\hg - \gamma(A,Z,X)]}{\gamma(A,Z,X) \hf \hp}\\
    &+ P_n\frac{A(Z+A)\I[\w - \hg]}{2\hg \hf \hp}\\&*[\delta(A,Z,X) - \hat{\delta}(A,Z,X)] + o_p(n^{-1/2})\\
    =& P_n\frac{A(Z+A)\I\delta(A,Z,X)[\hg - \gamma(A,Z,X)]}{2\hf \hp}\\
    &*\big[\frac{1}{\hg} - \frac{1}{\gamma(A,Z,X)}\big]\\
    &+ P_n\frac{A(Z+A)\I\delta(A,Z,X)[\hg - \gamma(A,Z,X)]}{2\gamma(A,Z,X) \hf \hp}\\
    &+ P_n\frac{A(Z+A)\I[\w - \hg]}{2\hg \hf \hp}\\
    &*[\delta(A,Z,X) - \hat{\delta}(A,Z,X)]+ o_p(n^{-1/2})\\
    =& P_n\frac{A(Z+A)\I\delta(A,Z,X)[\hg - \gamma(A,Z,X)]}{2\gamma(A,Z,X) \hf \hp}\\
    & \pm P_n\frac{A(Z+A)\I\delta(A,Z,X)[\hg - \gamma(A,Z,X)]}{2\gamma(A,Z,X) \f \hp}\\
    &+   P_n\frac{A(Z+A)\I\delta(A,Z,X)[\hg - \gamma(A,Z,X)]}{2\hf \hp}\\&*\big[\frac{1}{\hg} - \frac{1}{\gamma(A,Z,X)}\big]\\
    &+ P_n\frac{A(Z+A)\I[\w - \hg]}{2\hg \hf \hp}\\&*[\delta(A,Z,X) - \hat{\delta}(A,Z,X)] + o_p(n^{-1/2})\\
    =& P_n\frac{A(Z+A)\I\delta(A,Z,X)[\hg -\gamma(A,Z,X)]}{2\gamma(A,Z,X) \hp}\big[\frac{1}{\hf} - \frac{1}{\f}\big]
\end{align*}
}

{\small
\begin{align*}
    &+ P_n\frac{\hz \I \delta(A,Z,X) [\hg - \gamma(A,Z,X)]}{2\gamma(A,Z,X) \f \hp}\\
    &+  P_n\frac{A(Z+A)\I\delta(A,Z,X)[\hg - \gamma(A,Z,X)]}{2\hf \hp}\\&*\big[\frac{1}{\hg} - \frac{1}{\gamma(A,Z,X)}\big] \\
    +& P_n\frac{A(Z+A)\I[\w - \hg][\delta(A,Z,X) - \hat{\delta}(A,Z,X)]}{2\hg \hf \hp}\\& + o_p(n^{-1/2})\\
    =& P_n\frac{\hz \I \delta(A,Z,X)[\hg - \gamma(A,Z,X)]}{2\gamma(A,Z,X) \f \hp}\\
    & \pm P_n\frac{\hz \I \delta(A,Z,X)[\hg - \gamma(A,Z,X)]}{2\gamma(A,Z,X) \f \p}\\
    &+ P_n\frac{A(Z+A)\I\delta(A,Z,X)[\hg -\gamma(A,Z,X)]}{\gamma(A,Z,X) \hp}\\&*\Big[\frac{1}{\hf} - \frac{1}{\f}\Big]\\
    &+ P_n\frac{A(Z+A)I\delta(A,Z,X)[\hg - \gamma(A,Z,X)]}{2\hf \hp}\big[\frac{1}{\hg} - \frac{1}{\gamma(A,Z,X)}\big]\\
    +& P_n\frac{A(Z+A)\I[\w - \hg][\delta(A,Z,X) - \hat{\delta}(A,Z,X)]}{\hg \hf \hp}\\& + o_p(n^{-1/2})\\
    =& P_n\frac{\hz \I \delta(A,Z,X)[\hg - \gamma(A,Z,X)]}{\gamma(A,Z,X) \f}\\&*\big[\frac{1}{\hp} - \frac{1}{\p}\big]\\
    &+ P_n\frac{\hz \I \delta(A,Z,X)[\hg - \gamma(A,Z,X)]}{2\gamma(A,Z,X) \f \p}\\ 
    &+ P_n\frac{A(Z+A)\I\delta(A,Z,X)[\hg -\gamma(A,Z,X)]}{2\gamma(A,Z,X) \hp}\\&*\big[\frac{1}{\hf} - \frac{1}{\f}\big]\\
    &+ P_n\frac{A(Z+A)\I\delta(A,Z,X)[\hg - \gamma(A,Z,X)]}{2\hf \hp}\\&*\big[\frac{1}{\hg} - \frac{1}{\gamma(A,Z,X)}\big]\\
    +& P_n\frac{A(Z+A)\I[\w - \hg][\delta(A,Z,X) - \hat{\delta}(A,Z,X)}{2\hg \hf \hp}\\& + o_p(n^{-1/2})
\end{align*}
    
\begin{align*}
    =& P_0\frac{\hz \I \delta(A,Z,X)[\hg - \gamma(A,Z,X)]}{2\gamma(A,Z,X) \f}\\&*\big[\frac{1}{\hp} - \frac{1}{\p}\big]\\
    &+ P_n\frac{\hz \I \delta(A,Z,X)[\hg - \gamma(A,Z,X)]}{\gamma(A,Z,X) \f \p}\\ 
    &+ P_0\frac{A(Z+A)\I\delta[\hg -\gamma(A,Z,X)]}{2\gamma(A,Z,X) \hp}\big[\frac{1}{\hf} - \frac{1}{\f}\big] \\
    &+ P_0\frac{A(Z+A)\I\delta(A,Z,X)[\hg - \gamma(A,Z,X)]}{2\hf \hp}\\&*\big[\frac{1}{\hg} - \frac{1}{\gamma(A,Z,X)}\big]\\
    &+ P_0\frac{A(Z+A)\I [\w - \hg]}{2\hg \hf \hp}\\&*[\delta(A,Z,X) - \hat{\delta}(A,Z,X)] + o_p(n^{-1/2})\\
\end{align*}
}

{\small
\begin{align*}
    (4.&1) = P_n\frac{\hz \I \Q}{2\gamma(A,Z,X) \f \p}\\& - P_n\frac{\hz \I \hat{Q}(A,Z,X)}{2\hg \hf \hp}
    \pm P_n\frac{\hz \I\Q}{2\hg \hf \hp} \\
    =& P_n \frac{\hz \I \Q}{2}\big[\frac{1}{\gamma(A,Z,X) \f \p}\\
    &- \frac{1}{\hg \hf \hp}\big]\\ 
    &+ P_n\frac{\hz \I[\Q - \hat{Q}(A,Z,X)]}{2\hg \hf \hp} \\
    =& P_n \frac{\hz\I\Q}{2}\big[\frac{1}{\gamma(A,Z,X) \f \p}\\
    &- \frac{1}{\hg \hf \hp}\big]\\
    &+ P_n\frac{\hz \I[\Q - \hat{Q}(A,Z,X)]}{2\hg \hf \hp}\\
    &\pm P_n\frac{\hz \I[\Q - \hat{Q}(A,Z,X)]}{2\gamma(A,Z,X) \hf \hp}\\
    =& P_n\frac{\hz \I[\Q - \hat{Q}(A,Z,X)]}{2\hf \hp}\big[\frac{1}{\hg} - \frac{1}{\gamma(A,Z,X)}\big]\\
    &+ P_n\frac{\hz \I[\Q - \hat{Q}(A,Z,X)]}{2\gamma(A,Z,X) \hf \hp}\\
    &+ P_n \frac{\hz \I\Q}{2}\big[\frac{1}{\gamma(A,Z,X) \f \p}\\
    &- \frac{1}{\hg \hf \hp}\big]\\
    =& P_n\frac{\hz \I [\Q - \hat{Q}(A,Z,X)]}{2\gamma(A,Z,X) \hf \hp}\\
    &\pm P_n\frac{\hz \I[\Q - \hat{Q}(A,Z,X)]}{2\gamma(A,Z,X) \hf \p}\\
    &+ P_n\frac{\hz \I[\Q - \hat{Q}(A,Z,X)]}{2\hf \hp}\Big[\frac{1}{\hg} - \frac{1}{\gamma(A,Z,X)}\Big]\\
    &+ P_n \frac{\hz \I \Q}{2}\Big[\frac{1}{\gamma(A,Z,X) \f \p}\\
    &- \frac{1}{\hg \hf \hp}\Big]\\
\end{align*}
}

{\small
\begin{align*}
    =& P_n\frac{\hz \I [\Q - \hat{Q}(A,Z,X)]}{2\gamma(A,Z,X) \hf}\big[\frac{1}{\hp} - \frac{1}{\p}\big]\\
    &+ P_n\frac{\hz \I [\Q - \hat{Q}(A,Z,X)]}{2\gamma(A,Z,X) \hf \p} \\
    &+ P_n\frac{\hz \I[Q(A,Z,X) - \hat{Q}(A,Z,X)]}{2\hf \hp}\Big[\frac{1}{\hg} - \frac{1}{\gamma(A,Z,X)}\Big]\\
    &+ P_n \frac{\hz \I\Q}{2}\Big[\frac{1}{\gamma(A,Z,X) \f \p}\\
    &- \frac{1}{\hg \hf \hp}\Big]\\
    =& P_n\frac{\hz \I[\Q - \hat{Q}(A,Z,X)]}{2\gamma(A,Z,X) \hf \p}\\
    & \pm P_n\frac{A(Z+A)\I[\Q - \hat{Q}(A,Z,X)]}{2\gamma(A,Z,X) \f \p}\\
    &+ P_n\frac{\hz \I [\Q - \hat{Q}(A,Z,X)]}{2\gamma(A,Z,X) \hp}\Big[\frac{1}{\hf} - \frac{1}{\f}\Big] \\
    &+ P_n\frac{\hz \I[\Q - \hat{Q}(A,Z,X)]}{2\hf \hp}\Big(\frac{1}{\hg} - \frac{1}{\gamma(A,Z,X)}\Big)\\
    &+ P_n \frac{\hz\I\Q}{2}\Big[\frac{1}{\gamma(A,Z,X) \f \p}\\
    &- \frac{1}{\hg \hf \hp}\Big]\\
    =& P_n\frac{\hz \I[\Q - \hat{Q}(A,Z,X)]}{2\gamma(A,Z,X) \f \p}\\
    &+ P_n\frac{\hz \I[\Q - \hat{Q}(A,Z,X)]}{2\gamma(A,Z,X) \p}\Big[\frac{1}{\hf} - \frac{1}{\f}\Big]\\
    &+ P_n\frac{\hz \I[\Q - \hat{Q}(A,Z,X)]}{2\gamma(A,Z,X) \hf}\Big[\frac{1}{\hp} - \frac{1}{\p}\Big] \\
    &+ P_n\frac{\hz \I[\Q - \hat{Q}(A,Z,X)]}{2\hf \hp}\Big[\frac{1}{\hg} - \frac{1}{\gamma(A,Z,X)}\Big]\\
    &+ P_n\frac{\hz \I\Q}{2}\Big[\frac{1}{\gamma(A,Z,X) \f \p}\\
    &- \frac{1}{\hg \hf \hp}\Big]
\end{align*}
}

{\small
\begin{align*}
    =& P_n\frac{\hz \I[\Q - \hat{Q}(A,Z,X)]}{2\gamma(A,Z,X) \f \p}\\ 
    &+P_0\frac{\hz \I[\Q - \hat{Q}(A,Z,X)]}{2\gamma(A,Z,X) \p}\Big[\frac{1}{\hf} - \frac{1}{\f}\Big]\\ 
    &+ P_0\frac{\hz \I[\Q - \hat{Q}(A,Z,X)]}{2\gamma(A,Z,X) \hf}\Big[\frac{1}{\hp} - \frac{1}{\p}\Big] \\
    &+ P_0\frac{\hz \I[\Q - \hat{Q}(A,Z,X)]}{2\hf \hp}\Big[\frac{1}{\hg} - \frac{1}{\gamma(A,Z,X)}\Big]\\
    &+ P_0\frac{\hz \I\Q}{2}\Big[\frac{1}{\gamma(A,Z,X) \f \p}\\
    &- \frac{1}{\hg \hf \hp}\Big]\\
\end{align*}
}
In summary, we have 
{\small
\begin{align*}
    (2.1) =& P_n\kappa(Z,X)\I\big[\frac{1}{\f} - \frac{1}{\hf}\big]\\
    &+ P_n\frac{\I}{\f}\big[\kappa(Z,X) - \hkappa\big]\\
    &+ P_0\big[\kappa(Z,X) - \hkappa\big]\big[\frac{\I}{\hf} - \frac{\I}{\f}\big] + o_p(n^{-1/2})\\
    (3) =& P_0\frac{\hz \I \delta(A,Z,X)[\hg - \gamma(A,Z,X)]}{2\gamma(A,Z,X) \f}\\&*\big[\frac{1}{\hp} - \frac{1}{\p}\big]\\
    &+ P_n\frac{\hz \I \delta(A,Z,X)[\hg - \gamma(A,Z,X)]}{\gamma(A,Z,X) \f \p}\\ 
    &+ P_0\frac{A(Z+A)\I\delta(A,Z,X)[\hg -\gamma(A,Z,X)]}{2\gamma(A,Z,X) \hp}\\&*\big[\frac{1}{\hf} - \frac{1}{\f}\big] \\
    &+ P_0\frac{A(Z+A)\I\delta(A,Z,X)[\hg - \gamma(A,Z,X)]}{2\hf \hp}\\&*\big[\frac{1}{\hg} - \frac{1}{\gamma(A,Z,X)}\big]\\
    &+ P_0\frac{A(Z+A)\I [\w - \hg]}{2\hg \hf \hp}\\
    &*[\delta(A,Z,X) - \hat{\delta}(A,Z,X)]+ o_p(n^{-1/2})\\
    (4.1) =& P_n\frac{\hz \I[\Q - \hat{Q}(A,Z,X)]}{2\gamma(A,Z,X) \f \p}\\ 
    &+P_0\frac{\hz \I[\Q - \hat{Q}(A,Z,X)]}{2\gamma(A,Z,X) \p}\Big[\frac{1}{\hf} - \frac{1}{\f}\Big]\\ 
    &+ P_0\frac{\hz \I[\Q - \hat{Q}(A,Z,X)]}{2\gamma(A,Z,X) \hf}\\&*\Big[\frac{1}{\hp} - \frac{1}{\p}\Big] \\
    &+ P_0\frac{\hz \I[\Q - \hat{Q}(A,Z,X)]}{2\hf \hp}\\&*\Big[\frac{1}{\hg} - \frac{1}{\gamma(A,Z,X)}\Big]\\
    &+ P_0\frac{\hz \I\Q}{2}\\&*\Big[\frac{1}{\gamma(A,Z,X) \f \p}
    - \frac{1}{\hg \hf \hp}\Big]
\end{align*}
}
{\small
\begin{align*}
    (4.2) =& P_0\big[\hat{\theta}(Z,X) - \theta(Z,X)\big]\big[\frac{1}{\hf} - \frac{1}{\f}\big]\\
    &+ P_n\theta(Z,X)\big[\frac{1}{\hf} - \frac{1}{\f}\big] + P_n\frac{\hat{\theta}(Z,X) - \theta(Z,X)}{\f}
\end{align*}
}
Putting everything together, we have
{\small
\begin{align*}
    &R = P_n\Big\{ \frac{\hz}{2} Y \w \I\Big(\frac{1}{\hg \hf \hp}\\
    &- \frac{1}{\gamma(A,Z,X) \f \p} \Big)\Big\}\\
    &+ P_n\kappa(Z,X)\I\big[\frac{1}{\f} - \frac{1}{\hf}\big]\\
    &+ P_n\frac{\I}{\f}\big[\kappa(Z,X) - \hkappa\big]\\
    &+ P_0\big[\kappa(Z,X) - \hkappa\big]\big[\frac{\I}{\hf} - \frac{\I}{\f}\big] +  P_n\Big\{\hkappap - \kappa'(X)\Big\}\\
    &+ P_0\frac{\hz \I \delta(A,Z,X)[\hg - \gamma(A,Z,X)]}{2\gamma(A,Z,X) \f}\\&*\Big[\frac{1}{\hp} - \frac{1}{\p}\Big]\\
    &+ P_n\frac{\hz \I \delta(A,Z,X)[\hg - \gamma(A,Z,X)]}{2\gamma(A,Z,X) \f \p}\\ 
    &+ P_0\frac{A(Z+A)\I\delta(A,Z,X)[\hg -\gamma(A,Z,X)]}{2\gamma(A,Z,X) \hp}\\&*\Big[\frac{1}{\hf} - \frac{1}{\f}\Big] \\
    &+ P_0\frac{A(Z+A)\I\delta(A,Z,X)[\hg - \gamma(A,Z,X)]}{2\hf \hp}\\&*\Big[\frac{1}{\hg} - \frac{1}{\gamma(A,Z,X)}\Big]\\
    &+ P_0\frac{A(Z+A)\I [\w - \hg][\delta(A,Z,X) - \hat{\delta}(A,Z,X)]}{2\hg \hf \hp} 
    \end{align*}
    
\begin{align*}
    &+ P_n\frac{\hz \I[\Q - \hat{Q}(A,Z,X)]}{2\gamma(A,Z,X) \f \p}\\ 
    &+P_0\frac{\hz \I[\Q - \hat{Q}(A,Z,X)]}{2\gamma(A,Z,X) \p}\Big[\frac{1}{\hf} - \frac{1}{\f}\Big]\\ 
    &+ P_0\frac{\hz \I[\Q - \hat{Q}(A,Z,X)]}{2\gamma(A,Z,X) \hf}\Big[\frac{1}{\hp} - \frac{1}{\p}\Big] \\
    &+ P_0\frac{\hz \I[\Q - \hat{Q}(A,Z,X)]}{2\hf \hp}\Big[\frac{1}{\hg} - \frac{1}{\gamma(A,Z,X)}\Big]\\
    &+ P_n\frac{\hz \I\Q}{2}\Big[\frac{1}{\gamma(A,Z,X) \f \p}\\
    &- \frac{1}{\hg \hf \hp}\Big] \\
    &+ P_0\big[\hat{\theta}(Z,X) - \theta(Z,X)\big]\big[\frac{1}{\hf} - \frac{1}{\f}\big] + P_n\theta(Z,X)\big[\frac{1}{\hf} - \frac{1}{\f}\big]\\
    &+ P_n\frac{\hat{\theta}(Z,X) - \theta(Z,X)}{\f}+ o_p(n^{-1/2})
\end{align*}
}
We let $R_1$ be the term with $P_n$ and $R_2$ be the term with $P_0$ in R
{\small
\begin{align*}
    R_1 =& P_n\frac{\hz\I [Y w_{\alpha_Z}(Z,A,X,Y) - \Q]}{2}\\&*\Big[\frac{1}{\hg \hf \hp}
    - \frac{1}{\gamma(A,Z,X)\f\p}\Big]\\
    &+ P_n\Big[\kappa(Z,X)\I - \theta(Z,X)\Big]\Big[\frac{1}{\f} - \frac{1}{\hf}\Big]\\
    &+ P_n\Big\{\hkappap - \kappa'(X) + \frac{\I}{\f}\big[\kappa(Z,X) - \hkappa\big]\Big\}\\
    &+ P_n\frac{\hz \I \delta(A,Z,X)\big[\hg - \gamma(A,Z,X) \big]}{2 \gamma(A,Z,X)\f\p}\\
    &+ P_n\frac{\hz \I \big[\Q - \hat Q (A,Z,X) \big]}{2 \gamma(A,Z,X)\f\p} + P_n\frac{\hat \theta (Z,X) - \theta(Z,X)}{\f}\\
    =& R_1(1) + R_1(2) + R_1(3) + R_1(4) + R_1(5)
\end{align*}
}

{\small
\begin{align*}
    &R_2 = P_0\big[\kappa(Z,X) - \hkappa\big]\big[\frac{\I}{\hf} - \frac{\I}{\f}\big]\\
    &+ P_0\frac{\hz \I \delta(A,Z,X)[\hg - \gamma(A,Z,X)]}{2\gamma(A,Z,X) \f}\\&*\Big[\frac{1}{\hp} - \frac{1}{\p}\Big]\\
    &+ P_0\frac{A(Z+A)\I\delta(A,Z,X)[\hg -\gamma(A,Z,X)]}{2\gamma(A,Z,X) \hp}\\&*\Big[\frac{1}{\hf} - \frac{1}{\f}\Big] \\
    &+ P_0\frac{A(Z+A)\I\delta(A,Z,X)[\hg - \gamma(A,Z,X)]}{2\hf \hp}\\&*\Big[\frac{1}{\hg} - \frac{1}{\gamma(A,Z,X)}\Big]\\
    &+ P_0\frac{A(Z+A)\I [\w - \hg][\delta(A,Z,X) - \hat{\delta}(A,Z,X)]}{2\hg \hf \hp} \\
    &+P_0\frac{\hz \I[\Q - \hat{Q}(A,Z,X)]}{2\gamma(A,Z,X) \p}\Big[\frac{1}{\hf} - \frac{1}{\f}\Big]\\ 
    &+ P_0\frac{\hz \I[\Q - \hat{Q}(A,Z,X)]}{2\gamma(A,Z,X) \hf}\Big[\frac{1}{\hp} - \frac{1}{\p}\Big] \\
    &+ P_0\frac{\hz \I[\Q - \hat{Q}(A,Z,X)]}{2\hf \hp}\Big[\frac{1}{\hg} - \frac{1}{\gamma(A,Z,X)}\Big]\\
    &+ P_0\big[\hat{\theta}(Z,X) - \theta(Z,X)\big]\big[\frac{1}{\hf} - \frac{1}{\f}\big]+ o_p(n^{-1/2})
\end{align*}
}
We simplify the term in $R_1$

{\small
\begin{align*}
    &R_1(1) = o_p(n^{-1/2}) \textbf{ because } P_0\Big[Y\w - \Q\Big] = o_p(n^{-1/2}) \\
    &R_1(2) = o_p(n^{-1/2}) \textbf{ because } P_0\Big[\kappa(Z,X) \I - \theta(Z,X)\Big] = o_p(n^{-1/2})\\
    &R_1(3) = P_n\Big\{\hkappap - \kappa'(X) + \frac{\I}{\f}\big[\kappa(Z,X) - \hkappa\big]\Big\}\\
     =& P_0\hkappap - P_0\kappa'(X) + \Big[P_n - P_0\Big]\Big[\hkappap - \kappa'(X)\Big]\\ 
     &+ P_0\frac{\I \kappa(Z,X)}{\f} - P_0\frac{\I \hkappa}{\f} + (P_n - P_0)\frac{\kappa(X) - \hkappa}{\f}\\
    =& o_p(n^{-1/2}) \\ &\textbf{ because } P_0\hkappap =  P_0\frac{\I \hkappa}{\f},P_0\kappa'(X) =  P_0\frac{\I \kappa(Z,X)}{\f} 
\end{align*}
}
{\small
\begin{align*}
    R_1(5) =& P_n\frac{\hz \I\Big[\Q - \hat{Q}(A,Z,X)\Big]}{2\gamma(A,Z,X) \f \p}\\& + P_n\frac{\hat{\theta}(Z,X) - \theta(Z,X)}{\f} \\
    =& -P_0\Big[\frac{I\{\pi(X) = 1\}\hat{Q}(1,1,X)}{\gamma(1,1,X)} + \frac{I\{\pi(X) = -1\}\hat{Q}(-1,-1,X)}{\gamma(-1,-1,X)}\Big]\\
    &+ P_0\Big[\frac{I\{\pi(X) = 1\}\hat{Q}(1,1,X)}{\hat \gamma(1,1,X)} +\frac{I\{\pi(X) = -1\}\hat{Q}(-1,-1,X)}{\hat \gamma(-1,-1,X)}\Big] + o_p(n^{-1/2})\\ 
    =& P_0\Big[ I\{\pi(X) = 1\}\hat{Q}(1,1,X)\Big(\frac{ \gamma(1,1,X) - \hat \gamma(1,1,X)}{\hat \gamma(1,1,X)\gamma(1,1,X)}\Big)\Big]\\
    &+ P_0\Big[I\{\pi(X) = -1\}\hat{Q}(-1,-1,X)\Big(\frac{ \gamma(-1,-1,X) - \hat \gamma(-1,-1,X)}{\hat \gamma(-1,-1,X)\gamma(-1,-1,X)}\Big)\Big]\\& + o_p(n^{-1/2})\\
    R_1(4) =& P_n\frac{\hz \I \delta(A,Z,X)\Big[\hg - \gamma(A,Z,X)\Big]}{2\gamma(A,Z,X) \f \p}\\
    =& P_0 \frac{I\{\pi(X) = 1\}Q(1,1,X)[\hat \gamma(1,1,X) - \gamma(1,1,X)]}{[\gamma(1,1,X)]^2}\\
    & + P_0\frac{I\{\pi(X) = -1\}Q(-1,-1,X)[\hat \gamma(-1,-1,X) - \gamma(-1,-1,X)]}{[\gamma(-1,-1,X)]^2} 
\end{align*}
}
Lastly, we simplify two terms in $R_2$

\begin{align*}
   &P_0\frac{\hz \I}{2\gamma(A,Z,X) \p}\Big[\Q - \hat{Q}(A,Z,X)\Big]\Big[\frac{1}{\hf} - \frac{1}{\f}\Big]\\
   &+ P_0\Big[\hat{\theta}(Z,X) - \theta(Z,X)\Big]\Big[\frac{1}{\hf} - \frac{1}{\f}\Big] \\
   =&P_0\Big[\frac{1}{\hf} -  \frac{1}{\f}\Big]\Big[\frac{\hz \I \Q}{2\gamma(A,Z,X) \p} - \theta(Z,X) \Big]\\ 
   &- P_0\Big[\frac{1}{\hf} -  \frac{1}{\f}\Big]\Big[\frac{\hz \I \hat{Q}(A,Z,X)}{2\gamma(A,Z,X) \p} - \hat{\theta}(Z,X) \Big]\\
   =& o_p(n^{-1/2})
\end{align*}

The second equation is due to 

$$P_0\frac{\hz \I \Q}{2\gamma(A,Z,X) \p} = P_0\theta(Z,X)$$ 
 and 
$$P_0\frac{\hz \I \hat{Q}(A,Z,X)}{2\gamma(A,Z,X) \p} = P_0\hat{\theta}(Z,X)$$

Putting all the components together
{\small
\begin{align*}
     R =& P_0I\{\pi(X) = 1\}\frac{\hat{\gamma}(1,1,X) - \gamma(1,1,X)}{\gamma(1,1,X)}\Big[\frac{\hat{Q}(1,1,X)}{\hat \gamma(1,1,X)} - \frac{Q(1,1,X)}{\gamma(1,1,X)}\Big]\\
     &+ P_0 I\{\pi(X) = -1\}\frac{\hat{\gamma}^-(Z) - \gamma(-1,-1,X)}{\gamma(-1,-1,X)}\Big[\frac{\hat{Q}(-1,-1,X)}{\hat \gamma(-1,-1,X)} - \frac{Q(-1,-1,X)}{\gamma(-1,-1,X)}\Big]\\
     &+  P_0\Big[\kappa(Z,X)  - \hat{\kappa}(Z,X)\Big]\Big(\frac{\I}{\hf} - \frac{\I}{\f}\Big)\\ 
     &+ P_0\frac{\hz \I \delta(A,Z,X)}{2\hp \hf}\Big[\hg - \gamma(A,Z,X)\Big]\\&*\Big[\frac{1}{\hg} - \frac{1}{\gamma(A,Z,X)}\Big] \\
    &+ P_0\frac{\hz \I \delta(A,Z,X)}{2\gamma(A,Z,X) \hp}\Big[\hg - \gamma(A,Z,X)\Big]\\&*\Big[\frac{1}{\hf} - \frac{1}{\f}\Big]\\
    &+ P_0\frac{\hz \I \delta(A,Z,X)}{2\gamma(A,Z,X) \f}\Big[\hg - \gamma(A,Z,X)\Big]\\&*\Big[\frac{1}{\hp} - \frac{1}{\p}\Big]\\
    &+ P_0\frac{\hz \I}{2\hf \hp}\Big[\Q - \hat{Q}(A,Z,X)\Big]\\&*\Big[\frac{1}{\hg} - \frac{1}{\gamma(A,Z,X)}\Big]\\
    &+ P_0\frac{\hz \I}{2\gamma(A,Z,X) \hf}\Big[\Q - \hat{Q}(A,Z,X)\Big]\\&*\Big[\frac{1}{\hp} - \frac{1}{\p}\Big] + o_p(n^{-1/2})
    \end{align*}

    \begin{align*}
    =& P_0\frac{\hat{\gamma}(1,1,X) - \gamma(1,1,X)}{\gamma(1,1,X)}\Big[\hat{\delta}(1,1,X) - \delta(1,1,X)\Big]\\
    &+ P_0\frac{\hat{\gamma}(-1,-1,X) - \gamma(-1,-1,X)}{\gamma(-1,-1,X)}\Big[\hat{\delta}(-1,-1,X) - \delta(-1,-1,X) \Big]\\
    &+  P_0\Big[\kappa(Z,X)  - \hat{\kappa}(Z,X)\Big]\Big(\frac{\I}{\hf} - \frac{\I}{\f}\Big)\\ 
     &+ P_0\frac{\hz \I \delta(A,Z,X)}{2\hp \hf}\Big[\hg - \gamma(A,Z,X)\Big]\\&*\Big[\frac{1}{\hg} - \frac{1}{\gamma(A,Z,X)}\Big] \\
    &+ P_0\frac{\hz \I \delta(A,Z,X)}{2\gamma(A,Z,X) \hp}\Big[\hg - \gamma(A,Z,X)\Big]\\&*\Big[\frac{1}{\hf} - \frac{1}{\f}\Big]\\
    &+ P_0\frac{\hz \I \delta(A,Z,X)}{2\gamma(A,Z,X) \f}\Big[\hg - \gamma(A,Z,X)\Big]\\&*\Big[\frac{1}{\hp} - \frac{1}{\p}\Big]\\
    &+ P_0\frac{\hz \I}{2\hf \hp}\Big[\Q - \hat{Q}(A,Z,X)\Big]\Big[\frac{1}{\hg} - \frac{1}{\gamma(A,Z,X)}\Big]\\
    &+ P_0\frac{\hz \I}{2\gamma(A,Z,X) \hf}\Big[\Q - \hat{Q}(A,Z,X)\Big]\Big[\frac{1}{\hp} - \frac{1}{\p}\Big]\\& + o_p(n^{-1/2})
\end{align*}
}
Hence, the remainder will go to $o_p(n^{-1/2})$ under the assumption:

For any $z$, 
{\small
\begin{align*}
&(\|f- \hat f \|_{P_0} + \|p - \hat p \|_{P_{0}} +\|\gamma- \hat \gamma \|_{P_0})(\|\gamma- \hat \gamma \|_{P_0})\\& + (\|p - \hat p \|_{P_{0}} +\|\gamma- \hat \gamma \|_{P_0})(\|Q- \hat Q \|_{P_0})
\\&+  (\|\gamma- \hat \gamma)(\|\delta- \hat \delta\|) + (\|\kappa - \hat \kappa \|_{P_0})(\|f- \hat f \|_{P_0}) = o_p(n^{-1/2})
\end{align*}
}
where $f = f(Z|X)$ and $p = f(A|Z,X)$

\section{Generating the outcome of the compliers}
\label{rejectionsampling}
The outcome of the compliers, $y^{PS = S4}|x$, follows the distribution $Y|PS = S4,X,A,Z$ with the density 

$$f(y|PS = S4, X, A, Z) = \frac{w_{\alpha_Z}(A,Z,X,y)f_{Z}(y)}{\gamma(A,Z,X)}$$

We will use the rejection sampling method to draw samples of $Y|PS = S4,X,A,Z$ from a normal distribution $Y_N \sim N(1.5, 2^2)$ with the density 

$$f_{Y_N}(y) = \frac{1}{2\sqrt{2\pi}}\exp\left\{\frac{-1}{2}\left(\frac{y - 1.5}{2}\right)^2\right\}$$

using the following algorithm.

\begin{algorithm}
\caption{Rejection Sampling algorithm to draw $y^{PS = S4}|x$}\label{alg:one}
\begin{algorithmic}
\State \textbf{Input:}  $x, \alpha, z, y$
\For{$i = 1 \text{ to } 8000$}
   \State Draw $u_i$  from Uniform(0,1)
    \State Draw $y^i_N$ from Normal(1.5, 4)
   \State $C \gets \frac{f(y^i_N|PS = S4,x,a,z)}{M*f_{Y_N}(y^i_N)}$
   \If{$u_i < C$}
    Accept $y^i_N$
\ElsIf{$u_i > C$}
    Reject $y^i_N$
\EndIf
\EndFor
\State \textbf{Output:} $y^{PS = S4}|x = $  the mean of the accepted $y^i_N$
\end{algorithmic}
\end{algorithm}

We apply the Algorithm \ref{alg:one} to every subject that has the compliance level $A$ equal to the treatment $Z$.

\section{Addition Simulations with different values of \texorpdfstring{$(\alpha^Y_{-1}, \alpha^Y_{+1})$}{(alpha\_Y\_{-1}, alpha\_Y\_{+1})}}

\subsection{Sample Size 250}

\begin{table}[ht]
\centering
 \caption{Simulation Result: Mean (sd) of value functions. The empirical optimal value function is 1.68. The sensitivity parameter are ($\alpha^Y_{-1}, \alpha^Y_{+1}$) = (0.5,0.5) }
\begin{tabular}{rlllll}
  \hline
 & Case & OWL & IVT & IPW & MR\\ 
  \hline
  1 & All correctly specified & 1.05 (0.00) & 1.46 (0.18) & 1.63 (0.09) & 1.67 (0.09) \\ 
  2 & f(A,Z$|$X) misspecified & 1.06 (0.06) & 1.47 (0.16) & 1.54 (0.10) & 1.67 (0.08) \\ 
  3 & $Q(A,Z,X)$ misspecified & 1.04 (0.02) & 1.48 (0.16) & 1.58 (0.06) & 1.55 (0.09) \\ 
   \hline
\end{tabular}
\end{table}

\begin{table}[ht]
\centering
\caption{Simulation Result: Mean (sd) of correct classification rate. The sensitivity parameter are ($ \alpha^Y_{-1}, \alpha^Y_{+1}$) = (0.5,0.5) }
\begin{tabular}{rlllll}
  \hline
 & Case & OWL & IVT & IPW & MR\\ 
  \hline

  1 & All correctly specified & 0.51 (0.04) & 0.72 (0.10) & 0.87 (0.06) & 0.94 (0.04) \\ 
  2 & f(A,Z$|$X) misspecified & 0.52 (0.06) & 0.73 (0.09) & 0.78 (0.05) & 0.95 (0.03) \\ 
  3 & $Q(A,Z,X)$ misspecified & 0.52 (0.05) & 0.72 (0.09) & 0.87 (0.06) & 0.86 (0.07) \\ 
   \hline
\end{tabular}
\end{table}

\begin{table}[ht]
\centering
\caption{Simulation Result: Demonstration of the robustness of the multiply robust estimator of the value function of the IPW method. The sensitivity parameter are ($\alpha^Y_{-1}, \alpha^Y_{+1}$) = (0.5,0.5)}
\begin{tabular}{rllll}
  \hline
 & Case & Empirical & Non-robust Estimator & Multiply Robust \\
 & & & & Estimator \\
  \hline
  1 & All correctly specified & 1.63 (0.09) & 1.63 (0.21) & 1.59 (0.13) \\ 
  2 & f(Z$|$X) misspecified & 1.52 (0.09) & 1.93 (0.17) & 1.49 (0.10) \\ 
  3 & f(A$|$Z,X) misspecified & 1.62 (0.09) & 2.48 (0.40) & 1.52 (0.21) \\
  4 & $Q(A,Z,X)$ misspecified & 1.64 (0.06) & 1.66 (0.13) & 1.62 (0.10) \\
   \hline
\end{tabular}
\end{table}

\subsection{Sample Size 500}
\begin{table}[ht]
\centering
 \caption{Simulation Result: Mean (sd) of value functions. The empirical optimal value function is 1.62. The sensitivity parameter are ($\alpha^Y_{-1}, \alpha^Y_{+1}$) = (0.5,-0.5) }
\begin{tabular}{rlllll}
  \hline
 & Case & OWL & IVT & IPW & MR\\ 
  \hline

 1 & All correctly specified & 0.93 (0.00) & 1.39 (0.14) & 1.58 (0.07) & 1.61 (0.07) \\ 
  2 & f(A, Z$|$X) misspecified & 0.93 (0.02) & 1.40 (0.10) & 1.49 (0.07) & 1.62 (0.06) \\ 
  4 & $Q(A,Z,X)$ misspecified & 0.93 (0.00) & 1.38 (0.14) & 1.58 (0.06) & 1.59 (0.15) \\ 
   \hline
\end{tabular}
\end{table}

\begin{table}[ht]
\centering
\caption{Simulation Result: Mean (sd) of correct classification rate. The sensitivity parameter are ($ \alpha^Y_{-1}, \alpha^Y_{+1}$) = (0.5,-0.5) }
\begin{tabular}{rlllll}
  \hline
 & Case & OWL & IVT & IPW & MR\\ 
  \hline

 1 & All correctly specified & 0.48 (0.03) & 0.70 (0.07) & 0.89 (0.05) & 0.95 (0.03) \\ 
  2 & f(A, Z$|$X) misspecified & 0.48 (0.03) & 0.71 (0.06) & 0.79 (0.04) & 0.97 (0.02) \\ 
  3 & $Q(A,Z,X)$ misspecified & 0.48 (0.03) & 0.69 (0.07) & 0.89 (0.03) & 0.87 (0.06) \\ 
   \hline
\end{tabular}
\end{table}

\begin{table}[ht]
\centering
\caption{Simulation Result: Demonstration of the robustness of the multiply robust estimator of the value function of the IPW method. The sensitivity parameter are ($\alpha^Y_{-1}, \alpha^Y_{+1}$) = (0.5,-0.5)}
\begin{tabular}{rllll}
  \hline
 & Case & Empirical & Non-robust Estimator & Multiply robust \\ 
 &&&& Estimator \\
  \hline

  1 & All correctly specified & 1.58 (0.07) & 1.59 (0.14) & 1.56 (0.09) \\ 
  2 & f(Z$|$X) misspecified & 1.47 (0.07) & 1.84 (0.11) & 1.43 (0.08) \\ 
  3 & f(A$|$Z,X) misspecified & 1.57 (0.07) & 2.44 (0.26) & 1.53 (0.13) \\ 
  4 & $Q(A,Z,X)$ misspecified & 1.58 (0.06) & 1.60 (0.14) & 1.56 (0.14) \\
   \hline
\end{tabular}
\end{table}

\begin{table}
\centering
 \caption{Simulation Result: Mean (sd) of value functions. The empirical optimal value function is 1.65. The sensitivity parameter are ($\alpha^Y_{-1}, \alpha^Y_{+1}$) = (0,0) }
\begin{tabular}{rlllll}
  \hline
 & Case & OWL & IVT & IPW & MR\\ 
  \hline
%1 & All correctly specified & 1.0017 (0.0001) & 1.4384 (0.0983) & 1.6066 (0.0579) & 1.6248 (0.0564) \\ 
 % 2 & f(Z$|$X) misspecified & 1.0017 (0.0001) & 1.4877 (0.0753) & 1.6169 (0.0583) & 1.6261 (0.0568) \\ 
 % 3 & f(A$|$Z,X) misspecified & 1.0017 (0.0001) & 1.4339 (0.1124) & 1.5955 (0.0606) & 1.6190 (0.0579) \\ 
  %4 & Q(X) misspecified & 1.0017 (0.0001) & 1.4263 (0.1117) & 1.5986 (0.0564) & 1.5912 (0.0606) \\ 
  1 & All correctly specified & 1.00 (0.00) & 1.44 (0.13) & 1.61 (0.06) & 1.64 (0.06) \\ 
  2 & f(A, Z$|$X) misspecified & 1.00 (0.02) & 1.52 (0.10) & 1.65 (0.06) & 1.65 (0.06) \\ 
  3 & $Q(A,Z,X)$ misspecified & 1.00 (0.00) & 1.46 (0.19) & 1.61 (0.07) & 1.61 (0.07) \\ 
   \hline
\end{tabular}
\end{table}

\begin{table}
\centering
\caption{Simulation Result: Mean (sd) of correct classification rate. The sensitivity parameter are ($ \alpha^Y_{-1}, \alpha^Y_{+1}$) = (0,0) }
\begin{tabular}{rlllll}
  \hline
 & Case & OWL & IVT & IPW & MR\\ 
  \hline
%1 & All correctly specified & 0.5189 (0.0299) & 0.7216 (0.0530) & 0.9120 (0.0334) & 0.9671 (0.0193) \\ 
 % 2 & f(Z$|$X) misspecified & 0.5197 (0.0314) & 0.7594 (0.0463) & 0.9368 (0.0304) & 0.9688 (0.0200) \\ 
  %3 & f(A$|$Z,X) misspecified & 0.5227 (0.0319) &  0.7253 (0.0604) & 0.9022 (0.0342) &  0.9621 (0.0232) \\ 
 % 4 & Q(X) misspecified & 0.5251 (0.0323) & 0.7240 (0.0607) & 0.9138 (0.0344) & 0.8988 (0.0384) \\ 
 1 & All correctly specified & 0.51 (0.03) & 0.72 (0.07) & 0.89 (0.05) & 0.95 (0.03) \\ 
  2 & f(A,Z$|$X) misspecified & 0.51 (0.04) & 0.73 (0.06) & 0.79 (0.04) & 0.97 (0.02) \\ 
  3 & $Q(A,Z,X)$ misspecified & 0.51 (0.03) & 0.73 (0.07) & 0.89 (0.05) & 0.89 (0.05) \\ 
   \hline
\end{tabular}
\end{table}

\begin{table}
\centering
\caption{Simulation Result: Demonstration of the robustness of the multiply robust estimator of the value function of the IPW method. The sensitivity parameter are ($ \alpha^Y_{-1}, \alpha^Y_{+1}$) = (0,0)}
\begin{tabular}{rllll}
  \hline
 & Case & Empirical & Non-robust Estimator & Multiply Robust \\ 
 &&&& Estimator \\
  \hline
%1 & All correctly specified & 1.6066 (0.0579) & 1.6374 (0.0859) & 1.5854 (0.0877) \\ 
 % 2 & f(Z$|$X) misspecified & 1.6169 (0.0583) & 1.9592 (0.1244) & 1.5972 (0.0855) \\ 
  %3 & f(A$|$Z,X) misspecified & 1.5955 (0.0606) & 2.5346 (0.1759) & 1.6082 (0.0754) \\ 
  1 & All correctly specified & 1.61 (0.06) & 1.63 (0.15) & 1.59 (0.09) \\ 
  2 & f(Z$|$X) misspecified & 1.49 (0.07) & 1.88 (0.12) & 1.45 (0.07) \\ 
  3 & f(A$|$Z,X) misspecified & 1.59 (0.07) & 2.48 (0.27) & 1.55 (0.13)\\
  4 & $Q(A,Z,X)$ misspecified & 1.61 (0.07) & 1.62 (0.13) & 1.60 (0.12)\\
   \hline
\end{tabular}
\end{table}

\clearpage

\section{Additional Heat Maps}
\label{appendixd}

\subsection{\texorpdfstring{$\alpha^0_{Z}$}{alpha\_0\_Z} Correctly Specified}

Below are the heat maps of the multiply robust method when nuisance parameters are misspecified. In these cases, $\alpha^0_z$ are known and set at 0. In the cases when $f(Z|X), f(A|Z,X)$ are misspecified, the correct classification rate is consistently high, around $95\%$. This suggests that the MR method is insensitive to the misspecification of $\alpha^Y_z$. When $Q(A,Z,X)$ is misspecified, the performance of the MR method is similar to the IPW method.

\begin{figure}
    \centering
    \includegraphics[width = \textwidth]{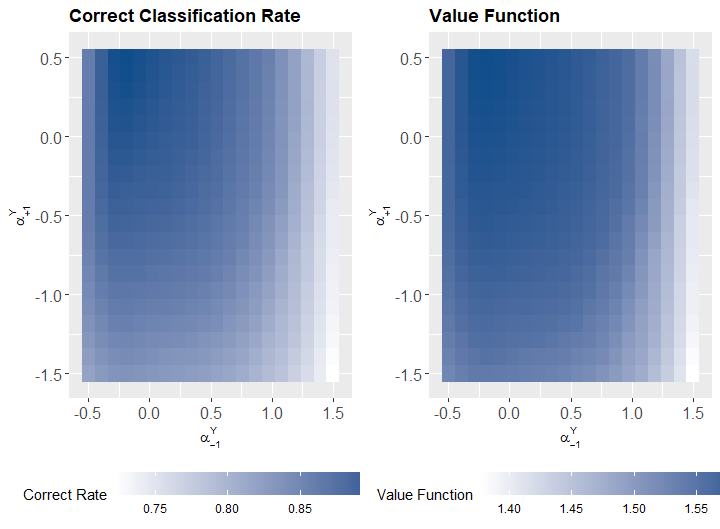}
    \caption{Sensitivity analysis of the multiply robust method when the nuisance parameters $Q(1,1,X)$ and $Q(-1,-1,X)$ are incorrectly specified: Sensitivity parameters are ($ \alpha^Y_{-1}, \alpha^Y_{+1}$). The darker area indicates a higher correct classification rate/value function. The true vector of sensitivity parameters are ($ \alpha^Y_{-1}, \alpha^Y_{+1}$) = (0.5,-0.5).}
    \label{fig:mrwsensi2}
\end{figure}

\begin{figure}
    \centering
    \includegraphics[width = \textwidth]{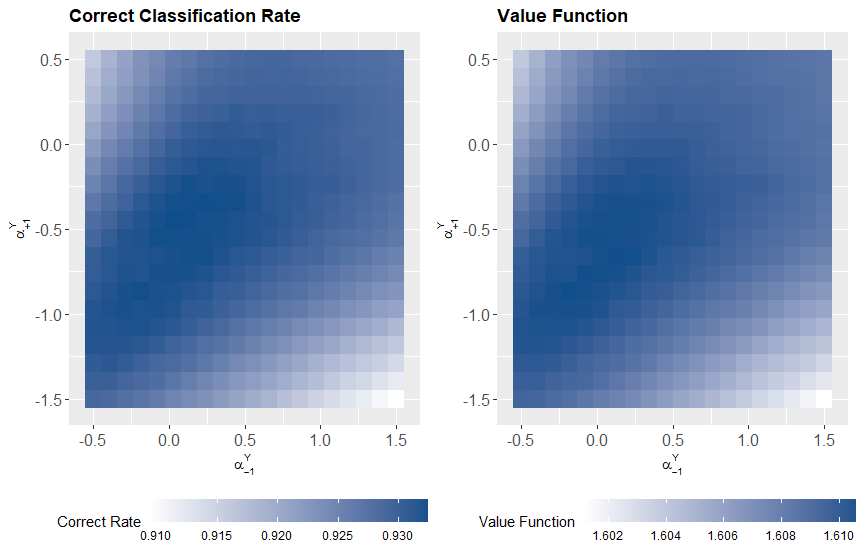}
    \caption{Sensitivity analysis of the multiply robust method when the nuisance parameter $f(A|Z,X)$ is incorrectly specified: Sensitivity parameters are ($ \alpha^Y_{-1}, \alpha^Y_{+1}$). The darker area indicates a higher correct classification rate/value function. The true vector of sensitivity parameters are ($ \alpha^Y_{-1}, \alpha^Y_{+1}$) = (0.5,-0.5).}
    \label{fig:mrwsensi3}
\end{figure}

\begin{figure}
    \centering
    \includegraphics[width = \textwidth]{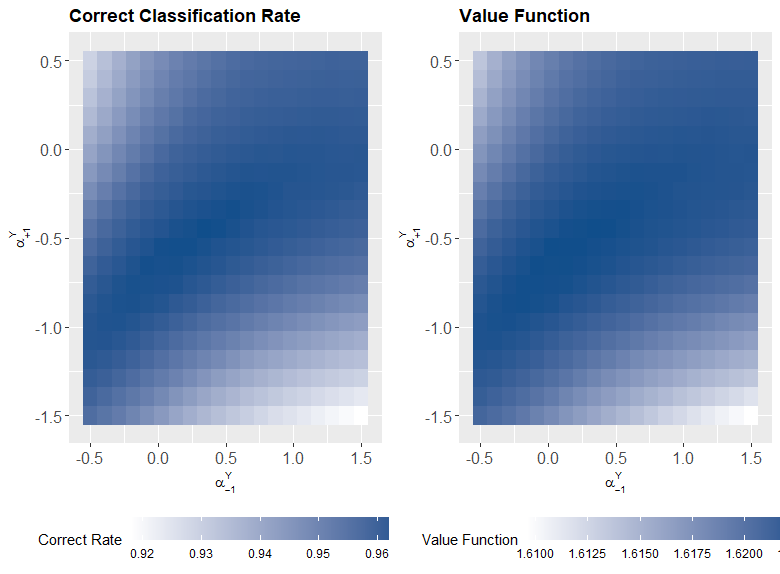}
    \caption{Sensitivity analysis of the multiply robust method when the nuisance parameter $f(Z|X)$ is incorrectly specified: Sensitivity parameters are ($ \alpha^Y_{-1}, \alpha^Y_{+1}$). The darker area indicates a higher correct classification rate/value function. The true vector of sensitivity parameters are ($\alpha^Y_{-1}, \alpha^Y_{+1}$) = (0.5,-0.5). }
    \label{fig:mrwsensi4}
\end{figure}

\subsection{\texorpdfstring{$\alpha^0_{Z}$}{alpha\_0\_Z} Incorrectly Specified}

We show below the additional heat maps of the case when both of the sensitivity parameters ($\alpha^0_{-1}, \alpha^0_{+1}$) are misspecified. 

\begin{figure}
    \centering
    \includegraphics[width = \textwidth]{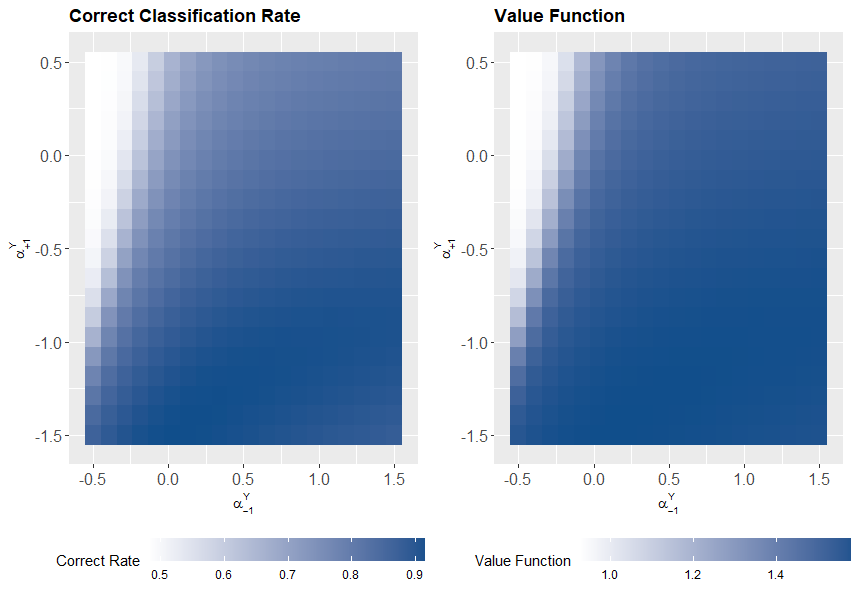}
    \caption{Sensitivity analysis of the proposed method: Sensitivity parameters are ($\alpha^Y_{-1}, \alpha^Y_{+1}$). The darker area indicates a higher correct classification rate/value function. The true vector of sensitivity parameters are ($\alpha^Y_{-1}, \alpha^Y_{+1}$) = (0.5,-0.5). ($\alpha^0_{-1}, \alpha^0_{+1}$) is misspecified and set at (0.5, 0.5). The true ($\alpha^0_{-1}, \alpha^0_{+1}$) = (0,0).}
    \label{fig:alpha0mis-propsensi3}
\end{figure}

\begin{figure}
    \centering
    \includegraphics[width = \textwidth]{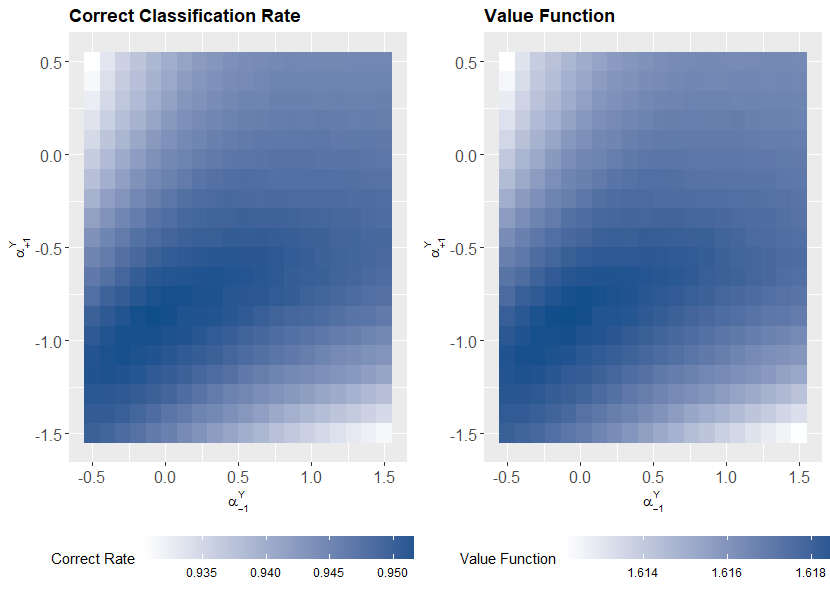}
    \caption{Sensitivity analysis of the multiply robust estimator when all nuisance parameters are correctly specified: Sensitivity parameters are ($\alpha^Y_{-1}, \alpha^Y_{+1}$). The darker area indicates a higher correct classification rate/value function. The true vector of sensitivity parameters are ($\alpha^Y_{-1}, \alpha^Y_{+1}$) = (0.5,-0.5). ($\alpha^0_{-1}, \alpha^0_{+1}$) is misspecified and set at (0.5, 0.5). The true ($\alpha^0_{-1}, \alpha^0_{+1}$) = (0,0).}
    \label{fig:alpha0mis-mrwsensi3}
\end{figure}

\begin{figure}
    \centering
    \includegraphics[width = \textwidth]{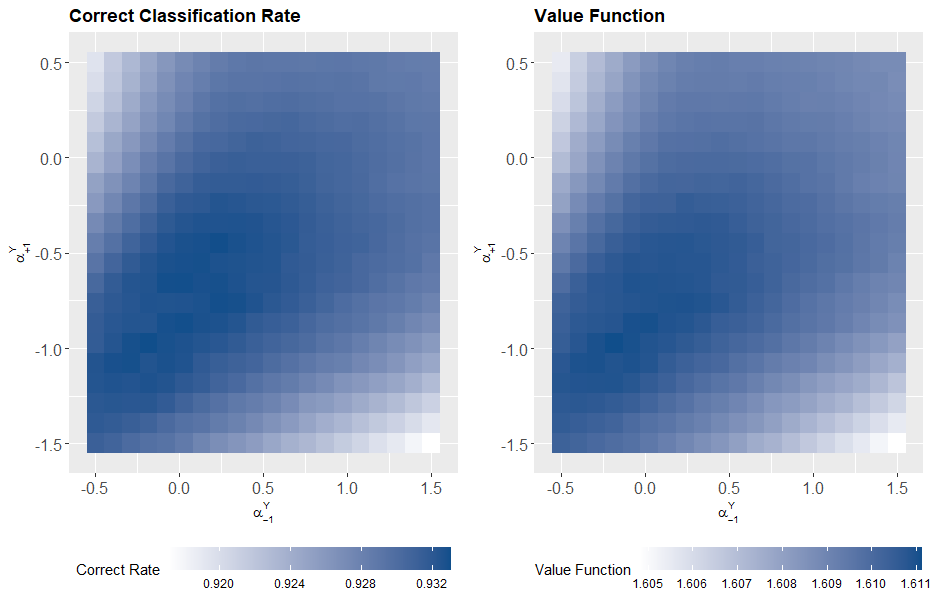}
    \caption{Sensitivity analysis of the multiply robust estimator when $f(A|Z,X)$ is mis-specified: Sensitivity parameters are ($ \alpha^Y_{-1}, \alpha^Y_{+1}$). The darker area indicates a higher correct classification rate/value function. The true vector of sensitivity parameters are ($\alpha^Y_{-1}, \alpha^Y_{+1}$) = (0.5,-0.5). ($\alpha^0_{-1}, \alpha^0_{+1}$) is set at (0.5, 0.5). The true ($\alpha^0_{-1}, \alpha^0_{+1}$) = (0,0).}
    \label{fig:alpha0mis-fazmissensi3}
\end{figure}

\begin{figure}
    \centering
    \includegraphics[width = \textwidth]{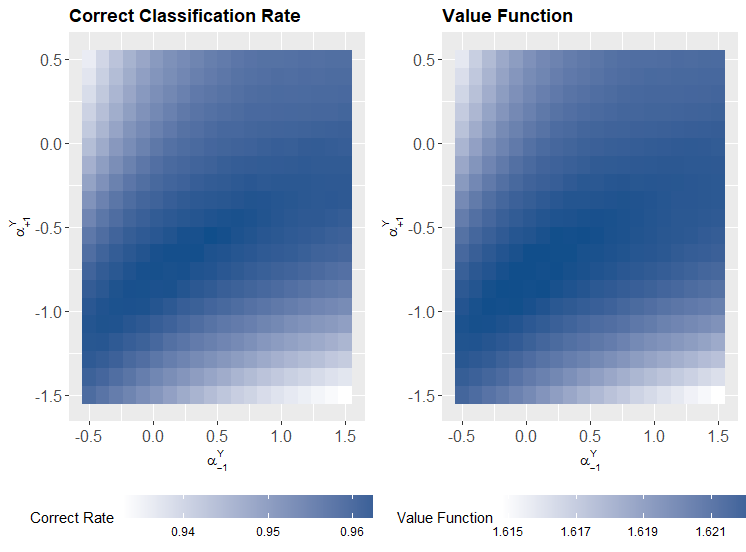}
    \caption{Sensitivity analysis of the multiply robust estimator when $f(Z|X)$ is mis-specified: Sensitivity parameters are ($\alpha^Y_{-1}, \alpha^Y_{+1}$). The darker area indicates a higher correct classification rate/value function. The true vector of sensitivity parameters are ($\alpha^Y_{-1}, \alpha^Y_{+1}$) = (0.5,-0.5). ($\alpha^0_{-1}, \alpha^0_{+1}$) is mis-specified and set at (0.5, 0.5). The true value of $(\alpha^0_{-1}, \alpha^0_{+1}) = (0,0)$. }
    \label{fig:alpha0mis-fzmis}
\end{figure}

\begin{figure}
    \centering
    \includegraphics[width = \textwidth]{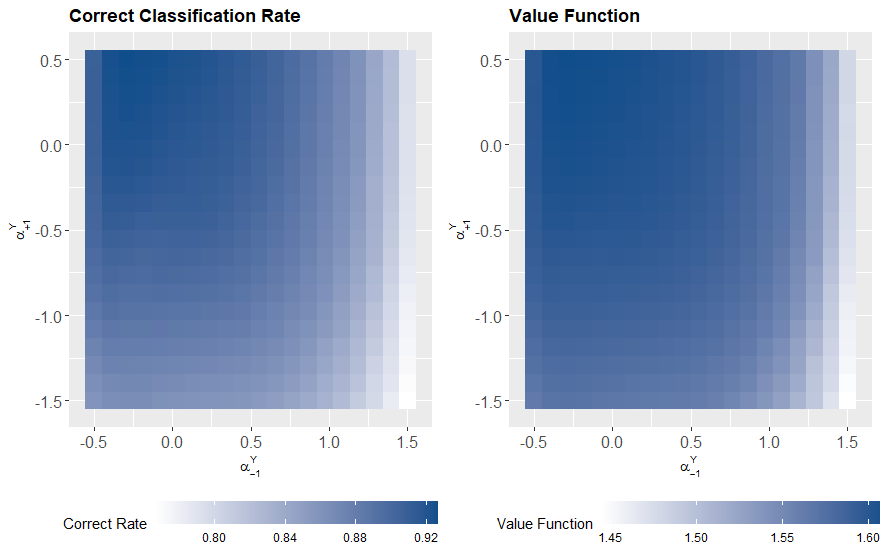}
    \caption{Sensitivity analysis of the multiply robust estimator when $Q(A,Z,X)$ is mis-specified: Sensitivity parameters are ($\alpha^Y_{-1}, \alpha^Y_{+1}$). The darker area indicates a higher correct classification rate/value function. The true vector of sensitivity parameters are ($\alpha^Y_{-1}, \alpha^Y_{+1}$) = (0.5,-0.5). ($\alpha^0_{-1}, \alpha^0_{+1}$) is mis-specified and set at (0.5, 0.5). The true $(\alpha^0_{-1}, \alpha^0_{+1}) = (0.0)$}
    \label{fig:alpha0mis-Qmis}
\end{figure}

\section{Sensitivity model misspecification}
\label{missenmod}

 We assess the sensitivity of our approach to the misspecification of the sensitivity model $w_{\alpha_Z}$. The true sensitivity function $\mathcal{G}$ is a function of the covariate X and the outcome Y, i.e $\mathcal{G}(X,Y, \alpha_Z) = \alpha^0_{Z}X + X\alpha^X_Z + Y\alpha^Y_{Z}$. We set $\alpha^0_{Z} =0.5$,$ \alpha^X_Z = (0.3, 0.3)$ for $Z \in \{-1,1\}$, $\alpha^Y_{-1} = 0.5, \alpha^Y_{+1} = -0.5 $. In the analysis, we let $\mathcal{G}(X,Y, \alpha_Z) = \alpha^Y_ZY$ or $\mathcal{G}(X,Y, \alpha_Z) = \alpha^{PCA_1}_Z\text{PCA}_1$, where $\text{PCA}_1$ indicates the first principal component analysis of Y and X. We show the result of our simulation for the MR method below. 

\begin{figure}
    \centering
    \includegraphics[width = \textwidth]{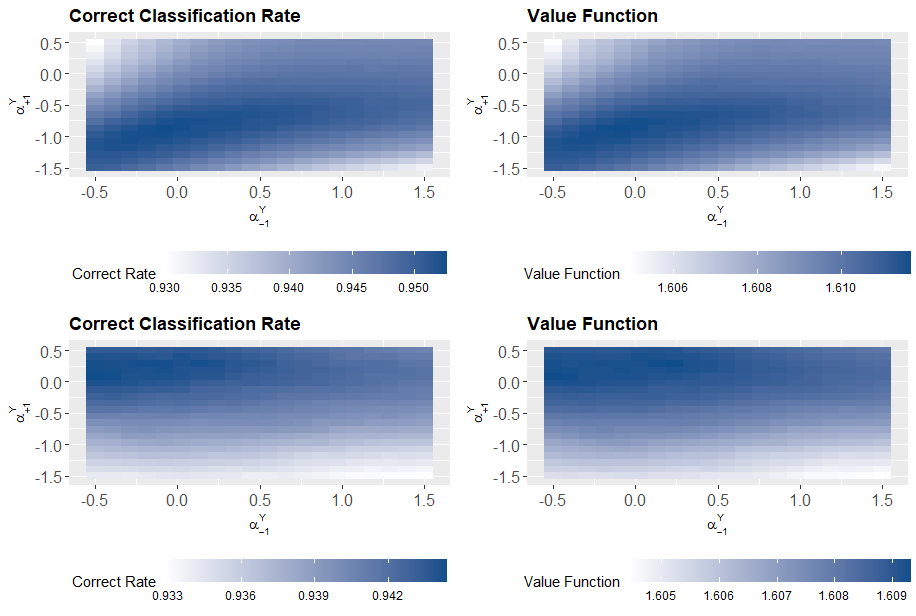}
    \caption{Sensitivity analysis of the multiply robust estimator when the sensitivity model is misspecified: Sensitivity parameters are ($\alpha^Y_{-1}, \alpha^Y_{+1}$) for the upper panel and ($\alpha^{PCA}_{-1}, \alpha^{PCA}_{+1}$). The upper panel shows the classification rate/value function when $\mathcal{G}(X,Y,\alpha_Z)$ is set a function of the outcome Y only. The lower panel shows the classification rate/value function when $\mathcal{G}(X,Y, \alpha_Z)$ is a function of the first principal component of the outcome Y and the covariates X.}
   % \label{fig:alpha0mis-Qmis}
\end{figure}

\section{Calculating the \texorpdfstring{$\alpha_Z^0$}{alpha\_Z\_0} in Section \ref{sec:application}}
\label{calalpha0}
We have 
\begin{align*}
    p(PS = S4| A = Z = z) &= \frac{p(PS = S4, A = z, Z = z)}{p(A = z, Z = z} \\
    &= \frac{p(PS = S4, A =z|Z = z)p(Z = z)}{p(A = z| Z=z)p(Z = z)} \\
    &= \frac{p(PS = S4|Z = z)}{p(A = z| Z=z)} \\ 
    &= \frac{p(PS = S4)}{p(A=z|Z=z)}
\end{align*}

The third equation is due to the definition of compliance and the last equation is because of randomization of the treatment $Z$ in ENGAGE study. We also have the following 

\begin{equation*}
%\label{eq:conds4}
    p(S4| A = Z = z) = \sum_{y \in \{-1,1\}} p(PS = S4| A = Z = z, y)p(Y = y| A = Z = z) 
\end{equation*}

Consider $p(PS = S4| A = Z = z, y) = \text{expit}(\alpha^0_z + y\alpha_z^Y)$ as our sensitivity function. Then, we have 
\begin{align*}
%\label{eq:conds42}
     p(&PS = S4| A = Z = z) = \sum_{y \in \{-1,1\}} \text{expit}(\alpha^0_z + y\alpha_z^Y) p(Y = y| A = Z = z) \\
      &= \frac{\exp(\alpha^0_z - \alpha_z^Y)}{1 + \exp(\alpha^{0}_z - \alpha^Y_z)}p(Y=-1\mid A=Z=z) \\
     &\hspace{.3in}+ \frac{\exp(\alpha^0_z + \alpha_z^Y)}{1 + \exp(\alpha^{0}_z + \alpha^Y_z)}p(Y=1\mid A=Z=z)\\
      &= \frac{p(Y = -1\mid A = Z =z)}{1 + \exp(-\alpha^0_z + \alpha^Y_z)} + \frac{p(Y = 1\mid A = Z =z)}{1 + \exp(-\alpha^0_z - \alpha^Y_z)}\\
      &= \frac{p(Y = -1 \mid A= Z = z) + p(Y = -1\mid A = Z = z)\exp(-\alpha^0_z)\exp(-\alpha^Y_z)}{1 + \exp(-\alpha^0_Z)\exp(-\alpha^Y_z)+ \exp(-\alpha^0_Z)\exp(\alpha^Y_z) + \exp(-2\alpha^0_z)} \\ 
     &\hspace{.3in} +\frac{p(Y = 1 \mid A= Z = z) + p(Y = 1\mid A =Z=z)\exp(-\alpha^0_z)\exp(\alpha^Y_z)}{1 + \exp(-\alpha^0_Z)\exp(-\alpha^Y_z)+ \exp(-\alpha^0_Z)\exp(\alpha^Y_z) + \exp(-2\alpha^0_z)}.
%     \Leftrightarrow& a\exp(-2\alpha^0_z) + b\exp(-\alpha^0_z) + c = 0\\
%     \Leftrightarrow&   \alpha^0_Z = -\log\frac{-b \pm \sqrt{b^2 - 4ac}}{2a}
\end{align*}
We can represent the above equality as $a\exp(-2\alpha^0_z) + b\exp(-\alpha^0_z) + c = 0$ which gives
\[
\alpha^0_Z = -\log\frac{-b \pm \sqrt{b^2 - 4ac}}{2a},
\]
where
\begin{align*}
    a &= p(PS = S4|A = Z =z) \\
    b &=a\exp(-\alpha^Y_z) + a\exp(\alpha^Y_z)\\&- p(Y = -1|A=Z=z)\exp(-\alpha^Y_z) - p(Y = 1|A=Z=z)\exp(\alpha^Y_z)\\
    c &= a -p(Y = 1|A=Z=z) - p(Y = -1|A=Z=z).
\end{align*}
We can estimate $p(Y = y|A = Z = z)$ from the ENGAGE data.  

In general, when $Y \in \mathcal{Y}$ is a continuous random variable, one can solve 
\[
p(PS = S4| A = Z = z) = \int_{\mathcal{Y}} \text{expit}(\alpha^0_z + y\alpha_z^Y) f(y| A = Z = z)dy,
\]
using a one-dimensional grid search on a real line.

\section{Additional Sensitivity Analysis for the ENGAGE dataset}
\label{sensiengage}

In the heatmaps below, we set the probability of compliers $P(PS = S4)$ to 0.25 and 0.35. In both of those cases, the MR method consistently outperforms the IVT and OWL method in the plausible range $(\alpha^Y_{-1} >0, \alpha^Y_{+1} > 0)$.

\begin{figure}
    \centering
    \includegraphics[width = \textwidth]{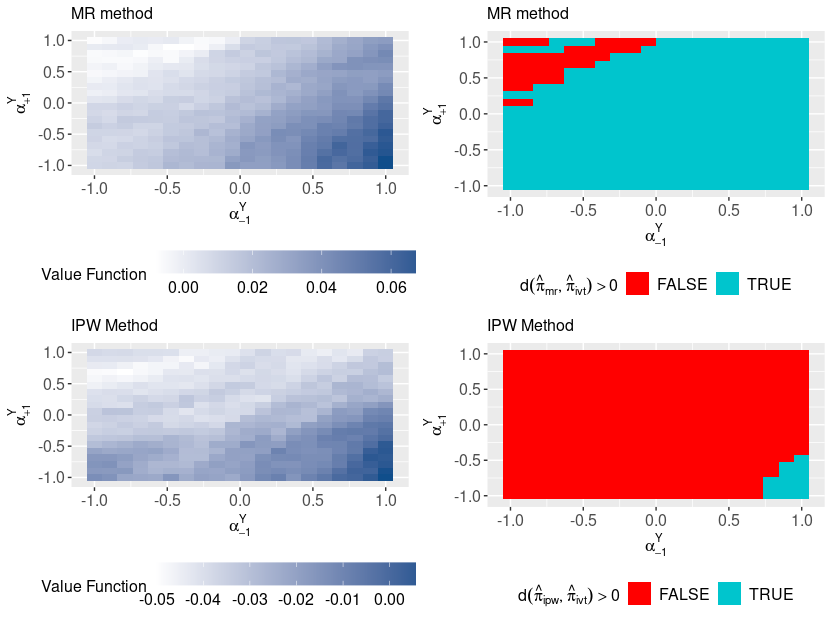}
    \caption{ENGAGE Study: Sensitivity analysis of the difference the proposed methods (MR and IPW) and the IVT method. The sensitivity parameters are $(\alpha^Y_{-1}, \alpha^Y_{+1})$. The probability of compliance is fixed at 0.35. The left panel shows the magnitude of the difference between the proposed  and IVT methods $d(\hat\pi_{mr},\hat\pi_{owl})$ and $d(\hat\pi_{ipw},\hat\pi_{owl})$ and the darker color the better performance of our methods). The right panel shows the area on the grid in which our methods outperform the IVT method where the red color indicates $d(\hat\pi_{mr},\hat\pi_{owl})<0$ or $d(\hat\pi_{ipw},\hat\pi_{owl})<0$.}
   % \label{fig:alpha0mis-Qmis}
\end{figure}

\begin{figure}
    \centering
    \includegraphics[width = \textwidth]{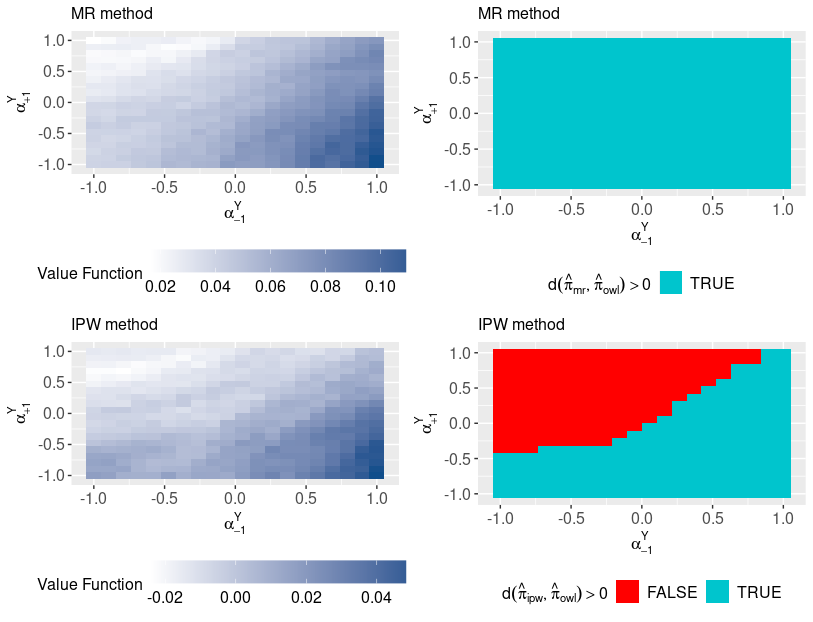}
    \caption{ENGAGE Study: Sensitivity analysis of the difference the proposed methods (MR and IPW) and the OWL method. The sensitivity parameters are $(\alpha^Y_{-1}, \alpha^Y_{+1})$. The probability of compliance is fixed at 0.35. The left panel shows the magnitude of the difference between the proposed and OWL methods $d(\hat\pi_{mr},\hat\pi_{owl})$ and $d(\hat\pi_{ipw},\hat\pi_{owl})$ and the darker color the better performance of our methods). The right panel shows the area on the grid in which our methods outperform the OWL method where the red color indicates $d(\hat\pi_{mr},\hat\pi_{owl})<0$ or $d(\hat\pi_{ipw},\hat\pi_{owl})<0$.}
   % \label{fig:alpha0mis-Qmis}
\end{figure}

\begin{figure}
    \centering
    \includegraphics[width = \textwidth]{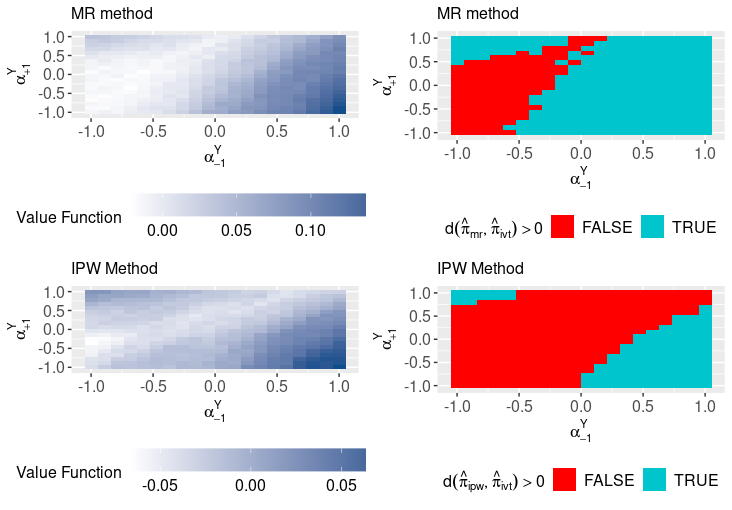}
    \caption{ENGAGE Study: Sensitivity analysis of the difference the proposed methods (MR and IPW) and the IVT method. The sensitivity parameters are $(\alpha^Y_{-1}, \alpha^Y_{+1})$. The probability of compliance is fixed at 0.25. The left panel shows the magnitude of the difference between the proposed and IVT methods ($d(\hat\pi_{mr},\hat\pi_{owl})$ and $d(\hat\pi_{ipw},\hat\pi_{owl})$) and the darker color the better performance of our methods). The right panel shows the area on the grid in which our methods outperform the IVT method where the red color indicates $d(\hat\pi_{mr},\hat\pi_{owl})<0$ or $d(\hat\pi_{ipw},\hat\pi_{owl})<0$.}
   % \label{fig:alpha0mis-Qmis}
\end{figure}

\begin{figure}
    \centering
    \includegraphics[width = \textwidth]{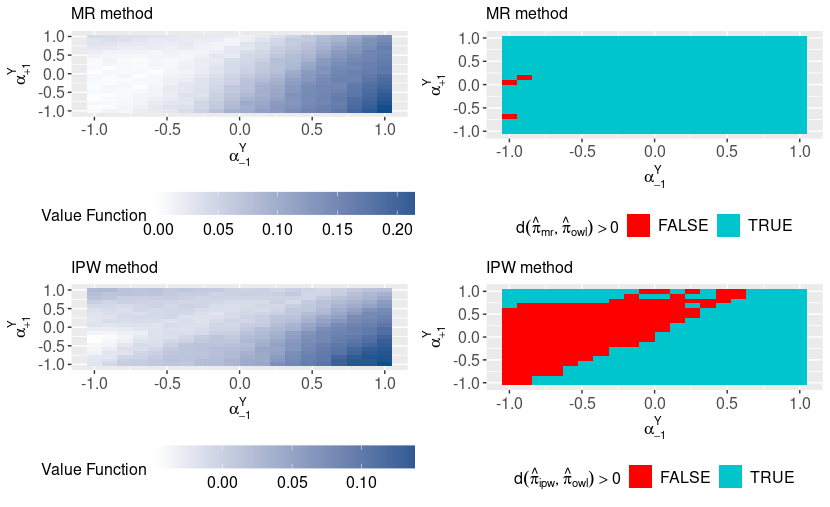}
    \caption{ENGAGE Study: Sensitivity analysis of the difference the proposed methods (MR and IPW) and the OWL method. The sensitivity parameters are $(\alpha^Y_{+1}, \alpha^Y_{-1})$. The probability of compliance is fixed at 0.25. The left panel shows the magnitude of the difference between the proposed and OWL methods $d(\hat\pi_{mr},\hat\pi_{owl})$ and $d(\hat\pi_{ipw},\hat\pi_{owl})$ and the darker color the better performance of our methods). The right panel shows the area on the grid in which our methods outperform the OWL method where the red color indicates $d(\hat\pi_{mr},\hat\pi_{owl})<0$ or $d(\hat\pi_{ipw},\hat\pi_{owl})<0$.}
   % \label{fig:alpha0mis-Qmis}
\end{figure}

%%  MR numbers will be added by VTeX.                      %%
%%                                                         %%
%%  Use \cite{...} to cite references in text.             %%
%%                                                         %%
%%%%%%%%%%%%%%%%%%%%%%%%%%%%%%%%%%%%%%%%%%%%%%%%%%%%%%%%%%%%%

      % Bibliography file (usually '*.bib')

%% or include bibliography directly:

\end{document}